\documentclass[ejsv2,noshowframe]{imsart}

\RequirePackage[authoryear]{natbib}%% uncomment this for author-year citations
\RequirePackage[colorlinks,citecolor=blue,urlcolor=blue]{hyperref}
\RequirePackage{graphicx}

\arxiv{2410.20319}
\startlocaldefs
\theoremstyle{plain}

\newtheorem{theorem}{Theorem}[section]
\newtheorem{proposition}{Proposition}[section]
\newtheorem{lemma}[theorem]{Lemma}
\theoremstyle{definition}

\theoremstyle{remark}

\usepackage{algorithm2e}
\usepackage{stackengine}
\usepackage{xcolor}

\usepackage{amssymb,amsmath,mathrsfs, bbm}
\usepackage{bm}
\usepackage{graphicx}
\usepackage{enumerate}
\usepackage{multirow}
\newtheorem{condition}{Condition}
\newcommand{\bfm}[1]{\ensuremath{\mathbf{#1}}}

\def\ba{\bfm a}   \def\bA{\bfm A}  
\def\bb{\bfm b}     
     
   \def\bD{\bfm D}

\def\bh{\bfm h}   \def\bH{\bfm H}  
   \def\bI{\bfm I}  
   \def\bJ{\bfm J}

   \def\bQ{\bfm Q}  
     \def\RR{\mathbb{R}}

   \def\bU{\bfm U}  
\def\bv{\bfm v}   \def\bV{\bfm V}  
     
\def\bx{\bfm x}   \def\bX{\bfm X}  
\def\by{\bfm y}     
\def\bz{\bfm z}   \def\bZ{\bfm Z}

\def\bzero{\bfm 0}

\newcommand{\bfsym}[1]{\ensuremath{\boldsymbol{#1}}}

\def\balpha{\bfsym \alpha}
\def\bbeta{\bfsym \beta}
       
       \def\bDelta{\bfsym \Delta}
\def\bfeta{\bfsym \eta}          
\def\bmu{\bfsym \mu}             

\def\btheta{\bfsym \theta}       
\def\beps{\bfsym \varepsilon}

\renewcommand{\hat}{\widehat}
\renewcommand{\tilde}{\widetilde}

\DeclareMathOperator*{\argmin}{arg\,min}

\endlocaldefs

\begin{document}
\begin{frontmatter}
\title{High-dimensional partial linear model with trend filtering}
%\title{A sample article title with some additional note\thanksref{t1}}
\runtitle{High-dimensional partial linear model with trend filtering}
%\thankstext{T1}{A sample additional note to the title.}

\begin{aug}
%%%%%%%%%%%%%%%%%%%%%%%%%%%%%%%%%%%%%%%%%%%%%%%
%% Only one address is permitted per author. %%
%% Only division, organization and e-mail is %%
%% included in the address.                  %%
%% Additional information can be included in %%
\author[A]{\fnms{Sang Kyu}~\snm{Lee}\ead[label=e1]{sangkyulee@konkuk.ac.kr}},
\author[B]{\fnms{Erikka}~\snm{Loftfield}\ead[label=e2]{erikka.loftfield@nih.gov}},
\author[C]{\fnms{Hyokyoung G.}~\snm{Hong}\ead[label=e3]{grace.hong@nih.gov}}
\and
\author[D]{\fnms{Haolei}~\snm{Weng}\ead[label=e4]{wenghaol@msu.edu}}\footnote{Corresponding author}
%%%%%%%%%%%%%%%%%%%%%%%%%%%%%%%%%%%%%%%%%%%%%%
%% Addresses                                %%
%%%%%%%%%%%%%%%%%%%%%%%%%%%%%%%%%%%%%%%%%%%%%%
\address[A]{Department of Applied Statistics, Konkuk University\printead[presep={,\ }]{e1}}

\address[B]{Metabolic Epidemiology Branch, Division of Cancer Epidemiology and Genetics, National Cancer Institute\printead[presep={,\ }]{e2}}

\address[C]{Biostatistics Branch, Division of Cancer Epidemiology and Genetics,
National Cancer Institute\printead[presep={,\ }]{e3}}

\address[D]{Department of Statistics and Probability,
Michigan State University\printead[presep={,\ }]{e4}}

\runauthor{S. K. Lee et al.}
\end{aug}

\begin{abstract}
Understanding the links between diet, metabolic changes, and health outcomes is a key focus in nutritional science and broader biological research. Analyzing relationships, such as those between ultra-processed food (UPF) intake and metabolites, offers insights into potential biomarkers for diet-related diseases and public health applications. However, these analyses are challenging due to high-dimensional data structures and complex, often nonlinear associations between covariates and health outcomes. Traditional linear models and conventional nonparametric methods often lack the flexibility to accurately capture such complexities in biological data. To address these challenges, we propose a high-dimensional partial linear regression model that captures both linear and nonlinear effects, combining the interpretability of linear models with the adaptability of nonparametric approaches. Our model leverages trend filtering to handle local smoothness variations effectively and achieves minimax optimal rates, making it suitable for complex biological datasets. We apply this model to data from the Interactive Diet and Activity Tracking in AARP (IDATA) Study, demonstrating its utility in identifying biomarkers associated with UPF intake and illustrating its potential for broader applications in dietary, metabolic, and health-related research.
\end{abstract}



\end{frontmatter}
%%%%%%%%%%%%%%%%%%%%%%%%%%%%%%%%%%%%%%%%%%%%%%
%% Please use \tableofcontents for articles %%
%% with 50 pages and more                   %%
%%%%%%%%%%%%%%%%%%%%%%%%%%%%%%%%%%%%%%%%%%%%%%
\tableofcontents

\sloppy
\section{Introduction}  
\label{sec:intro}

{
%Ultra-processed foods (UPFs) are increasingly recognized as significant contributors to poor health outcomes, including obesity, cardiovascular diseases, and metabolic disorders (ref). These foods, characterized by their high levels of added sugars, unhealthy fats, and synthetic additives, are pervasive in modern diets, especially in regions with widespread access to processed food products. Understanding the relationship between UPF consumption and biomarkers of health is crucial for shaping dietary recommendations. Such insights can guide public health initiatives aimed at reducing the consumption of UPFs and promoting healthier eating patterns.

%Analyzing the relationship between UPF intake and metabolites helps identify metabolic signatures associated with dietary patterns. These analyses are essential for distinguishing the direct effects of UPF consumption from broader dietary behaviors and for identifying potential biomarkers of disease risk. Understanding these associations enables researchers to pinpoint biochemical changes linked to UPF consumption.

Ultra-processed food (UPF), characterized by high levels of synthetic additives, unhealthy fats, and added sugars, is associated with negative health impacts, including obesity and metabolic disorders. Investigating these relationships helps identify biomarkers that reflect dietary effects and guide strategies for improving public health and nutrition.

%In nutritional science, investigating the links between diet and metabolic changes provides valuable insights into how different foods influence health outcomes. Understanding the relationship between ultra-processed food (UPF) intake and metabolites is crucial, as it can shed light on the mechanisms by which UPFs negatively impact health, serve as early biomarkers for diet-related diseases, inform public health recommendations and precision nutrition, and reveal associations with the gut microbiome that contribute to overall health outcomes. 

Studying this relationship is challenging due to the high dimensionality of metabolomic data. Accounting for key confounding variables, such as age \citep{fu2024association} and obesity-related measurements \citep{canhada2020ultra}, adds further complexity, as these factors may exhibit nonlinear associations with UPF intake. Similar challenges arise in other fields; for instance, lung cancer research faces similar challenges when associating mutation signatures with DNA methylation, where nonlinear relationships between factors like smoking history or tumor purity and mutation signatures add further complexity \citep{zhang2024apobec}.}

%for instance, in lung cancer research, analyses of associations between mutation signatures and DNA methylation often account for tumor purity. Variations in tumor purity can affect observed methylation levels by shifting the proportions of tumor and normal cells, which may also exhibit nonlinear relationships with methylation levels (ref).

%However, analyzing this relationship presents challenges due to the high-dimensional data structure of metabolites and the potential for nonlinear or complex relationships between UPF intake and other covariates, such as age \citep{fu2024association} and obesity-related measurements \citep{canhada2020ultra}. Such complexities are not unique to nutrition; for instance, lung cancer research faces similar challenges when associating mutation signatures with methylation, where nonlinear relationships with factors like tumor purity add further complexity.

Traditional linear models are often limited in capturing nonlinear associations, as they are designed to model only linear relationships. While nonparametric smoothing methods, such as splines and local polynomial smoothing, can accommodate certain nonlinear relationships, they may lack the flexibility necessary to capture more intricate, {highly heterogeneous associations between predictors and responses}, ultimately resulting in suboptimal predictive performance. %{\bf LASSO is primarily used for feature selection,  identifying the most important predictors by shrinking less relevant coefficients to zero. In contrast, smoothing methods focus on model fitting, which aims to capture patterns in the data by reducing noise. These two approaches serve fundamentally different purposes and are not directly comparable.}
These limitations underscore the need for an innovative modeling approach capable of flexibly capturing such complex associations.

%To address these complexities, we propose a high-dimensional partial linear regression model that captures both linear and nonlinear effects flexibly. Specifically, we assume the response variable $y$ is influenced by two types of predictors: a high-dimensional vector $\mathbf{x} \in \mathbb{R}^p$ with a linear effect on $y$, and a lower-dimensional vector $\mathbf{z} \in \mathbb{R}^d$ with a nonlinear effect. 

{In this paper, we propose a high-dimensional partial linear regression model in which the response variable $y$ is modeled as having a linear relationship with a high-dimensional vector of predictors $\mathbf{x} \in \mathbb{R}^p$ and a nonlinear relationship with \textcolor{black}{a one-dimensional predictor $\mathbf{z} \in \mathbb{R}$.} More specifically, we consider the following partial linear model:
\begin{align*}
%\label{hdplm}
y = \mathbf{x}' \boldsymbol{\beta}^0 + g_0(\mathbf{z}) + \varepsilon,
\end{align*}
where $\boldsymbol{\beta}^0$ is a sparse $p$-dimensional vector of coefficients, $g_0(\cdot)$ is an unknown nonparametric function, $\mathbf{z}$ is a univariate predictor chosen based on prior knowledge or structure selection methods \citep{zhang2011linear, huang2012semiparametric, lian2015separation}, and $\varepsilon$ is a zero-mean error term. 

Partial linear models have been widely used in various fields due to their flexibility in modeling complex relationships by combining the interpretability of linear models with the adaptability of nonparametric components.
However, their traditional application has largely been confined to settings where the dimension of  $\boldsymbol{\beta}^0$ is small or fixed relative to $n$; see, for instance, \cite{engle1986semiparametric, chen1988convergence, wahba1990spline, mammen1997penalized, hardle2000partially, bunea2004consistent, xie2009scad}. For high-dimensional $p$, partial linear models have also been extended \citep{muller2015partial, ma2016asymptotic, zhu2017nonasymptotic, yu2019minimax, zhu2019high, lv2022debiased, fu2024semiparametric}, but these extensions are limited by the assumption that $g_0$ lies within a smooth function class (e.g., Sobolev or Hölder). These methods commonly apply penalized regression approaches, such as LASSO \citep{tibshirani1996regression} or SCAD \citep{fan2001variable}, to estimate the sparse vector $\boldsymbol{\beta}^0$ and use nonparametric smoothing methods, such as splines or local polynomial smoothing, for $g_0$. However, the assumption that $g_0$ belongs to a smooth function class may impose overly restrictive conditions in practical settings, as $g_0$ may exhibit heterogeneous smoothness.

Similar to \cite{muller2015partial} and \cite{yu2019minimax}, our model accommodates cases where $p$ is comparable to or larger than the sample size $n$, but unlike their approach, we do not require $g_0$ to belong to a smooth function class.
We assume that $g_0$ belongs to the following function class:
\begin{align}
\label{tvf:def}
\mathcal{V}_k(C) := \Big\{ g_0 : [0, 1] \rightarrow \mathbb{R} : \text{TV}(g_0^{(k)}) \leq C \Big\},
\end{align}
where $\text{TV}(\cdot)$ denotes the total variation operator, $g_0^{(k)}$ represents the $k$-th weak derivative of $g_0$, and $C > 0$ is a constant.
The function class $\mathcal{V}_k(C)$ allows for heterogeneous smoothness of $g_0$, providing greater flexibility in capturing local behavior than traditional Sobolev or Hölder classes.

To estimate $g_0$, we utilize \textit{trend filtering} \citep{steidl2006splines, kim2009ell_1}, which extends the idea of \textit{locally adaptive regression splines} \citep{mammen1997locally}. Trend filtering estimates $g_0$ by minimizing an objective function regularized with a total variation penalty.
}
%In cases where $g_0$ varies locally, such classical methods may fail to capture these variations accurately. Suppose $g_0$ lies in the function class
%\begin{align}
%\label{tvf:def}
%\mathcal{V}_k(C) := \Big\{ g_0 : [0, 1] \rightarrow \mathbb{R} : \text{TV}(g_0^{(k)}) \leq C \Big\},
%\end{align}
%where $\text{TV}(\cdot)$ denotes the total variation operator, $g_0^{(k)}$ is the $k$-th weak derivative of 
%$g_0$, and $C > 0$ is a constant. The function class $\mathcal{V}_k(C)$ allows for heterogeneous smoothness of $g_0$, making it more adaptable to local variations than traditional Sobolev or Hölder classes. To address the limitations of classical smoothers, we adopt \textit{trend filtering} \citep{steidl2006splines, kim2009ell_1} as a locally adaptive method for estimating $g_0$. \\
%Trend filtering builds on the concept of \textit{locally adaptive regression splines} \citep{mammen1997locally}, 
%where the estimation of $g_0$ involves minimizing a total variation-regularized objective. 
In the univariate setting, the trend filtering estimator is given by
\begin{align}
\min_{g\in \mathcal{H}^k_n} \frac{1}{2} \sum_{i=1}^{n} \left( y_i - g(z_i) \right)^2 + {\color{black}\gamma}  \text{TV}(g^{(k)}), \label{eq:locally adaptive} 
\end{align}
where {\color{black}$\gamma \geq 0$} is a tuning parameter and $\mathcal{H}^k_n$ is an $n$-dimensional space spanned by spline-like basis functions known as falling factorial basis (see Section \ref{model:estimation} for details). This method adapts to local smoothness variations more effectively 
than traditional smoothers and achieves minimax optimal rates over $\mathcal{V}_k(C)$ \citep{donoho1998minimax, tibshirani2014adaptive}. Unlike linear smoothers, such as local polynomials or splines, which struggle to adapt to changes in smoothness, 
trend filtering provides a robust solution for estimating functions with varying smoothness levels.

Recent advancements in trend filtering have significantly expanded its range of applications, including univariate nonparametric regression under strong sparsity \citep{guntuboyina2020adaptive, ortelli2021prediction}, graph trend filtering \citep{wang2016trend, madrid2020adaptive}, functional trend filtering \citep{wakayama2023trend}, scalar-on-image regression models \citep{wang2017generalized}, additive models \citep{sadhanala2019additive, petersen2019data, tan2019doubly}, quantile regression models \citep{madrid2022risk}, and spatiotemporal models \citep{padilla2023temporal, rahardiantoro2024spatio}, among others and with further references therein. Note that while \cite{petersen2019data} emphasizes practical and algorithmic advances for an additive modeling approach that adaptively selects variables, linearity, and knot locations via trend filtering, our work instead provides a comprehensive theoretical framework under a partial linear model. In particular, we assume a known nonparametric component to accommodate highly heterogeneous smoothness, offering a distinct perspective compared to purely additive methods.

%{\color{orange}To the best of our knowledge, our method addresses an important gap in the literature on additive trend filtering. Specifically, existing approaches either cannot be effectively adapted to high-dimensional settings in practice \citep{sadhanala2019additive} or lack rigorous theoretical justifications, such as established convergence rates \citep{petersen2019data}. Our work resolves these limitations by developing additive trend filtering within a high-dimensional partial linear model framework, supported by comprehensive theoretical results. As far as we know, this is the first work to consider the application of trend filtering within the high-dimensional partial linear model framework.}

%Unlike fully linear or nonparametric models, the partial linear model comprises two distinct components, presenting unique challenges for trend filtering, a relatively new approach. This duality raises both theoretical and practical complexities. Theoretically, we derive the rate of convergence, proving it to be minimax optimal under relaxed conditions compared to prior work \citep{muller2015partial, yu2019minimax}. For implementation, we develop an efficient blockwise coordinate descent algorithm, made publicly available as an R package.

In this paper, we extend trend filtering to the high-dimensional partial linear model setting. 
%The partial linear model stands out by combining two distinct components—linear and nonparametric—making its estimation uniquely challenging. Simultaneously estimating these components requires a careful balance to avoid bias and overfitting, adding complexity to both theoretical analysis and practical implementation, particularly in the context of trend filtering.
The partial linear model integrates parametric and nonparametric components, presenting unique estimation challenges. Jointly estimating these components requires carefully balancing bias and variance to avoid overfitting. These complexities are especially significant in the context of trend filtering, affecting both theoretical analysis and practical implementation.
%{\color{blue} Unlike fully linear or purely nonparametric models, the partial linear model presents two distinct components, making it a unique challenge. Estimating the linear and nonparametric components simultaneously requires careful balancing to mitigate bias and overfitting. Furthermore, the nonparametric component $g_0$ may exhibit heterogeneous smoothness, complicating the selection of appropriate regularization. These factors significantly increase the complexity of both the theoretical analysis and practical implementation of trend filtering in partial linear models.}
%Unlike fully linear or purely nonparametric models, the partial linear model presents two distinct components, making it challenging to directly apply trend filtering—a relatively new approach in statistical modeling. This dual structure introduces significant theoretical and practical challenges in both statistical theory and implementation. 
From a theoretical standpoint, we prove that our estimate for $\bbeta^0$ attains the oracle rate $(s\log p)/n$ as if $g_0$ were known, and the convergence rate for $g_0$ exhibits a phase transition between $(s\log p)/n$ and the optimal nonparametric rate $n^{-(2k+2)/(2k+3)}$. This dual-rate result highlights the unique adaptive benefits of partial linear trend filtering, distinguishing it from existing partial linear estimators that do not incorporate trend filtering or from trend filtering methods that have not established such optimal rate results. Notably, the conditions we impose are more relaxed than, or at least comparable to, those typically required in existing partial linear model frameworks \citep{muller2015partial, yu2019minimax}. %On the implementation side, we develop a blockwise coordinate descent algorithm that ensures efficient computation, and we make this method accessible by providing an R package for public use.
{For implementation, we develop a blockwise coordinate descent algorithm to enable efficient computation and have made the method accessible through an accompanying {R} package.}

%{\color{blue} We apply the proposed model to investigate the association between UPF intake and high-dimensional metabolites, while accounting for age, BMI, hip circumference, and waist circumference. Followed by previous research, these variables may also exhibit nonlinear relationships that do not necessarily follow smooth functional forms. The dataset originates from the Interactive Diet and Activity Tracking in AARP (IDATA) Study, conducted by the National Cancer Institute (NCI). This study was specifically designed to support research on dietary intake, nutrition, and cancer prevention (ref).}

%These covariates are relevant to include as they may influence both dietary patterns and metabolic markers. Age, for example, is often linked to variations in dietary intake and overall health status, and BMI, hip and waist circumferences are commonly used measures of body composition. 

We apply the proposed model to examine the association between UPF intake and high-dimensional metabolite profiles, accounting for potential nonlinear and non-smooth effects of age, BMI, hip circumference, and waist circumference. The dataset originates from the Interactive Diet and Activity Tracking in AARP (IDATA) Study, conducted by the National Cancer Institute (NCI). This study was specifically designed to support research on dietary intake, nutrition, and cancer prevention \citep{subar2020performance}.

%We apply this model to real-world data from the Interactive Diet and Activity Tracking in AARP (IDATA) Study, conducted by the National Cancer Institute (NCI), which is a comprehensive data resource designed to support research on dietary intake, nutrition, and cancer prevention. IDATA integrates data from multiple large-scale dietary assessment studies, linking detailed dietary, biomarker, and health outcome data across various population groups. As mentioned earlier, identifying biomarkers for UPF intake is essential, as UPFs are associated with chronic disease risk.
Building on this dataset, our method utilizes high-dimensional feature selection within a partial linear model framework, incorporating trend filtering to enhance the detection of biomarkers associated with UPF intake. By improving prediction accuracy, this approach not only identifies meaningful metabolites but also contributes to the development of evidence-based public health guidelines and dietary recommendations.
%{Our method integrates high-dimensional feature selection within a partial linear model framework, leveraging trend filtering to enhance the identification of biomarkers associated with UPF intake. This approach may improve prediction accuracy and contribute to public health guidelines and dietary recommendations.}

%Our model, by leveraging trend filtering, provides a flexible and superior approach for identifying biomarkers linked to UPF intake, which could inform public health guidelines and dietary recommendations.

{
The remainder of the paper is organized as follows. Section \ref{main:part:paper}
 introduces the proposed method and examines its statistical properties. In Section \ref{sec:simulation}, we present a comprehensive numerical simulation study, and in Section \ref{real:data}, we analyze the relationship between UPF intake and metabolite profiles, accounting for nonlinear confounding effects. Section \ref{discuss:con} concludes with remarks on potential extensions.}
%The rest of the paper is organized as follows. Section \ref{main:part:paper} details the problem setting, our proposed method, and associated theory. Sections \ref{sec:simulation} and \ref{real:data} present a comprehensive simulation study and a real data analysis, respectively. Section \ref{discuss:con} concludes with remarks on potential extensions. All proofs are in the supplementary material.

\subsection{Notations}
{We begin by introducing the notation that will be used throughout the paper.} For $\ba=(a_1,\ldots,a_p)' \in \mathbb{R}^p$, denote $ \|\ba\|_q = (\sum_{i=1}^p |a_i|^q)^{\frac{1}{q}}$ for $q\in [1,\infty)$ and $\|\ba\|_{\infty} = \max_{1\leq i\leq p} |a_i|$. For two vectors $\ba,\bb\in \mathbb{R}^n$, we write $\|\ba\|^2_n=\frac{1}{n}\ba'\ba, \langle \ba,\bb \rangle_n=\frac{1}{n}\ba'\bb$. For a vector $\ba\in \mathbb{R}^n$ and a function $g: \mathbb{R}\rightarrow \mathbb{R}$, let $g(\ba)=(g(a_1),\ldots, g(a_n))'$. Given a square matrix $\bA=(a_{ij}) \in \mathbb{R}^{p\times p}$, $\lambda_{\max}(\bA)$ and $\lambda_{\min}(\bA)$ represent its largest and smallest eigenvalues respectively. For a general matrix $\bA=(a_{ij}) \in \RR^{p\times q}$, $\|\bA\|_2$ denotes its spectral norm; $\|\bA\|_{\max} = \max_{ij}|a_{ij}|, \|\bA\|_F=\sqrt{\sum_{i,j} a_{ij}^2}$. For $a,b\in\RR$, $a\wedge b = \min(a,b), a\vee b = \max(a,b)$. For a set $A$, $ \mathbbm{1}_A(\cdot)$ is the usual indicator function, and $|A|$ to be its cardinality. Moreover, $a_n\lesssim b_n~ (a_n\gtrsim b_n)$ means there exists some constant $C>0$ such that $a_n\leq Cb_n~(a_n\geq C b_n)$ for all $n$; thus $a_n\lesssim b_n~ (a_n\gtrsim b_n)$ is equivalent to $a_n=O(b_n)~(a_n=\Omega(b_n))$; $a_n\asymp b_n$ if and only if $a_n\lesssim b_n$ and $b_n\gtrsim a_n$; $a_n\gg b_n$ means $b_n=o(a_n)$. We put subscript $p$ on $O$ and $o$ for random variables. For i.i.d. samples $\{w_1,\ldots, w_n\}$ from a distribution $Q$ supported on some space $\mathcal{W}$, denote by $Q_n$ the associated empirical distribution. The $L_2(Q)$ and $L_2(Q_n)$ norms for functions $f: \mathcal{W}\rightarrow \mathbb{R}$ are: $\|f\|^2_{L_2(Q)}=\int_{\mathcal{W}}f^2(w)dQ(w), \|f\|^2_{L_2(Q_n)}=\frac{1}{n}\sum_{i=1}^nf^2(w_i)$. For simplicity we will abbreviate subscripts and write $\|f\|, \|f\|_n$ for $\|f\|_{L_2(Q)}, \|f\|_{L_2(Q_n)}$ respectively, whenever $Q$ is the underlying distribution of the covariates. For a random variable $x\in \mathbb{R}$, we also write $\|x\|$ for $\sqrt{\mathbb{E}x^2}$. Given random variables $z_1,z_2,\ldots, z_n$, the order statistics are denoted by $z_{(1)}\leq z_{(2)}\leq \cdots \leq z_{(n)}$. The sub-Gaussian norm of a random variable $x \in \mathbb{R}$ is defined as $\|x\|_{\psi_2}=\inf\{t>0: \mathbb{E}\exp(x^2/t^2)\leq 2\}$.

\section{Trend filtering in high-dimensional partial linear models}
\label{main:part:paper}

\subsection{Problem setting and the proposed method}
\label{model:estimation}

We consider the partial linear regression model:
\begin{align*}
y=\bx'\bbeta^0+g_0(z)+\varepsilon,
\end{align*}
where $\varepsilon$ is independent of $(\bx,z) \in \mathbb{R}^{p+1}$, $\bbeta^0\in \mathbb{R}^p$ has the support $S=\{j: \beta^0_j \neq 0\}$ with $|S|=s$, and $g_0: [0,1] \rightarrow \mathbb{R}$ is a nonparametric function. Without loss of generality, the support of $z$ is assumed as $[0,1]$. It can be relaxed to any compact interval. Let $\{(y_i,\bx_i, z_i)\}_{i=1}^{n}$ be $n$ independent observations of $(y,\bx,z)$, and denote $\by=(y_1,\ldots, y_n)', \bX=(\bx_1,\ldots, \bx_n)', \bz=(z_1,\ldots, z_n)'$. We focus on the high-dimensional setting in which the dimension $p$ can be much larger than the sample size $n$, and assume ${\rm TV}(g_0^{(k)})\leq L_g$ with some constant $L_g > 0$ to allow for a large degree of heterogeneous smoothness of $g_0$.

Expanding upon the univariate trend filtering \citep{tibshirani2014adaptive} discussed in Section \ref{sec:intro}, we consider the following $k$th degree partial linear trend filtering estimation,
\begin{align}
(\hat{\bbeta},\hat{g})\in \argmin_{\bbeta\in \mathbb{R}^p, g\in \mathcal{H}^k_n} \frac{1}{2}\|\by-\bX\bbeta-g(\bz)\|_n^2+\lambda\|\bbeta\|_1+\gamma {\rm TV}(g^{(k)}), \label{est:p}
\end{align}
where $\lambda,\gamma \geq 0$ are tuning parameters, and $\mathcal{H}^k_n$ is the span of the $k$th degree falling factorial basis functions defined over the ordered input points $z_{(1)}<z_{(2)}<\cdots<z_{(n)}$. The set of basis functions take the form \citep{tibshirani2014adaptive, wang2014falling, tibshirani2022divided},
\begin{align}
    q_{i}(t) &= \prod_{l=1}^{i-1}(t - z_{(l)}),~~i=1,\dots,k+1, \label{ffb:one}\\
    q_{i+k+1}(t) &= \prod_{l=1}^{k}(t - z_{(i+l)})  \mathbbm{1}\{ 
t > z_{(i+k)}\},~~i=1,\dots,n-k-1,  \label{ffb:two}
\end{align}
where we adopt the convention $\prod_{i=1}^0c_i=1$. The above falling factorial basis looks similar to the standard truncated power basis for $k$th degree splines with knots at $z_{(k+1)},\ldots, z_{(n-1)}$. In fact, it is straightforward to verify that the two bases are equal when $k=0,1$, and they span different spaces when $k\geq 2$--the falling factorial functions in \eqref{ffb:two} are piecewise polynomials with discontinuities in their derivatives of orders $1,2,\ldots, k-1$. Define the matrix $\bQ\in \mathbb{R}^{n\times n}$ with entries $q_{\ell k}=q_{\ell}(z_k), 1\leq \ell, k\leq n$. Then for any $g\in \mathcal{H}^k_n$, we can write $g(\bz)=\bQ\balpha$ for some $\balpha \in \mathbb{R}^n$. The estimation in \eqref{est:p} is thus equivalent to 
\begin{align}
(\hat{\bbeta},\hat{\balpha})\in \argmin_{\bbeta\in \mathbb{R}^p, \balpha \in \mathbb{R}^n} \frac{1}{2}\|\by-\bX\bbeta-\bQ\balpha\|_n^2+\lambda\|\bbeta\|_1+\gamma k!\sum_{l=k+2}^{n} |\alpha_{l}|. \label{est:falling}
\end{align}
Further representing $\bQ\balpha=\btheta$ and using the formula for $\bQ^{-1}$ \citep{wang2014falling}, we can reformulate the optimization \eqref{est:falling} as
\begin{align}
(\hat{\bbeta},\hat{\btheta})\in \argmin_{\bbeta\in \mathbb{R}^p, \btheta\in \mathbb{R}^n} \frac{1}{2}\|\by-\bX\bbeta-\btheta\|_n^2+\lambda\|\bbeta\|_1+\gamma \| D^{ (\bz, k+1)} \btheta \|_1, \label{est:d}
\end{align}
where $D^{(\bz, k+1)} \in \mathbb{R}^{(n-k-1)\times n}$ is the discrete difference operator of order $k+1$. {\color{black} Formulation \eqref{est:d} shares similarity with the method in  \cite{drikvandi2025high}, where they apply distinct penalties to two separate parameter blocks by treating some covariates as being of interest and others as nuisance. The key difference between their approach and our approach is: (1) Their approach is focused on high-dimensional estimation and inference under linear models, while our approach centers on the estimation under high-dimensional partial linear models; (2) Their approach uses smooth penalties to shrink parameters of interest and control variance, while ours employs a non-smooth penalty to estimate the nonlinear function and achieve local adaptivity (see next paragraph for more details).} When $k=0$,
\begin{align}
    D^{(\bz, 1)} = \begin{pmatrix}
        -1 & 1 & 0 & \cdots & 0 & 0  \\
        0 & -1 & 1 & \cdots & 0 & 0 \\
        \vdots &&&&& \\
        0 & 0 & 0 & \cdots & -1 & 1
    \end{pmatrix} \in \mathbb{R}^{(n-1) \times n}. \label{mat:difference operator}
\end{align}
For $k \geq 1$, the difference operator is defined recursively, that is,
\begin{align*}
    D^{(\bz,k+1)} = D^{(\bz,1)} \cdot \text{diag} \left( \frac{k}{z_{(k+1)} - z_{(1)}}, \cdots, \frac{k}{z_{(n)} - z_{(n-k)}}  \right) \cdot D^{(\bz,k)}.
\end{align*}
Here, $D^{(\bz, 1)}$ is defined as the form of \eqref{mat:difference operator} with the dimension of $(n-k-1) \times (n-k)$. The problem \eqref{est:d} is a generalized LASSO problem. {The sparsity and banded structure of the penalty matrix $D^{(\bz,k+1)}$ provides significant advantages in solving the optimization problem
\eqref{est:d}. %The computational details are deferred to Section \ref{com:detail}.} 
%The fact that the penalty matrix $D^{(\bz,k+1)}$ is sparse and banded is of great advantage for solving the optimization \eqref{est:d}. We defer the computational details to Section \ref{com:detail}. 
Once $\hat{\btheta}$ is computed from \eqref{est:d}, the estimator $\hat{g}$ in \eqref{est:p} can be obtained as $\hat{g}(t)=\sum_{\ell=1}^n\hat{\alpha}_{\ell}q_{\ell}(t)$ with $\hat{\balpha}=\bQ^{-1}\hat{\btheta}$.

%{\color{blue} Our estimation approach in \eqref{est:p} employs a doubly penalized least squares estimator, combining a shrinkage penalty to induce sparsity in the parametric component with a smoothness penalty to regulate the complexity of the nonparametric component, similar to methods proposed by \cite{muller2015partial} and \cite{yu2019minimax}.}
{Our estimation approach in \eqref{est:p} introduces a doubly penalized least squares estimator, similar to those proposed by \cite{muller2015partial} and \cite{yu2019minimax}.} This approach involves two penalties: the first shrinkage penalty induces sparsity on the parametric part, and the second smoothness penalty controls the complexity of the nonparametric part. The primary distinction between our approach and theirs lies in the estimation method of the function $g_0$. While both of their methods employ the smoothing spline technique with a squared $\ell_2$ penalty, our method utilizes trend filtering based on an $\ell_1$ type penalty which can achieve a finer degree of local adaptivity.  Given the structure of the partial linear model, a better estimation of $g_0$ is expected to lead to a better estimation of $\bbeta^0$. Therefore, our method improves over the methods in \cite{muller2015partial} and \cite{yu2019minimax} when $g_0$ possesses heterogeneous smoothness. The results on degrees of freedom provide a clear rationale for our method's superior ability to adapt to heterogeneous smoothness. %We will elaborate on this comparison both theoretically and empirically in the following sections.

\subsection{Degrees of freedom}

We assume $\bX$ and $\bz$ are fixed with $z_{(1)}<z_{(2)}<\cdots<z_{(n)}$ to examine 
%{\bf perhaps, analyze, investigate, examine?}  
the degrees of freedom for the proposed partial linear trend filtering method \eqref{est:d}. Recall that for given data $\by\in \mathbb{R}^n$ with $\mathbb{E}(\by)=\bmu, {\rm Cov}(\by)=\sigma^2 \bI$, the effective degrees of freedom \citep{stein1981estimation, hastie1990generalized} of $\hat{\bmu}$, as an estimator of $\bmu$, is defined as
\begin{align}
\label{df:orig}
    {\rm df}\left(\hat{\bmu}\right) &= \frac{1}{\sigma^2} \sum_{i=1}^{n} \text{Cov} (\hat{\mu}_i, y_i). 
    %\\
    %&\approx \frac{1}{\sigma^2 Q} \sum_{i=1}^{n} \sum_{m=1}^{M} \hat{\mu}_{im} \epsilon_{im}, \text{ where } \hat{\mu}_{i} = \bx_i ' \hat{\bbeta} + \hat{g}(z_i)
\end{align}
The degrees of freedom measures the complexity of an estimator and plays an important role  in model assessment and selection. Since \eqref{est:d} follows the generalized LASSO form, we can apply established results on the generalized LASSO \citep{tibshirani2011solution, tibshirani2012degrees} to derive the degrees of freedom for \eqref{est:d}.  

\begin{proposition} \label{prop:df}
Consider $( \hat{\bbeta}, \hat{\btheta})$ from \eqref{est:d}. Define the two active sets 
\[
\mathcal{A}=\{1\leq j\leq p: \hat{\beta}_j\neq 0\}, \quad \mathcal{B}=\{1\leq j\leq n-k-1: (D^{(\bz, k+1)}\hat{\btheta})_j\neq 0\}.
\] 
Assume the Gaussian partial linear model $\by \sim \mathcal{N}\left(\bX\bbeta^0 + g_0(\bz), \sigma^2 \bI\right)$.
\begin{itemize}
\item[(i)] For any fixed $\bX, \bz$ and $\lambda \geq 0, \gamma \geq 0$,  the degrees of freedom for the fitting $\bX \hat{\bbeta} + \hat{\btheta}$ is
\begin{align*}
    {\rm df}\left(\bX \hat{\bbeta} + \hat{\btheta}\right) =\mathbb{E}\Big[{\rm dim}\big({\rm Col}(\bX_{\mathcal{A}})+{\rm Null}(D^{(\bz, k+1)}_{-\mathcal{B}})\big)\Big],
\end{align*}
where ${\rm Col}(\bX_{\mathcal{A}})$ is the subspace spanned by the columns of $\bX$ that are indexed by $\mathcal{A}$, and ${\rm Null}(D^{(\bz, k+1)}_{-\mathcal{B}})$ is the nullspace of the matrix $D^{(\bz, k+1)}$ after removing the rows indexed by $\mathcal{B}$. 
\item[(ii)] In addition, denote the first $k+1$ columns and last $n-k-1$ columns of $\bQ$ in \eqref{est:falling} by  $\bQ_1\in \mathbb{R}^{n\times (k+1)},\bQ_2\in \mathbb{R}^{n\times (n-k-1)}$ respectively. Let $\bU\bU'$ be the projection operator onto the space orthogonal to ${\rm Col}(\bQ_1)$ where $\bU \in \mathbb{R}^{n\times (n-k-1)}$ has orthonormal columns. If the matrix $[\bU'\bX, \frac{\lambda}{\gamma k!}\bU'\bQ_2]$ has columns in general position \citep{tibshirani2013lasso}, then
\begin{align*}
    {\rm df}\left(\bX \hat{\bbeta} + \hat{\btheta}\right) &=\mathbb{E}\Big[{\rm dim}\big({\rm Col}(\bX_{\mathcal{A}}))+{\rm dim}({\rm Null}(D^{(\bz, k+1)}_{-\mathcal{B}})\big)\Big] \\
    &=\mathbb{E}\big[|\mathcal{A}|+|\mathcal{B}|\big]+k+1.
\end{align*}
\end{itemize}
\end{proposition}

It is known that $\mathbb{E}\big[{\rm dim}({\rm Col}(\bX_{\mathcal{A}}))\big]$ is the degrees of freedom for $\bX\hat{\bbeta}$ when $\hat{\bbeta}$ is a standard LASSO estimate \citep{tibshirani2012degrees}, and $\mathbb{E}\big[{\rm dim}({\rm Null}(D^{(\bz, k+1)}_{-\mathcal{B}}))\big]$ is the degrees of freedom for $\hat{\btheta}$ if $\hat{\btheta}$ is from univariate trend filtering \citep{tibshirani2014adaptive}. Part (i) of Proposition \ref{prop:df} shows that the degrees of freedom for the partial linear trend filtering \eqref{est:d}--based on the idea of combining LASSO and trend filtering--equals to the expected dimension of the sum of ${\rm Col}(\bX_{\mathcal{A}})$ and ${\rm Null}(D^{(\bz, k+1)}_{-\mathcal{B}})$. When these two subspaces have no intersection except $\bzero$, the sum becomes direct sum so that 
\[
{\rm dim}\big({\rm Col}(\bX_{\mathcal{A}})+{\rm Null}(D^{(\bz, k+1)}_{-\mathcal{B}})\big)={\rm dim}\big({\rm Col}(\bX_{\mathcal{A}}))+{\rm dim}({\rm Null}(D^{(\bz, k+1)}_{-\mathcal{B}})\big).
\]
Part (ii) of Proposition \ref{prop:df} provides a sufficient condition for the above to hold. The result in Part (ii) admits a more direct interpretation: for an unbiased estimate of the degrees of freedom of $\bX \hat{\bbeta} + \hat{\btheta}$, we count the number of nonzeros in $ \hat{\bbeta}$ and the number of changes in the $(k+1)$th discrete derivative of $\hat{\btheta}$, and add them up together with $k+1$. The proof for Proposition \ref{prop:df} is provided in Appendix \ref{proof:prop}.

%The degrees of freedom measures the complexity of an estimator, and plays an important role in model assessment and selection. Since \eqref{est:d} takes the generalized lasso form, we can apply known results on the generalized lasso \citep{tibshirani2011solution, tibshirani2012degrees} to derive the degrees of freedom for \eqref{est:d}. Using this property, the alternative form of degrees of freedom is presented in Proposition 1 which is provided in the supplementary material.

\begin{figure}[!t]
    \centering
    \includegraphics[width=0.75\textwidth, height=0.57\textheight]{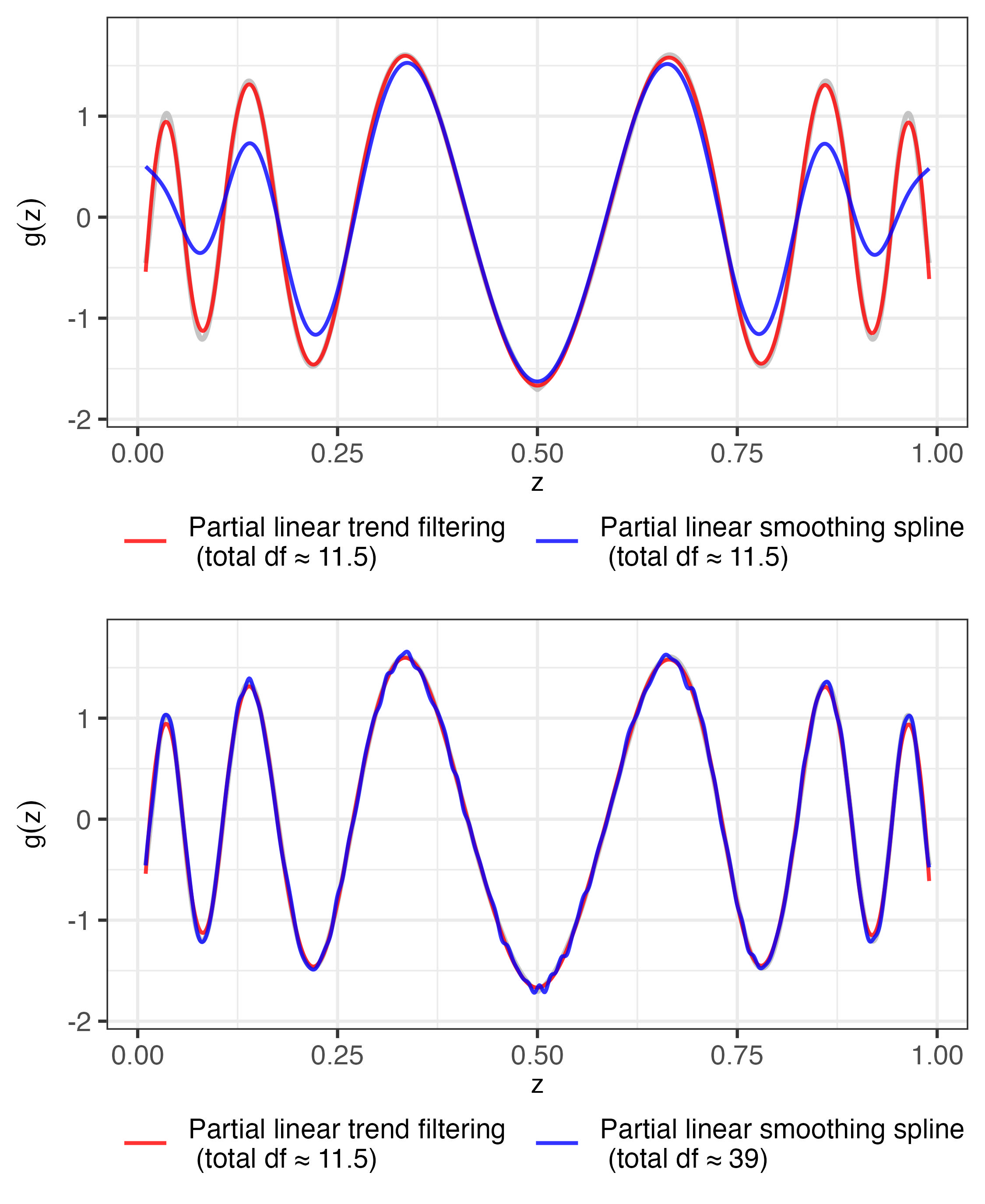}
    \caption{{Comparison of $\hat{g}(z)$ estimates using partial linear trend filtering (PLTF) and partial linear smoothing splines (PLSS) for two different total degrees of freedom.} The total degree of freedom is approximately 11.5 for both PLTF and PLSS for the upper plot, and 11.5 and 39 for the lower plot correspondingly. The grey line denotes the true function.}
    \label{fig:df}
\end{figure} 

The degrees of freedom allows us to calibrate model complexities for a fair comparison of different methods. Consider the method from \cite{muller2015partial, yu2019minimax} based on smoothing splines,
\begin{align}
\label{plss:def}
(\hat{\bbeta},\hat{g})\in \argmin_{\bbeta\in \mathbb{R}^p, g\in \mathcal{G}^k_n} \frac{1}{2}\|\by-\bX\bbeta-g(\bz)\|_n^2+\lambda\|\bbeta\|_1+\gamma \int_0^1(g^{((k+1)/2)}(t))^2dt,
\end{align}
where $\mathcal{G}^k_n$ is the space of $k$th degree natural splines with knots at the input points $z_1,\ldots, z_n$. We refer to it as \emph{partial linear smoothing spline}. Using $n$ basis functions spanning $\mathcal{G}^k_n$, the partial linear smoothing spline estimation can be rewritten in a form which is similar to \eqref{est:falling} or \eqref{est:d} except that the $\ell_1$ penalty for the nonparametric part is replaced with a squared $\ell_2$ penalty. We omit the detail as this is standard in the nonparametrics literature. Figure \ref{fig:df} presents a comparison between the partial linear smoothing spline \eqref{plss:def} and partial linear trend filtering \eqref{est:p} across various degrees of freedom (by changing the tuning parameters $(\lambda, \gamma)$). Note that the degrees of freedom for partial linear smoothing spline does not admit a simple form. We hence use the original definition \eqref{df:orig} to numerically compute it. In this comparison, we consider the model $g_0(z) = 2\min(z, 1-z)^{0.2} \sin\left\{2.85 \pi / (0.3 + \min(z, 1-z) )\right\}$ with setting $n=1000, s = 4$ and $p=100$.  Figure \ref{fig:df} illustrates that, when both methods are applied with the same total degrees of freedom of 11.5, the partial linear smoothing spline fails to adequately capture local smoothness at the function's boundaries, whereas the partial linear trend filtering method performs more effectively in this regard. Although increasing the degrees of freedom to 39 improves the fit of the partial linear smoothing spline at the boundaries, it leads to oversmoothing in the central regions of the function. {These observations highlight that partial linear trend filtering better adapts to varying levels of local smoothness compared to partial linear smoothing splines, providing a key motivation for our study.} %These observations illustrate that partial linear trend filtering adapts to the local level of smoothness better than partial linear smoothing spline. 
We provide more empirical comparisons in Sections \ref{sec:simulation} and \ref{real:data}.

\subsection{Theoretical Properties}

In this section, we study the rate of convergence for the proposed estimators $(\hat{\bbeta},\hat{g})$ in \eqref{est:p}. We first introduce our technical conditions.

\begin{condition} \label{con:one}
The covariate $\bx=(x_1,\ldots, x_p)$ has sub-Gaussian coordinates: $\max_{1\leq j \leq p}\|x_j\|_{\psi_2}\leq K_x$. 
\end{condition}

\begin{condition} \label{con:two}
The noise $\varepsilon$ is sub-Gaussian: $\mathbb{E}(\varepsilon)=0, {\rm Var}(\varepsilon)=\sigma^2, \|\varepsilon\|_{\psi_2}\leq K_{\varepsilon}\sigma$.  
\end{condition}

\begin{condition} \label{con:three}
$z$ has a continuous distribution supported on $[0,1]$. Its density is bounded below by a constant $\ell_z>0$.  
\end{condition}

\begin{condition} \label{con:four}
Define $\bh(z)=(h_1(z),\ldots, h_p(z))=\mathbb{E}(\bx|z)$ and $\tilde{\bx}=\bx-\bh(z)$. Assume $\lambda_{\min}(\mathbb{E}\tilde{\bx}\tilde{\bx}')\geq \Lambda_{min}>0$ and $\lambda_{\max}(\mathbb{E}\bh(z)\bh(z)')\leq \Lambda_{max}<\infty$, {\color{black}where $\Lambda_{\min}$ and $\Lambda_{\max}$ are positive and bounded constants}. 
\end{condition}

\begin{condition} \label{con:five}
$\max_{1\leq j \leq p} {\rm TV}(h_j^{(k)})\leq L_h.  $
\end{condition}

\begin{condition} \label{con:six}
$\frac{s^2\log p +s\log^2 p}{n}=o(1)$ and $p\rightarrow \infty$, as $n\rightarrow \infty$. 
\end{condition}

Compared to \cite{muller2015partial}, Condition \ref{con:one} relaxes the assumption of $\bx$ from being uniformly bounded to sub-Gaussian. Compared to \cite{yu2019minimax}, Condition \ref{con:one} only requires sub-Gaussianity for marginal distributions of $\bx$, instead of the joint distribution. Condition \ref{con:two} is the same as in \cite{yu2019minimax}, relaxing the errors from being standard normal in \cite{muller2015partial} to sub-Gaussian. For Condition \ref{con:three}, the continuity assumption is very weak, and the lower bound on the density is mainly used to bound the maximum gap between adjacent input points with high probability. See \cite{wang2014falling, sadhanala2019additive} for similar assumptions in the context of trend filtering. Condition \ref{con:four} is common in semiparametric literature \citep{yu2011semi, muller2015partial, yu2019minimax}. It ensures that there is enough information in the data to identify the parameters in the linear part. Condition \ref{con:five} is similar to Condition 2.6 in \cite{muller2015partial} and Assumption A.5 in \cite{yu2019minimax}. This condition enables to obtain the fast rate for $\hat{\bbeta}$. Condition \ref{con:six} is a scaling condition in high dimension. 
%While the first part  $(s^2\log p)/n=o(1)$ is stronger than the common assumption $(s\log p)/n=o(1)$ in the lasso {\color{blue}LASSO} literature, the gain is to avoid making any joint distribution assumption for $\bx$. 
{While the condition $(s^2\log p)/n=o(1)$ is stronger than the commonly assumed $(s\log p)/n=o(1)$  in the LASSO literature, it allows for the advantage of avoiding any assumptions about the joint distribution of $\bx$.} It is possible to only require the weaker condition $(s\log p)/n=o(1)$, if certain distributional assumption (e.g. joint sub-Gaussian) on $\bx$ is made. We leave this for a future study. 
%{\bf this part might be more apporpriately placed in the discussion sec}

Our main results consist of two parts, the result regarding $\hat{g}$ for the nonparametric part, and the result about $\hat{\bbeta}$ for the high-dimensional linear part. We now move on to the convergence rate result for $\hat{g}$. 

%{\color{blue} Our primary results are twofold: the convergence properties of $\hat{g}$ for the nonparametric component and those of $\hat{\bbeta}$, the high-dimensional linear coefficient vector. We first  establish the convergence rate for $\hat{g}$.}
\begin{theorem}
\label{thm:g}
Assume Conditions \ref{con:one}-\ref{con:four} and \ref{con:six}. Choose $\lambda=c_1\sqrt{\frac{\log p}{n}}, \gamma =c_2(\frac{s\log p}{n}+n^{-\frac{2k+2}{2k+3}})$ with large enough constants $c_1,c_2>0$. Then, there exist constants $c_3,c_4,n_0>0$ such that any solution $\hat{g}$ in \eqref{est:p} satisfies 
\begin{gather*}
\|\hat{g}-g_0\|^2 \leq c_3\Big(\frac{s\log p}{n}+n^{-\frac{2k+2}{2k+3}}\Big), \\ \|\hat{g}-g_0\|_n^2\leq c_3\Big(\frac{s\log p}{n}+n^{-\frac{2k+2}{2k+3}}\Big),
\end{gather*} 
with probability at least $1-p^{c_4}-n^{c_4}$, as long as $n\geq n_0$. The constants $c_1,c_2,c_3,c_4,n_0$ may depend on $k, L_g, L_h, K_x, K_{\epsilon}, \ell_z, \Lambda_{min}, \Lambda_{max}, \sigma$. 
\end{theorem}

\textcolor{black}{The expressions for the constants $\{c_i\}_{i=1}^4$, though potentially involved and not sharp, can be derived by keeping track of the explicit forms of constants $\{D_i\}$ in the proof of Theorem \ref{thm:g} (see Appendix \ref{prove:both} for the details). Given that our focus is on the convergence rate in terms of the parameters $\{n,s,p,k\}$, we do not define $\{c_i\}_{i=1}^4$ explicitly in the theorem. Similar treatments appear in the partial linar model and trend filtering literature \citep{yu2019minimax,tibshirani2014adaptive,sadhanala2019additive}. Viewing $\{c_i\}_{i=1}^4$ as fixed,} Theorem \ref{thm:g} shows that the integrated squared error and the input-averaged squared error have the same convergence rate, and the rate is determined by the maximum between a sparse estimation rate $(s\log p)/n$ and a nonparametric rate $n^{-(2k+2)/(2k+3)}$. When $g_0$ is sufficiently smooth, belonging to a $k$-order Sobolev or Holder class, a similar rate-switching phenomenon (switching between $(s\log p)/n$ and $n^{-(2k)/(2k+1)}$) has been revealed for partial linear smoothing spline \eqref{plss:def} \citep{muller2015partial}, and the rate is proved to be (nearly) minimax optimal \citep{yu2019minimax}. Given that the function class $\mathcal{V}_k(L_g)$ from \eqref{tvf:def} considered in Theorem \ref{thm:g} is larger than a $(k+1)$-order Sobolev class, the minimax lower bound derived for $(k+1)$-order Sobolev classes in \cite{yu2019minimax}, i.e. $(s\log(p/s))/n+n^{-(2k+2)/(2k+3)}$, implies that the rate obtained by our partial linear trend filtering \eqref{est:p} is (nearly) minimax optimal. In particular, when $(s\log p)/{n} = o(n^{-(2k+2)/(2k+3)})$, our method $\hat{g}$ achieves the optimal nonparametric rate $n^{-(2k+2)/(2k+3)}$ that is not attainable by partial linear smoothing spline (see related discussions in Section \ref{sec:intro}). We proceed to the convergence rate result for $\hat{\bbeta}$.
\begin{theorem} \label{thm:beta}
Assume Conditions \ref{con:one}-\ref{con:six}, with the same choice of $\lambda, \gamma$ in Theorem \ref{thm:g}, any solution $\hat{\bbeta}$ in \eqref{est:p} satisfies 
\begin{gather*}
\|\hat{\bbeta}-\bbeta^0\|_2^2 \leq c_3 \frac{s\log p}{n}, \\ \|\bx'(\hat{\bbeta}-\bbeta^0)\|^2\leq c_3 \frac{s\log p}{n}, \\ \|\bX(\hat{\bbeta}-\bbeta^0)\|_n^2 \leq c_3 \frac{s\log p}{n},
\end{gather*}
with probability at least $1-p^{c_4}-n^{c_4}$, as long as $n\geq n_0$. \textcolor{black}{The constants $c_1,c_2,c_3,c_4,n_0$ are identical to those in Theorem \ref{thm:g}.}
\end{theorem}

\textcolor{black}{As discussed after Theorem 2.1, the constants $\{c_i\}_{i=1}^4$ can be explicitly defined through careful bookkeeping in the proof. We do not present the explicit expressions in Theorem \ref{thm:beta}, as we focus on the convergence rate with respect to $\{n,s,p\}$. Treating $\{c_i\}_{i=1}^4$ as fixed,} Theorem \ref{thm:beta} demonstrates that the estimation error, out-of-sample prediction error, and in-sample prediction error, all have the same convergence rate $(s\log p)/n$. It is well known that this rate is the typical rate that the LASSO achieves in standard high-dimensional sparse linear regressions \citep{tsybakov2009simultaneous, ye2010rate, raskutti2011minimax, verzelen2012minimax}, \textcolor{black}{though the constant $c_3$ depends on additional parameters such as $\ell_z, \Lambda_{\min}, \Lambda_{\max}, L_h$.} Therefore, we can conclude that our estimator $\hat{\bbeta}$ attains the oracle rate \textcolor{black}{(up to a constant factor)} as if the true function $g_0$ were known. \cite{muller2015partial, yu2019minimax} showed that the partial linear smoothing spline can achieve the same rate, however, only when $g_0$ lies in Sobolev or Holder classes. In contrast, our method obtains the rate when $g_0 $ belongs to a larger class $\mathcal{V}_k(L_g)$ that covers more heterogeneously smooth functions. All the proofs related to Theorem \ref{thm:g} and Theorem \ref{thm:beta} are provided in Appendix \ref{proof:thm1} and Appendix \ref{proof:thm2}.

{\color{black}
\subsection{Computational details}
\label{com:detail}

As described in Section \ref{model:estimation}, to compute $(\hat{\bbeta}, \hat{g})$ in \eqref{est:p}, we first solve \eqref{est:d} to obtain $(\hat{\bbeta}, \hat{\btheta})$. Then, $\hat{g}(t)=\sum_{\ell=1}^n\hat{\alpha}_{\ell}q_{\ell}(t)$, where $\hat{\balpha}=\bQ^{-1}\hat{\btheta}, \bQ=(q_{\ell}(z_k))_{1\leq \ell, k\leq n}$, and $q_{\ell}$'s are the falling factorial basis defined in \eqref{ffb:one}-\eqref{ffb:two}. We use a Block Coordinate Descent (BCD) algorithm for solving the optimization \eqref{est:d}. The algorithm iterates over two blocks, $\bbeta$ and $\btheta$, by solving a standard LASSO problem and univariate trend filtering respectively. The detailed steps of the algorithm are outlined in Algorithm \ref{alg:lm}.  

\RestyleAlgo{ruled}
\SetKwComment{Comment}{/* }{ */}
\SetKw{TuningParam}{Fixed (tuning) Parameters:}
\SetKw{ConvErr}{Predefined Error Thereshold:}
\begin{algorithm}[!t]
\caption{A BCD algorithm for high-dimensional partial linear trend filtering}\label{alg:lm}
\small
\KwData{$\{y_i, \bx_i, z_i\},~i=1,\dots,n$}
\TuningParam{$\lambda, \gamma$} \\
\ConvErr{$\epsilon$}
\begin{enumerate}
    \item Set $t = 0$ and initialization $\btheta^{(0)}$
    \item For $(t+1)$-th iteration, where $t = 0,1,2,\dots$:
    \begin{enumerate}
    %\item[] \textbf{- Block 1}
        \item \textbf{Block 1}: Let $y_{i}^{(t)*} = y_i - \theta^{(t)}_i$, and update $\bbeta^{(t)}$ by fitting the LASSO:
        \begin{align*}
            \bbeta^{(t+1)} = \argmin_{\bbeta} \frac{1}{2}\| \by^{(t)*} - \bX \bbeta \|^2_n + \lambda \| \bbeta \|_1
        \end{align*}
        %\item[] \textbf{- Block 2}
        \item \textbf{Block 2}: Let $y_{i}^{(t)**} = y_i - \bx_i'\bbeta^{(t+1)}$, and update $\btheta^{(t)}$ by fitting the univariate trend filtering:
            \begin{align*}
                \btheta^{(t+1)} = \argmin_{\btheta} \frac{1}{2}\| \by^{(t)**} - \btheta \|^2_n + \gamma \| D^{(\bz, k+1)}\btheta \|_1
            \end{align*}
        \item If $\| \bX\bbeta^{(t+1)} + \btheta^{(t+1)} -  \bX\bbeta^{(t)} - \btheta^{(t)}\|_{n}^{2} < \epsilon$, then stop the iteration. If not, continue the iteration until it reaches the predefined maximum iteration number
    \end{enumerate}
    \item Return ($\bbeta^{(t+1)}, \btheta^{(t+1)})$ at convergence 
\end{enumerate}
\end{algorithm}

The objective function in \eqref{est:d} can be written in the following form:
\begin{align}
     \frac{1}{2}\|\by-\bX\bbeta-\btheta\|_n^2+\lambda\|\bbeta\|_1+\gamma \| D^{ (\bz, k+1)} \btheta \|_1: = f_0(\bbeta, \btheta) + f_1(\bbeta) + f_2(\btheta) = f(\bbeta, \btheta),\label{resp:eq1}
\end{align}
where $f_0$ is convex and differentiable, and $f_1,f_2$ are convex but nondifferentiable. It is direct to verify that $f$ is continuous on the compact set $\{(\bbeta,\btheta): f(\bbeta,\btheta)\leq f(\bbeta^{(1)},\btheta^{(1)})\}$, and $f$ attains its minimum, denoted by $f(\bbeta^*, \btheta^*)=\min_{\bbeta,\btheta}f(\bbeta, \btheta)$. This fact combined with Theorem 4.1 and Lemma 3.1 in \cite{tseng2001convergence} shows that there exists a subsequence $\{(\bbeta^{(k_n)}, \btheta^{(k_n)})\}$ converging to a stationary point of $f$. Due to the convexity of $f$, it further implies $(\bbeta^{(k_n)}, \btheta^{(k_n)})\rightarrow (\bbeta^*,\btheta^*)$, as $n\rightarrow \infty$. Given that $f$ is continuous and $f(\bbeta^{(t+1)}, \btheta^{(t+1)})\leq f(\bbeta^{(t)}, \btheta^{(t)}), \forall t=1,2,\ldots$, we can conclude that $f(\bbeta^{(t+1)}, \btheta^{(t+1)})$ converges to the global minimum $f(\bbeta^*, \btheta^*)$ as $t\rightarrow \infty$.

Our algorithm is implemented in R, utilizing the \texttt{glmnet} package for computing the LASSO in the first block update and the \texttt{glmgen} package for univariate trend filtering in the second block. In practice, we compute solutions over a two-dimensional grid of tuning parameters $(\lambda, \gamma)$, and use model selection criteria such as cross-validation to select the tuning parameter. This makes our method more computationally demanding than single-parameter cases because the tuning grid scales quadratically with the number of parameter values -- an issue extends to other double penalization approaches in the context of partial linear models. To address this issue, we employ efficient warm-start tricks \citep{glmnet2010, ramdas2016fast}, using warm-starts across adjacent $(\lambda,\gamma)$ values, to speed up the computations of solutions on the two-dimensional grid. Consequently, the computation time remains manageable for both our simulation settings and real data analysis. A similar strategy is adopted to compute partial linear smoothing splines. The R function, \texttt{stats::smooth.spline}, is used to calculate the univariate smoothing spline for the second block update. The implemented BCD algorithms as an R package, \texttt{plmR}, for PLTF and PLSS are publicly available at \url{https://github.com/SangkyuStat/plmR}.}
	%%%%%%%%%%%%%%%%%%%%%%%%%%%%%%%%%%%%%%%%%%%%%%%%%%%%%%%%%%%%%%%%%%%%%%%%%%%%%%%%%%%%%%%%%%%%%%%%%%%%%%%%%%%%%%%%%%%%%%%%%%%%

	\section{Simulations} \label{sec:simulation}

Through empirical experiments, we evaluate the performance of partial linear trend filtering (PLTF) introduced in \eqref{est:p}, in comparison to partial linear smoothing splines (PLSS) defined in \eqref{plss:def} \citep{muller2015partial, yu2019minimax}. 

\subsection{Simulation settings and results}\label{sec:simul_res}

We generate the $p$-dimensional covariates $\bx$ and the univariate covariate $z$ in the following way: we first sample $\tilde{\bx} = (\tilde{x}_1, \dots, \tilde{x}_{p+1})$ from $\mathcal{N}(\mathbf{0}, \Sigma)$, where $\Sigma = (\sigma_{jk})$ with $\sigma_{jk} = 0.5^{|j-k|}$ and $j,k=1,\dots,p+1$; then we set $z = \Phi (\tilde{x}_{25})$ with $\Phi$ being the standard normal's CDF, $x_j = \tilde{x}_j$ for $j=1,\dots,24$, and $x_j = \tilde{x}_{j+1}$ for $j = 25,\dots,p$. We consider three different partial linear models as follows: for $i=1,2,\ldots, n$,
\begin{align*}
    %&\text{Model 1 (Smooth):}~y_i = x_{i6}\beta_1 + x_{i12}\beta_2 + x_{i15}\beta_3 + x_{i20}\beta_4 + \sin (2\pi z_{i}) + \epsilon_i,\\
    %&\text{Model 2 (Heterogeneous smooth):}~y_i = x_{i6}\beta_1 + x_{i12}\beta_2 + x_{i15}\beta_3 + x_{i20}\beta_4 + \\
    %&~~~~~~~~~~~~~~~~~~~~~~~~~~~~~~~~~~~~~~~~~~~~~~~~~~~~e^{3 z_{i}}\sin(6 \pi z_{i})/7 + \epsilon_i,\\
    %&\text{Model 3 (Doppler-type):}~y_i = x_{i6}\beta_1 + x_{i12}\beta_2 + x_{i15}\beta_3 + x_{i20}\beta_4 + \\
    %&~~~~~~~~~~~~~~~~~~~~~~~~~~~~~~~~~~~~~~~~\sin(4/z_{i}) + \epsilon_i,
    &\text{Model 1 (Smooth function model):}~\\
    &~~~~~~~~~~~~~~~~~~~~~~~~~~~~~~~~~~y_i = x_{i6}\beta_1 + x_{i12}\beta_2 + x_{i15}\beta_3 + x_{i20}\beta_4 + \sin (2\pi z_{i}) + \epsilon_i,\\
    &\text{Model 2 (Heterogeneous smooth function model):}~\\
    &~~~~~~~~~~~~~~~~~~~~~~~~~~~~~~~~~~y_i = x_{i6}\beta_1 + x_{i12}\beta_2 + x_{i15}\beta_3 + x_{i20}\beta_4 + e^{3 z_{i}}\sin(6 \pi z_{i})/7 + \epsilon_i,\\
    &\text{Model 3 (Doppler-type function model):}~\\
    &~~~~~~~~~~~~~~~~~~~~~~~~~~~~~~~~~~y_i = x_{i6}\beta_1 + x_{i12}\beta_2 + x_{i15}\beta_3 + x_{i20}\beta_4 + \sin(4/z_{i}) + \epsilon_i,
\end{align*}
where $\epsilon_i \sim \mathcal{N}(0,\sigma_\epsilon^2)$. The $g_0$ function displays varying levels of heterogeneous smoothness in the three models. The actual forms of these functions are shown in Figure \ref{fig:gfunction}. The values for $\beta_j,~j=1,\dots,4$ are (0.5, 1, 1, 1.5). Various values for $\sigma_\epsilon^2$ are used to vary the signal-to-noise ratio for the models. Similar to \cite{muller2015partial, sadhanala2019additive}, we define the total signal-to-noise ratio (tSNR) as
\begin{align*}
    \text{tSNR} = \sqrt{ \frac{ \mathbb{E}(\bx' \bbeta^0 + g_0(z))^2}{\sigma^2_\epsilon} }.
\end{align*}
\begin{figure}[!t]
    \centering
    \includegraphics[width=0.9\textwidth, height=0.23\textheight]{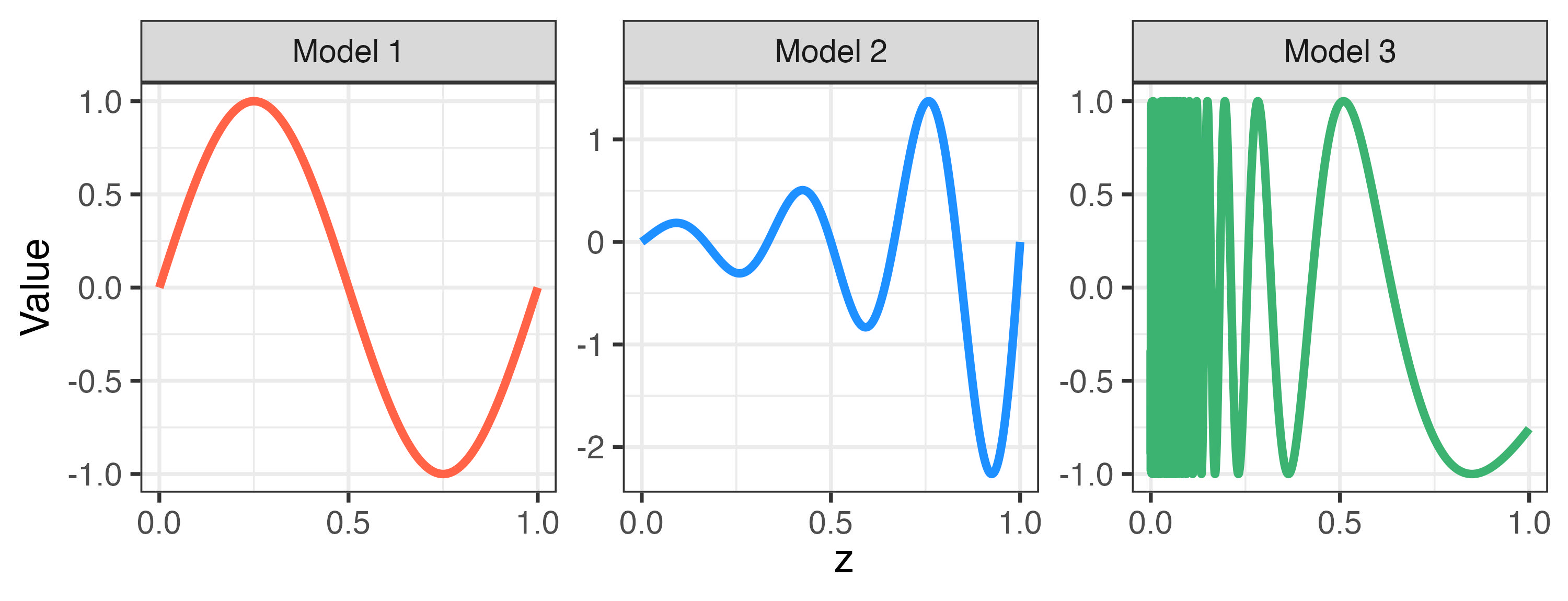}
    \caption{The function $g_0(z)$ in Models 1 -- 3. The function becomes increasingly locally heterogeneous from Model 1 (left) to Model 3 (right).}
    \label{fig:gfunction}
\end{figure}
We vary the tSNR from 4 to 16 on a logarithmic scale and then calculate the error metrics for each method. Specifically, we consider three different error metrics:
\begin{enumerate}
    \item $\| \bX \hat{\bbeta} + \hat{g}(\bz) - (\bX \bbeta^0 + g_0(\bz)) \|_n^2$ : mean squared error (MSE) for $\bX\hat{\bbeta}+\hat{g}(\bz)$.
    \item $\|\hat{\bbeta} - \bbeta^0 \|_2^2$ : $l_2$-norm squared error for $\hat{\bbeta}$.
    \item $\| \hat{g}(\bz) - g_0(\bz) \|_n^2$ : mean squared error for $\hat{g}$.
\end{enumerate}
The metrics are computed over 150 repetitions of randomly generated datasets for each tSNR value, and the medians of each metric are selected as the final results.
We consider $p$ to be 100 and 1000 for low and high-dimensional cases, respectively, with $n$ fixed at 500. We compare partial linear cubic smoothing spline ($k=3$ in \eqref{plss:def}) and second degree partial linear trend filtering ($k=2$ in \eqref{est:p}) so that both methods regularize the second derivative of $g$. A similar comparison has been performed in \cite{sadhanala2019additive} under additive models. To ensure fair comparisons, we present results using both optimally tuned parameters and cross-validation (CV)-tuned parameters for $(\lambda, \gamma)$ in the two methods.

\begin{figure}[!t]
\centering
%\textbf{\small Conditional longitudinal disparity decomposition with modifier (cmLDD)}
\includegraphics[width=0.78\textwidth, height=0.7\textheight]{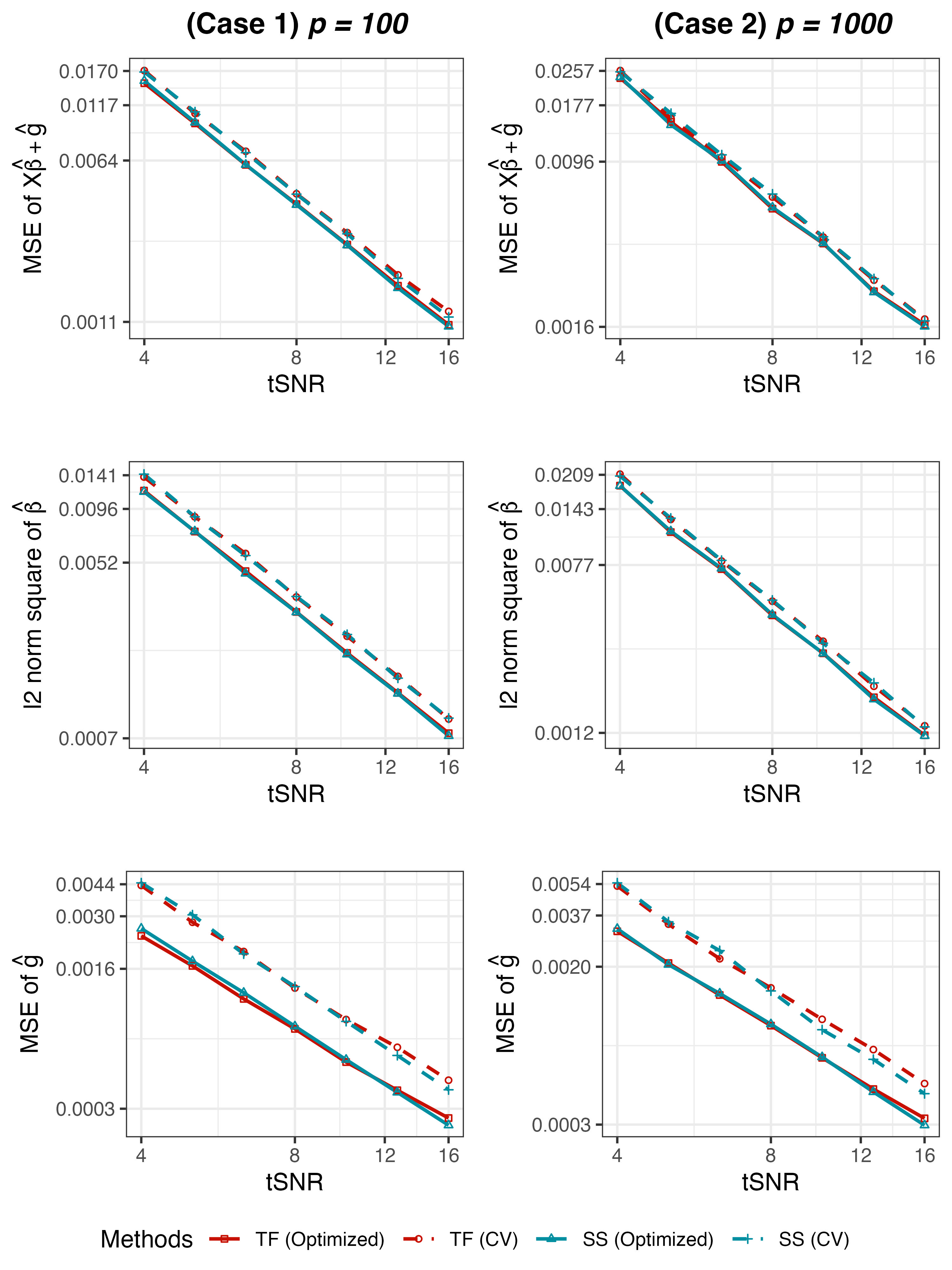}
\caption{PLTF v.s. PLSS under Model 1, with tSNR ranging from 4 to 16. TF (Optimized or CV) denotes PLTF with (optimally or CV) tuned parameters. SS (Optimized or CV) denotes PLSS with (optimally or CV) tuned parameters.
}
\label{fig:model1}
\end{figure}

\begin{figure}[!t]
\centering
%\textbf{\small Conditional longitudinal disparity decomposition with modifier (cmLDD)}
\includegraphics[width=0.78\textwidth, height=0.7\textheight]{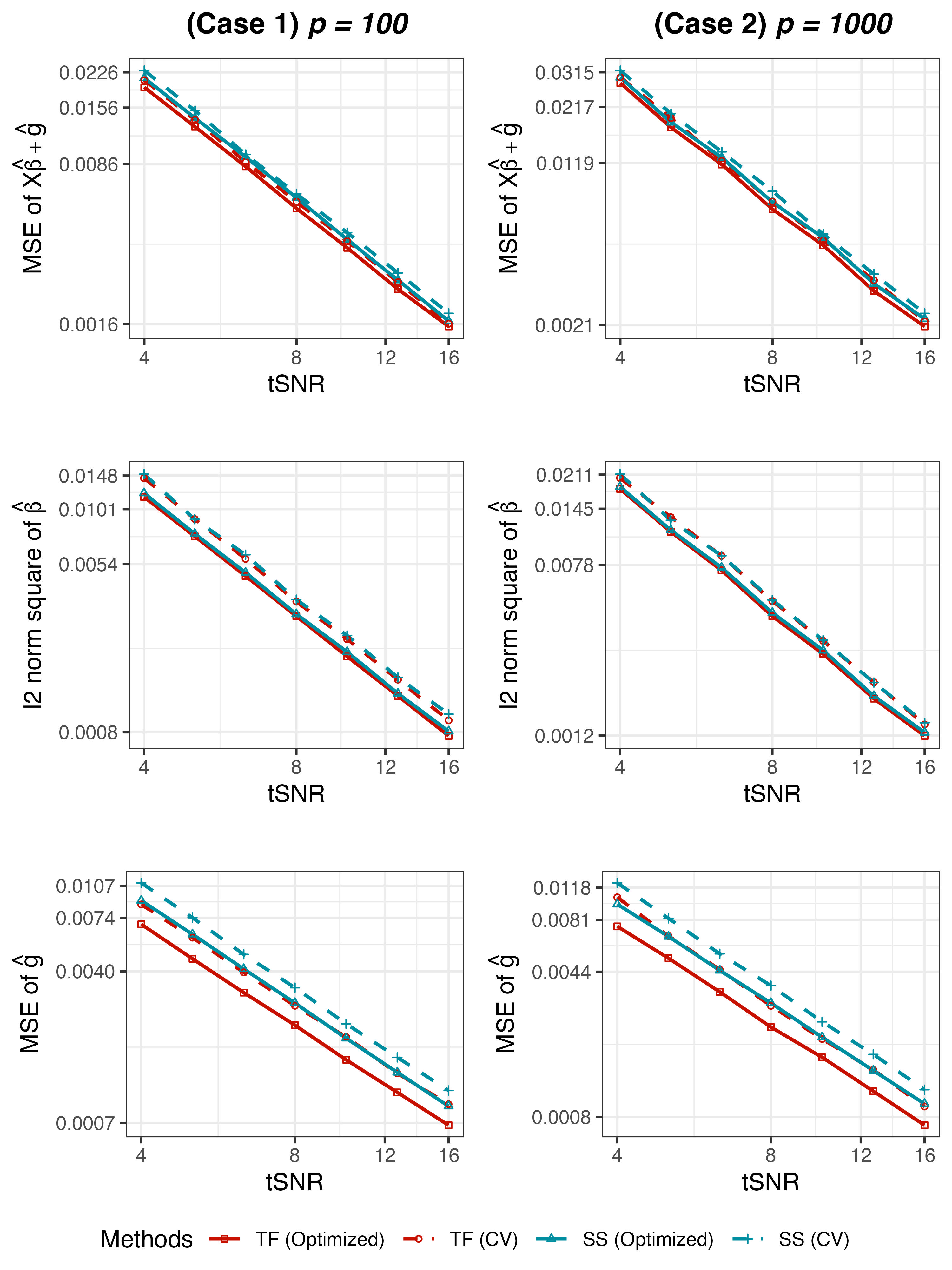}
\caption{PLTF v.s. PLSS under Model 2, with tSNR ranging from 4 to 16. TF (Optimized or CV) denotes PLTF with (optimally or CV) tuned parameters. SS (Optimized or CV) denotes PLSS with (optimally or CV) tuned parameters.
}
\label{fig:model2}
\end{figure}

\begin{figure}[!t]
\centering
%\textbf{\small Conditional longitudinal disparity decomposition with modifier (cmLDD)}
\includegraphics[width=0.78\textwidth, height=0.7\textheight]{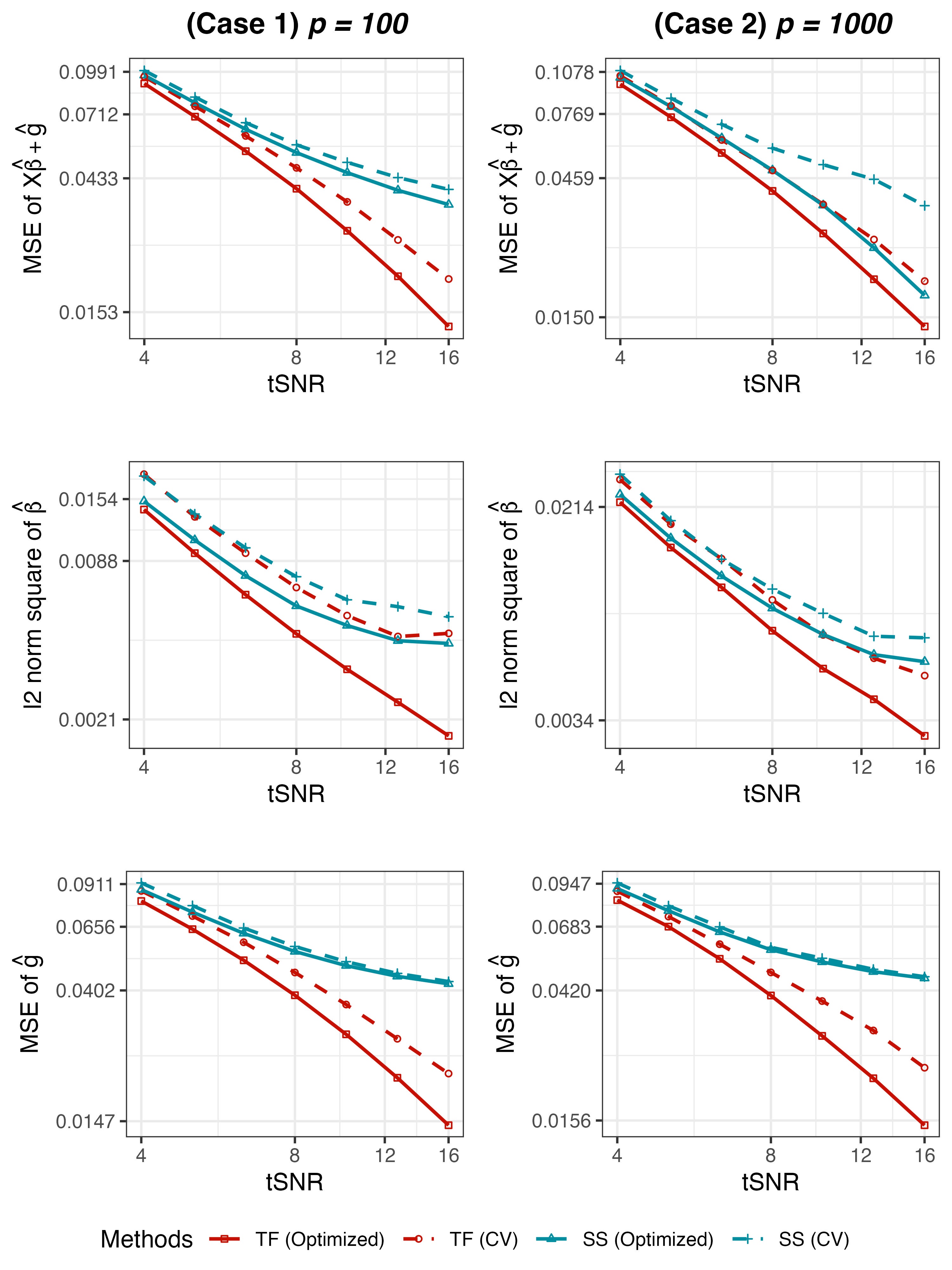}
\caption{PLTF v.s. PLSS under Model 3, with tSNR ranging from 4 to 16. TF (Optimized or CV) denotes PLTF with (optimally or CV) tuned parameters. SS (Optimized or CV) denotes PLSS with (optimally or CV) tuned parameters.
}
\label{fig:model3}
\end{figure}

The simulation results are displayed in Figures \ref{fig:model1}-\ref{fig:model3}. %Figure S1 is displayed in the supplementary material due to the space limitation.
Figure \ref{fig:model1} demonstrates that when the function $g_0$ has homogeneous smoothness, the performance of the PLTF method is comparable, almost the same, to that of the PLSS method, with no significant differences observed. However, as shown in Figure \ref{fig:model2}, when the function exhibits {varying} smoothness, the difference between the two methods becomes more pronounced.
%{\bf seems to be a redundant sentence: PLTF starts to outperform PLSS, especially for the estimation of $g_0$.} 
When the degree of heterogeneous smoothness of $g_0$ increases, as seen in Figure \ref{fig:model3}, all three errors of PLTF are significantly lower than those of PLSS in both low and high-dimensional cases. This disparity becomes more apparent as the tSNR increases. These simulation results suggest that for functions with homogeneous smoothness, PLTF and PLSS are competitive. However, as the function becomes more heterogeneously smooth, PLTF can significantly outperform PLSS in both low and high-dimensional scenarios, particularly when the tSNR is high. The empirical findings are aligned with those for comparing trend filtering and smoothing splines in the context of univariate nonparametric regressions and additive regression models \citep{tibshirani2014adaptive, sadhanala2019additive}.

{\color{black}
We also conduct similar simulations for $n=100$, to evaluate the performance of PLTF and PLSS under smaller sample size. The results show similar patterns to those observed in Figures \ref{fig:model1}-\ref{fig:model3}. We defer the details to Appendix \ref{supp:n_100}.
}
	
\section{Applications to the IDATA Study} 
\label{real:data}

%\subsection{Interactive Diet and Activity Tracking in AARP (IDATA) Study}
%The Interactive Diet and Activity Tracking in AARP (IDATA) Study aimed to improve dietary intake assessments and their links to health outcomes \citep{subar2020performance}. Using web-based tools and reference biomarkers, the study validated dietary patterns and their associations with metabolic and nutritional biomarkers. Participants were recruited from AARP members in Pittsburgh, Pennsylvania. Its robust design and comprehensive data make IDATA a key resource for exploring the complex relationships between dietary habits, metabolic profiles, and health outcomes.

%The amount of the intake of ultra-processed foods (UPF) was estimated using the NOVA system, which is developed to measure the intake of UPF based on the participants' diet \citep{steele2023identifying}. The metabolomics data were collected from the serum (blood) and urine samples. The detailed data collection protocols and procedures are available in the supplementary materials.

%\subsection{Data analyses and results}

%{\color{orange} Relationship between upf intake and sex (to make linkage to use subpopulation):}
The Interactive Diet and Activity Tracking in AARP (IDATA) Study aimed to improve dietary intake assessments and their links to health outcomes \citep{subar2020performance}. The amount of the intake of ultra-processed foods (UPF) was estimated using the NOVA system, which is developed to measure the intake of UPF based on the participants' diet \citep{steele2023identifying}. This approach and works already have been established between UPF and biomarkers \citep{abar2025identification}. More detailed data collection protocols, information and procedures are available in the Appendix \ref{app:IDATA}.

{To investigate the relationship between metabolomic profiles and UPF intake, we define the response variable $\by$ as UPF intake measured in grams and the covariates $\bX$ as metabolite measurements. We analyze 952 metabolites ($p = 952$) from the serum dataset and 1,045 metabolites ($p = 1,045$) from the urine dataset.}
{UPF intake is well established to differ between females and males \citep{juul2018ultra, sung2021consumption}. Accordingly, we classify participants into two groups: females ($n = 365$) and males ($n = 353$). Previous studies have also shown that certain obesity-related body measurements,  such as BMI, waist circumference, and hip circumference \citep{canhada2020ultra}, as well as age \citep{fu2024association}, exhibit significant nonlinear relationships with UPF intake.
Therefore,  we consider a univariate $z$ in \eqref{est:d}, representing as age, BMI, waist circumference, or hip circumference. Previously, a LASSO-based approach was used for feature selection in this dataset, but it did not account for potential nonlinear relationships \citep{abar2025identification}. This motivated our use of a partial linear modeling framework, which provides added flexibility.
We then apply the proposed method, PLTF, to identify significant metabolites associated with UPF and construct a prediction model to evaluate its predictive accuracy. 

To evaluate predictive accuracy, we compare the proposed PLTF method with PLSS and LASSO. Since our simulation settings reflect the structure of the real data, we maintain consistency by setting $k=2$ for PLTF and $k=3$ for PLSS, regularizing the second derivative of $g$.} 
%Motivated by these findings, we aim to explore the heterogeneous effects of these variables by generating two subgroups based on sex. The two categories considered are female ($n=365$) and male ($n=353$). Additionally, previous research has shown that certain obesity-related body measurements \citep{canhada2020ultra} and age \citep{fu2024association} exhibit significant nonlinear relationships with UPF intake. Therefore, we include \textit{age}, \textit{body mass index (BMI)}, \textit{hip circumference}, and \textit{waist circumference} as a variable $z$ and model them flexibly, without assuming linearity, by applying univariate trend filtering. 
%Inspired from the simulation, we present two main types of results from the analyses: (1) prediction accuracy and (2) selected metabolites. For prediction accuracy, we compare three distinct methods: PLTF, PLSS, and LASSO. This comparison is motivated by the use of PLTF and PLSS in both the theory and the simulation, while LASSO serves as a benchmark high-dimensional linear model when assuming $g_0$ to be a linear function of $z$. In the case of LASSO, $z$ is incorporated into the base model as a chosen covariate. Since our simulation settings closely resemble the structure of our real data, we make the data analysis setting consistent with our simulation settings. We set $k=2$ for PLTF and $k=3$ for PLSS, to regularize the second derivative of $g$. 
{To assess prediction performance, the dataset is randomly split into 90\% for training and 10\% for testing.} %To evaluate prediction performance, we randomly divide the dataset, using 90\% for training and 10\% for testing. 
Models are trained on the training dataset using 10-fold cross-validation (CV), and the mean squared error (MSE) is computed on the testing set. {This process is repeated 500 times to achieve standard errors within 1\% to 2\% of the MSE.} %This procedure is repeated 500 times, with the number of repetitions chosen to ensure standard errors are within 1\% to 2\% of the MSE. 
The prediction errors are reported in Table \ref{real:tab:prediction}. The prediction results indicate that, in the majority of cases, PLTF outperformed the other models (12 out of 16 cases). Notably, for the serum datasets, PLTF significantly outperformed the other methods, particularly in the female subgroup. In the urine datasets, PLTF generally demonstrated better performance. Even in cases where other models marginally outperformed PLTF, the difference between PLTF and the best-performing model was not statistically significant.

\begin{table}[!htbp] \centering 
\scriptsize
\caption{Prediction MSE results for different variables ($z$) with 500 repetitions for different types of datasets. Numbers in parentheses are the corresponding standard errors. The bolded numbers indicate the smallest error for each case, and it is underlined if it is at least the 1 standard error far from any other models. When it is farther than the 2 standard error, then it is double-underlined.} \label{real:tab:prediction} 
\begin{tabular}{@{\extracolsep{0pt}} ccccccc} 
\\[-1.8ex]\hline 
\hline 
& & & Age & BMI & Hip & Waist \\ 
\hline 
\multirow{6}{*}{Serum} & \multirow{3}{*}{Female} & PLTF & \underline{\textbf{175.390}} (2.289) & \underline{\underline{\textbf{174.789}}} (2.328) & \underline{\textbf{176.591}} (2.312) & \underline{\textbf{177.725}} (2.360) \\  
&&PLSS & 178.696 (2.383) & 180.286 (2.420) & 180.230 (2.390) & 180.331 (2.440) \\ 
&&LASSO & 188.370 (2.516) & 190.647 (2.555) & 190.461 (2.545) & 190.115 (2.543) \\ 
\cline{2-7} %%% between male female
&\multirow{3}{*}{Male}& PLTF & \textbf{152.774} (2.012) & \underline{\textbf{150.925}} (1.994) & 153.823 (1.991) & \textbf{154.350} (2.055) \\ 
&&PLSS & 152.904 (2.019) & 154.373 (2.019) & \textbf{153.470} (1.996) & 155.080 (2.070) \\ 
&&LASSO & 155.738 (2.036) & 156.044 (2.045) & 156.342 (2.040) & 156.619 (2.062) \\ 
\cline{1-7}
 \multirow{6}{*}{Urine} & \multirow{3}{*}{Female} & PLTF & \textbf{171.911} (2.361) & \textbf{169.977} (2.371) & \textbf{170.057} (2.350) & 172.937 (2.469) \\ 
&&PLSS & 172.365 (2.425) & 171.910 (2.427) & 171.859 (2.421) & {172.696} (2.473) \\ 
&&LASSO & 172.343 (2.468) & 172.626 (2.494) & 173.094 (2.488) & \textbf{172.584} (2.483) \\ 
\cline{2-7}
 & \multirow{3}{*}{Male} & PLTF & 144.786 (1.894) & \textbf{143.859} (1.905) & 145.162 (1.868) & \textbf{144.145} (1.920) \\ 
&&PLSS & \textbf{144.299} (1.894) & 144.790 (1.877) & \textbf{145.145} (1.856) & 144.803 (1.904) \\ 
&&LASSO & 148.020 (1.996) & 147.060 (1.997) & 147.009 (1.982) & 146.509 (2.017) \\ 
\hline \hline \\
\end{tabular} 
\end{table} 

{Next, in Tables \ref{real:tab:total_selected:serum} and \ref{real:tab:total_selected:urine}, we present the top 10 most commonly selected for each $z$ variable (i.e., age, BMI, waist circumference, and hip circumference) stratified by biospecimen types (serum and urine) and sex (male and female). These metabolites are ranked according to the absolute values of their coefficients.} To ensure fair comparisons across coefficients, all covariates are standardized. The full lists of the top 10 selected variables are provided in Tables \ref{supp:tab:serum_selected:f}-\ref{supp:tab:urine_selected:m} for different datasets in Appendix \ref{supp:tables}. 

The selected variables differ across different subsets, with their super pathways and biological functions also varying. For example, saccharin, which was selected in the male serum dataset and the female urine dataset, is a well-known artificial sweetener commonly found in UPF products. Given that the signs of the coefficients are positive, this aligns with expected findings. Quinate, selected in both the male serum and urine datasets, is commonly present in plant-based foods such as coffee, fruits, and vegetables. This suggests that plant-based food consumption in the male population may have an indirect negative association with UPF intake, as reflected by the negative coefficient which aligns with the results in \cite{playdon2024measuring}.

%\begin{table}[!htbp] \centering 
%\begin{tabular}{@{\extracolsep{5pt}} cccc} 
%\\[-1.8ex]\hline 
%\hline \\[-1.8ex] 
% & & & Selected Variables \\ 
%\hline \\[-1.8ex] 
% & \multirow{8}{*}{Serum} & \multirow{4}{*}{Female} & eicosapentaenoylcholine, diglycerol, \\ 
% &&& anthranilate, 4-allylphenol-sulfate, \\ 
% &&& 25523, 21807, \\ 
% &&& 17145, 1-stearoyl-2-adrenoyl-gpc--18 \\ 
% \cline{3-4} \\ [-1.8ex]
%&& \multirow{4}{*}{Male} & saccharin, quinate, \\ 
% &&& 23655, 21442 \\
% &&& 19183, 1-methylhistidine \\
% &&& 1-lignoceroyl-gpc--24-0-, 1--1-enyl-palmitoyl--2-o-gpe \\
% \cline{3-4} \\ [-1.8ex]
%&\multirow{10}{*}{Urine} & \multirow{5}{*}{Female} & saccharin, riboflavin--vitamin-b2-, \\ 
% &&& galactonate, 999925952, \\ 
% &&& 25936, 23680, \\ 
% &&& 12753, 12096, \\ 
% &&& 100020267, 2-3-dihydroxypyridine \\  \cline{3-4} \\ [-1.8ex]
%&& \multirow{5}{*}{Male} & ursocholate, quinate, \\ 
% &&& n-acetylcitrulline, cortisone, \\ 
% &&& 999925952, 999923161, \\ 
% &&& 999913847, 3-methoxytyramine, \\ 
% &&& 1-6-anhydroglucose \\ 
%\hline \hline \\[-1.8ex] 
%\end{tabular} 
%  \caption{Top 10 selected variables results for urine dataset, ordered based on the magnitude of coefficients.} 
%  \label{real:tab:total_selected} 
%\end{table} 

\begin{table}[!htbp] \centering 
\scriptsize
\caption{Top 10 selected metabolites and their super pathways results for the serum dataset, which are selected at least 3 times for 4 different $z$ variables. + and - in parentheses denote the signs of coefficients for metabolites. } 
  \label{real:tab:total_selected:serum} 
\begin{tabular}{@{\extracolsep{5pt}} ccll} 
\\[-1.8ex] 
\hline 
\hline
 & & \multicolumn{2}{c}{Selected Variables} \\ 
  % & & \multicolumn{2}{c}{Female} & \multicolumn{2}{c}{Male} \\ 
  & & Biochemical & Super Pathway \\ 
\hline 
 \multirow{18}{*}{Serum} & \multirow{9}{*}{Female} & Diglycerol (+) & Xenobiotics \\ 
 && 4-Allylphenol sulfate (-) & Xenobiotics \\ 
 && Anthranilate (+) & Amino Acid  \\ 
 && Eicosapentaenoylcholine (-) & Lipid  \\ 
 && 1-Stearoyl-2-Adrenoyl-GPC & Lipid  \\ 
 && (18:0/22:4) (+) &   \\ 
 && Branched chain 14:0 dicarboxylic acid (-) & Lipid  \\   
 && X-21807 (-) & Unknown  \\ 
 && X-25523 (-) & Unknown  \\   \cline{2-4} 
 & \multirow{9}{*}{Male} & Quinate (-) & Xenobiotics \\ 
 && Saccharin (+) &  Xenobiotics \\
 && 1-Methylhistidine (+) & Amino Acid \\
 && 1-(1-Enyl-palmitoyl)-2-Oleoyl-GPE & Lipid \\
 && (P-16:0/18:1) (-) &  \\
 && 1-Lignoceroyl-GPC (24:0) (-) & Lipid \\
 && X-19183 (-) & Unknown \\
 && X-21442 (-) & Unknown \\
 && X-23655 (-) & Unknown \\
\hline \hline \\[-1.8ex] 
\end{tabular} 
\end{table} 

\begin{table}[!htbp] \centering 
\scriptsize
\caption{Top 10 selected metabolites and their super pathways results for the urine dataset, which are selected at least 3 times for 4 different $z$ variables. + and - in parentheses denote the signs of coefficients for metabolites. } 
  \label{real:tab:total_selected:urine}
\begin{tabular}{@{\extracolsep{5pt}} ccll} 
\\[-1.8ex] 
\hline 
\hline
 & & \multicolumn{2}{c}{Selected Variables} \\ 
  % & & \multicolumn{2}{c}{Female} & \multicolumn{2}{c}{Male} \\ 
  & & Biochemical & Super Pathway \\ 
\hline 
\multirow{19}{*}{Urine} & \multirow{10}{*}{Female} & Galactonate (-) & Carbohydrate \\ 
 && Riboflavin (vitamin B2) (-) & Cofactors and Vitamins \\
 && Saccharin (+) & Xenobiotics \\
 && 2,3-Dihydroxypyridine (-) & Xenobiotics \\
 && Glutamine conjugate of C8H12O2 (3) (-) & Partially Characterized Molecules \\
 && X-12096 (+) & Unknown \\
 && X-12753 (-) & Unknown \\
 && X-23680 (+) & Unknown \\
 && X-25936 (+) & Unknown \\
 && X-25952 (+) & Unknown \\ \cline{2-4} 
  %&& &  \\
 & \multirow{9}{*}{Male} & Quinate (-) & Xenobiotics \\ 
 && 1,6-Anhydroglucose (+) & Xenobiotics \\
 && N-Acetylcitrulline (+) & Amino Acid \\
 && 3-Methoxytyramine (+) & Amino Acid \\
 && Cortisone (-) & Lipid \\
 && Ursocholate (+) & Lipid \\
 && X-13847 (-) & Unknown \\
 && X-23161 (-) & Unknown \\
 && X-25952 (+) & Unknown \\
\hline \hline \\[-1.8ex] 
\end{tabular} 
\end{table}

Several other findings, including the direction of the coefficients, are consistent with existing research on UPF intake. For instance, 4-allylphenol sulfate was identified as a potential biomarker for UPF intake in a randomized, controlled, domiciled feeding trial \citep{o2023metabolomic}. Diglycerol can be found in cosmetics or can be esterified to diglycerol esters of fatty acids. Diglycerol esters of fatty acids are used as emulsifiers in food and are a hallmark of UPF. Food additive emulsifiers have been associated with higher risk of some types of cancer risks by \citet{sellem2024food}. 

Additionally, eicosapentaenoylcholine has been positively 
correlated with Vitamin D consumption \citep{leung2020serum}, and previous 
research shows that UPF intake is associated with higher prevalence of inadequate Vitamin D intake \citep{falcao2019processed}. Other nutrient deficiencies, including Vitamin B2 deficiency, have also been associated with higher UPF intake in women in \citet{schenkelaars2024intake}. Lastly, a positive association of levoglucosan with UPF intake has been identified in urine datasets 
\citep{muli2024association} and it could be a food contaminant from processing or packaging. Several other findings involve metabolites that have not been extensively studied or are unknown. These may represent potential novel biomarkers associated with UPF intake, providing new insights for future research.

\section{Discussion}
\label{discuss:con}
In this paper, we extend trend filtering to high-dimensional partial linear models via doubly penalized least squares, preserving the local adaptivity and efficiency of univariate trend filtering. We analyze degrees of freedom, establish optimal convergence rates, and demonstrate superior performance over standard smoothing splines for functions with varying smoothness. Our method identifies metabolites linked to UPF with reduced prediction error compared to alternatives. The following discussion explores key directions for future research.
%In this paper, we extend trend filtering to high-dimensional partial linear models using doubly penalized least squares estimation, preserving the local adaptivity and computational efficiency of univariate trend filtering. We analyze its degrees of freedom, establish optimal convergence rates, and demonstrate through empirical examples that our method outperforms standard smoothing splines when the underlying functions possess varying smoothness. {The proposed method identifies metabolites associated with UPF by leveraging univariate trend filtering, resulting in a prediction result with lower prediction error compared to competing methods.} There are several important directions for future research.

%When the covariate $\bz \in \mathbb{R}^d$ has dimension $d>1$, we can extend our work by considering an additive function form $g_0(\bz)=\sum_{j=1}^dg_{0j}(z_j)$ where $\sum_{j=1}^d{\rm TV}(g_{0j}^{(k)})\leq C$ for some constant $C>0$. The penalty ${\rm TV}(g^{(k)})$ in \eqref{est:p} is then replaced with $\sum_{j=1}^d{\rm TV}(g_{j}^{(k)})$. 

When the covariate $\bz \in \mathbb{R}^d$ has dimension $d>1$, our work extends naturally to an additive function form $g_0(\bz)=\sum_{j=1}^dg_{0j}(z_j)$ with $\sum_{j=1}^d{\rm TV}(g_{0j}^{(k)})\leq C$ for some constant $C>0$. Utilizing some techniques from additive trend filtering \citep{sadhanala2019additive}, the method \textcolor{black}{is expected to apply} when $d$ is bounded. For high-dimensional cases where $d$ is comparable to or exceeds $n$, a sparse additive form for $g_0(\bz)$ with a sparsity penalty (e.g., empirical $\ell_2$ norm) for each $g_{0j}$ may be employed. While theoretically challenging, recent results in high-dimensional additive trend filtering \citep{tan2019doubly} provide useful guidance.

A key question in partial linear models is identifying which covariates exhibit linear versus nonlinear effects. While this paper assumes such structure is known, in practice, prior knowledge is often unavailable. Penalization-based methods for structure selection typically assume the underlying functions belong to Sobolev or Hölder classes \citep{zhang2011linear, huang2012semiparametric, lian2015separation}. Future work could investigate how trend filtering aids in discovering model structures with more heterogeneously smooth functions.

Under partial linear models, variable selection and statistical inference for the linear part can be important tasks. It would be particularly valuable to study how recent advances in high-dimensional linear regression models \citep{zhang2014confidence, van2014asymptotically, javanmard2014confidence, dezeure2015high, barber2015controlling, candes2018panning, dai2023false} can be extended to high-dimensional partial linear trend filtering, in order to perform variable selection and inference with guaranteed error control.
%%%%%%%%%%%%%%%%%%%%%%%%%%%%%%%%%%%%%%%%%%%%%%
%% Example with single Appendix:            %%
%%%%%%%%%%%%%%%%%%%%%%%%%%%%%%%%%%%%%%%%%%%%%%
\begin{appendix}
\sloppy
\section*{}\label{appn} %% if no title is needed, leave empty \section*{}.
%Appendices should be provided in \verb|{appendix}| environment,
%before Acknowledgements.

%If there is only one appendix,
%then please refer to it in text as \ldots\ in the \hyperref[appn]{Appendix}.
\end{appendix}
%%%%%%%%%%%%%%%%%%%%%%%%%%%%%%%%%%%%%%%%%%%%%%
%% Example with multiple Appendixes:        %%
%%%%%%%%%%%%%%%%%%%%%%%%%%%%%%%%%%%%%%%%%%%%%%
\begin{appendix}

\section{Proof of Proposition \ref{prop:df}}\label{proof:prop}

We first cite Theorem 3 from \cite{tibshirani2012degrees} on the generalized LASSO degrees of freedom.

\begin{theorem}[\cite{tibshirani2012degrees}]
\label{glasso:thm}
Consider the generalized LASSO problem,
\begin{align}
\label{glasso:form}
\hat{\bfeta}\in \argmin_{\bfeta\in \mathbb{R}^q}\frac{1}{2}\|\by-\bH \bfeta\|_2^2+\lambda \|\bD \bfeta\|_1,
\end{align}
where $\bH \in \mathbb{R}^{n\times q}$ and $\bD\in \mathbb{R}^{m\times q}$. Define the active set $\mathcal{I}=\{1\leq i\leq m: (\bD \hat{\bfeta})_i\neq 0\}$. Assuming $\by\sim\mathcal{N}(\bmu, \sigma^2\bI)$, for any $\bH, \bD$ and $\lambda\geq 0$, the degrees of freedom of $\bH \hat{\bfeta}$ is 
\[
{\rm df}(\bH \hat{\bfeta})=\mathbb{E}\Big[{\rm dim}\big(\bH ({\rm null}(\bD_{-\mathcal{I}}))\big)\Big].
\]
\end{theorem}

\noindent The proof of Part (i) is a direct application of the above theorem. Denote
\begin{align*}
\bH=
\begin{bmatrix}
\bX & \bI 
\end{bmatrix}
\in \mathbb{R}^{n\times (n+p)}, ~~\bfeta=(\bbeta',\btheta')' \in \mathbb{R}^{n+p}, ~~\bD=
\begin{bmatrix}
\bI & \bzero \\
\bzero & \frac{\gamma}{\lambda} D^{(\bz, k+1)}
\end{bmatrix}
.
\end{align*}
Then, \eqref{est:d} can be rewritten in the form of \eqref{glasso:form}. Thus we can invoke Theorem \ref{glasso:thm} and obtain
\[
{\rm null}(\bD_{-\mathcal{I}})=\Big\{(\bbeta,\btheta): \bbeta_{\mathcal{A}^c}=\bzero, D^{(\bz, k+1)}_{-\mathcal{B}}\btheta=\bzero\Big\},
\]
where $\bbeta_{\mathcal{A}^c}$ denotes the vector $ \bbeta$ after removing components indexed by $\mathcal{A}$.
As a result,
\[
\bH({\rm null}(\bD_{-\mathcal{I}}))={\rm Col}(\bX_{\mathcal{A}})+{\rm Null}(D^{(\bz, k+1)}_{-\mathcal{B}}).
\]

\noindent Regarding Part (ii), the proof is motivated by that of Lemma 3 in \cite{sadhanala2019additive}. We use the equivalent formulation Equation (2.6) to derive the degrees of freedom. We have
\begin{align}
\label{equip:form}
&\min_{\bbeta\in \mathbb{R}^p, \balpha \in \mathbb{R}^n} \frac{1}{2}\|\by-\bX\bbeta-\bQ\balpha\|_n^2+\lambda\|\bbeta\|_1+\gamma k!\sum_{l=k+2}^{n} |\alpha_{l}| \nonumber \\
=&\min_{\bbeta\in \mathbb{R}^p, \balpha_2 \in \mathbb{R}^{n-k-1}} \frac{1}{2}\|\bU'(\by-\bX\bbeta-\bQ_2\balpha_2)\|_n^2+\lambda\|\bbeta\|_1+\gamma k!\|\balpha_2\|_1 \nonumber \\
=&\min_{\bbeta\in \mathbb{R}^p, \tilde{\balpha}_2 \in \mathbb{R}^{n-k-1}} \frac{1}{2}\|\tilde{\by}-\tilde{\bX}\bbeta-\tilde{\bQ}_2\tilde{\balpha}_2\|_n^2+\lambda\|\bbeta\|_1+\lambda\|\tilde{\balpha}_2\|_1,
\end{align} 
where $\tilde{\by}=\bU'\by, \tilde{\bX}=\bU'\bX, \tilde{\bQ}_2=\frac{\lambda}{\gamma k!}\bU'\bQ_2, \tilde{\balpha}_2=\frac{\gamma k!}{\lambda}\balpha_2$. The problem \eqref{equip:form} is a standard LASSO problem,
\begin{align}
\label{std:lasso}
\hat{\bfeta}\in \argmin_{\bfeta \in \mathbb{R}^{n+p-k-1}} \frac{1}{2}\|\tilde{\by}-\bH \bfeta\|_n^2+\lambda \|\bfeta\|_1,
\end{align}
where $\bH=[\bU'\bX, \frac{\lambda}{\gamma k!}\bU'\bQ_2]$ and $\bfeta=(\bbeta', \tilde{\balpha}_2')'$. Let $\tilde{\bU}\tilde{\bU}'$ be the projection operator onto the column space ${\rm Col}(\bQ_1)$ where $\tilde{\bU} \in \mathbb{R}^{n\times (k+1)}$ has orthonormal columns, and $\bJ=[\bX, \frac{\lambda}{\gamma k!}\bQ_2]$. Then the equivalence among Equation (2.6), \eqref{est:d}, \eqref{equip:form} and \eqref{std:lasso} yields 
\[
\bX\hat{\bbeta}+\hat{\btheta}=\tilde{\bU}\tilde{\bU}'\by + \bU\bU'\bJ\hat{\bfeta}.
\]
Therefore, 
\begin{align*}
{\rm df}(\bX\hat{\bbeta}+\hat{\btheta}) &= \frac{1}{\sigma^2} {\rm trace}({\rm Cov}(\bX\hat{\bbeta}+\hat{\btheta}, \by))\\
&=\frac{1}{\sigma^2}  {\rm trace}({\rm Cov}(\tilde{\bU}\tilde{\bU}'\by + \bU\bU'\bJ\hat{\bfeta}, \tilde{\bU}\tilde{\bU}'\by+\bU\bU'\by)) \\
&=\frac{1}{\sigma^2}  {\rm trace}({\rm Cov}(\tilde{\bU}\tilde{\bU}'\by , \tilde{\bU}\tilde{\bU}'\by))+ \frac{1}{\sigma^2}  {\rm trace}({\rm Cov}(\bU\bU'\bJ\hat{\bfeta}, \bU\bU'\by))\\
&=k+1+ \frac{1}{\sigma^2}  {\rm trace}({\rm Cov}(\bU'\bJ\hat{\bfeta}, \bU'\by))=k+1+ \frac{1}{\sigma^2}  {\rm trace}({\rm Cov}(\bH\hat{\bfeta}, \tilde{\by})).
\end{align*}
Note that $\frac{1}{\sigma^2}  {\rm trace}({\rm Cov}(\bH\hat{\bfeta}, \tilde{\by}))$ above is the degrees of freedom of the standard LASSO \eqref{std:lasso} under the model $\tilde{\by}\sim \mathcal{N}\left(\bU'\bX\bbeta^0 + g_0(\bz), \sigma^2 \bI\right)$. Since $\bH$ has columns in general position, the LASSO solution is unique \citep{tibshirani2013lasso}. As a result,
\begin{align*}
\frac{1}{\sigma^2}  {\rm trace}({\rm Cov}(\bH\hat{\bfeta}, \tilde{\by}))&=\mathbb{E}[{\rm number~of~nonzeros~in~} \hat{\bfeta}] \\
&=\mathbb{E}\big[\big|\{1\leq j\leq p: \hat{\beta}_j\neq 0\}\big|+\big|\{k+2\leq \ell \leq n: \hat{\alpha}_{\ell}\neq 0\}\big|\big] \\
&=\mathbb{E}[|\mathcal{A}|+|\mathcal{B}|],
\end{align*}
where in the last step we have used $\hat{\balpha}=\bQ^{-1}\hat{\btheta}$ and the expression for $\bQ^{-1}$ from \cite{wang2014falling}.

%\section{Title of the first appendix}\label{appA}
%If there are more than one appendix, then please refer to it
%as \ldots\ in Appendix \ref{appA}, Appendix \ref{appB}, etc.

%\section{Title of the second appendix}\label{appB}
%\subsection{First subsection of Appendix \protect\ref{appB}}

%Use the standard \LaTeX\ commands for headings in \verb|{appendix}|.
%Headings and other objects will be numbered automatically.
%\begin{equation}
%\mathcal{P}=(j_{k,1},j_{k,2},\dots,j_{k,m(k)}). \label{path}
%\end{equation}

\section{Proof of Theorem \ref{thm:g}}\label{proof:thm1}

Throughout the proof, we use $C_1,C_2,\ldots$ to denote universal constants and use $D_1,D_2,\ldots$ to denote constants that may only depend on the constants $k, L_g, L_h, K_x, K_{\epsilon}, \ell_z, \Lambda_{min}, \Lambda_{max}$ from Conditions \ref{con:one}-\ref{con:five}. An explicit (though not optimal) dependence of $\{D_j, j=1,2,\ldots\}$ on the aforementioned constants can be tracked down. However, since it does not provide much more insight, we will often not present the explicit forms of $\{D_j, j=1,2,\ldots\}$, and this will greatly help streamline the proofs. The constants $\{C_j, D_j,j=1,2,\ldots\}$ may vary from lines to lines.

\subsection{Roadmap of the proof}

For a given function $f(\bx,z)=\bx'\bbeta+g(z)$ with $\bbeta\in \mathbb{R}^p, {\rm TV}(g^{(k)})<\infty$, introduce the functional 
\begin{align}
\label{funl:def}
\tau_{\delta_0,R}(f)=\frac{\lambda\|\bbeta\|_1+\gamma {\rm TV}(g^{(k)})}{20\delta_0R}+\|\bx'\bbeta+g(z)\|, ~~~\delta_0\in (0,1),~~R>0, 
\end{align}
where for notational simplicity we have suppressed the dependence of $\tau_{\delta_0,R}(f)$ on the tuning parameters $\lambda, \gamma>0$. This functional will serve as a critical measure for $\bbeta$ and $g(z)$. Define the following events:
\begin{align*}
\mathcal{T}_1(\delta_0, R)&=\Bigg\{\sup_{\tau_{\delta_0,R}(f)\leq R}\Big|\|f\|_n^2-\|f\|^2\Big| \leq \delta_0 R^2\Bigg\}, \\
\mathcal{T}_2(\delta_0, R)&=\Bigg\{\sup_{\tau_{\delta_0,R}(f)\leq R}\Big|\langle \beps, f(\bX, \bZ)\rangle_n\Big| \leq \delta_0 R^2\Bigg\}, \\
\mathcal{T}_3&=\Bigg\{ \max_{1\leq i \leq n-1} (z_{(i+1)}-z_{(i)})\leq \frac{22\log n}{\ell_z n} \Bigg\}.
\end{align*}
The general proof idea is inspired by \cite{muller2015partial} (see also \cite{yu2019minimax}). We first use a localization technique based on convexity to show that under the event $\mathcal{T}_1(\delta_0, R) \cap \mathcal{T}_2(\delta_0, R) \cap \mathcal{T}_3$, the bound on $\|\hat{g}-g_0\|$ in Theorem \ref{thm:g} (bound on $\|\hat{g}-g_0\|_n$ is derived similarly) holds. This is done in Lemmas \ref{lemma:1}-\ref{lemma:2}. Then in Lemmas \ref{hg:one} and \ref{hp:event2} we show that the event $\mathcal{T}_1(\delta_0, R) \cap \mathcal{T}_2(\delta_0, R) \cap \mathcal{T}_3$ happens with high probability. Along the way we need to choose $R,\lambda, \gamma, \delta_0$ in \eqref{funl:def} properly to meet conditions in Lemmas \ref{lemma:1}-\ref{hp:event2} and to achieve the desirable error rate in Theorem \ref{thm:g}. We complete this step in Section \ref{prove:both}.

\subsection{Important lemmas}

\begin{lemma}
\label{lemma:1}
Assuming Condition \ref{con:four} and $\delta_0\leq \frac{1}{100}, \delta_0R^2\geq D_1(\frac{\log^2 n}{n^2}+\gamma+ \lambda^2 s)$, then on $\mathcal{T}_1(\delta_0, R) \cap \mathcal{T}_2(\delta_0, R) \cap \mathcal{T}_3$, there exists $\bar{g}\in \mathcal{H}_n$ such that
\begin{itemize}
\item[(i)] ${\rm TV}(\bar{g}^{(k)})\leq a_k\cdot {\rm TV}(g_0^{(k)}), ~\|\bar{g}-g_0\|_{\infty}\leq  \frac{22b_k\log n}{\ell_z n} \cdot  {\rm TV}(g_0^{(k)})$, where $a_k,b_k>0$ are constants only dependent on $k$.
\item[(ii)] $\tau_{\delta_0,R}(\bx'(\hat{\bbeta}-\bbeta^0)+\hat{g}(z)-\bar{g}(z))\leq R$.

\end{itemize}
\end{lemma}
\begin{proof}
According to Lemma 16 in \cite{tibshirani2022divided} (see also Lemma 13 in \cite{sadhanala2019additive}), there exists $\bar{g}\in \mathcal{H}_n$ such that result (i) holds on $\mathcal{T}_3$. We use such a $\bar{g}$ in the rest of the proof. Consider the convex combination 
\[
\tilde{\bbeta}=t\hat{\bbeta}+(1-t)\bbeta^0, \quad \tilde{g}=t\hat{g}+(1-t)\bar{g}, \quad t\in [0,1].
\]
Accordingly, define 
\[
 \hat{f}(\bx,z)=\bx'\hat{\bbeta}+\hat{g}(z), ~~\bar{f}(\bx,z)=\bx'\bbeta^0+\bar{g}(z), ~~\tilde{f}(\bx,z)=t\hat{f}+(1-t)\bar{f}=\bx'\tilde{\bbeta}+\tilde{g}(z). 
\]
For the choice of $t=\frac{R}{R+\tau_{\delta_0,R}(\hat{f}-\bar{f})}$, it is straightforward to verify that
\begin{align}
\label{trick:0}
\tau_{\delta_0,R}(\tilde{f}-\bar{f})=\frac{R\cdot \tau_{\delta_0,R}(\hat{f}-\bar{f})}{R+ \tau_{\delta_0,R}(\hat{f}-\bar{f})} \leq R.
\end{align}
Hence, to prove result (ii) it is sufficient to show $\tau_{\delta_0,R}(\tilde{f}-\bar{f})\leq \frac{R}{2}$. We start with the basic inequality due to the convex problem \eqref{est:p},
\begin{align*}
\frac{1}{2}\|\by-\bX\tilde{\bbeta}-\tilde{g}(\bZ)\|_n^2+\lambda\|\tilde{\bbeta}\|_1+\gamma {\rm TV}(\tilde{g}^{(k)}) \leq \frac{1}{2}\|\by-\bX\bbeta^0-\bar{g}(\bZ)\|_n^2+\lambda\|\bbeta^0\|_1+\gamma {\rm TV}(\bar{g}^{(k)}),
\end{align*}
which can be simplified to
\begin{align*}
&\frac{1}{2}\|\bX(\bbeta^0-\tilde{\bbeta})+g_0(\bZ)-\tilde{g}(\bZ)\|_n^2+\lambda\|\tilde{\bbeta}\|_1+\gamma {\rm TV}(\tilde{g}^{(k)}) \nonumber \\
\leq & ~\frac{1}{2}\|g_0(\bZ)-\bar{g}(\bZ)\|_n^2+ \langle \beps, \bX(\tilde{\bbeta}-\bbeta^0)+\tilde{g}(\bZ)-\bar{g}(\bZ)\rangle_n+ \lambda\|\bbeta^0\|_1+\gamma {\rm TV}(\bar{g}^{(k)}).
\end{align*}
Using the triangle inequality for $\|\cdot\|_1$ and ${\rm TV}(\cdot)$, we can continue from the above to obtain
\begin{align}
\label{L2:bound}
&\frac{1}{2}\|\bX(\bbeta^0-\tilde{\bbeta})+g_0(\bZ)-\tilde{g}(\bZ)\|_n^2+\lambda\|\tilde{\bbeta}-\bbeta^0\|_1+\gamma {\rm TV}(\tilde{g}^{(k)}-\bar{g}^{(k)}) \nonumber \\
\leq & ~\frac{1}{2}\|g_0(\bZ)-\bar{g}(\bZ)\|_n^2+ \langle \beps, \bX(\tilde{\bbeta}-\bbeta^0)+\tilde{g}(\bZ)-\bar{g}(\bZ)\rangle_n+ 2\lambda\|\tilde{\bbeta}_S-\bbeta^0_S\|_1+2\gamma {\rm TV}(\bar{g}^{(k)}).
\end{align}
Observing from \eqref{trick:0} that $\tilde{f}-\bar{f}=\bx'(\tilde{\bbeta}-\bbeta^0)+\tilde{g}(z)-\bar{g}(z)$ belongs to the set $\{f:\tau_{\delta_0,R}(f)\leq R\}$, on $\mathcal{T}_1(\delta_0, R) \cap \mathcal{T}_2(\delta_0, R) \cap \mathcal{T}_3$ we can proceed with
\begin{align}
& \|\tilde{f}-\bar{f}\|^2 +  4\lambda\|\tilde{\bbeta}-\bbeta^0\|_1+4\gamma {\rm TV}(\tilde{g}^{(k)}-\bar{g}^{(k)}) \nonumber  \\
\leq &~ \|\bX(\bbeta^0-\tilde{\bbeta})+\bar{g}(\bZ)-\tilde{g}(\bZ)\|_n^2+\delta_0 R^2 + 4\lambda\|\tilde{\bbeta}-\bbeta^0\|_1+4\gamma {\rm TV}(\tilde{g}^{(k)}-\bar{g}^{(k)})\nonumber \\
 \leq & ~2 \|\bX(\bbeta^0-\tilde{\bbeta})+g_0(\bZ)-\tilde{g}(\bZ)\|_n^2+ 2\|\bar{g}(\bZ)-g_0(\bZ)\|_n^2+\delta_0 R^2+ 4\lambda\|\tilde{\bbeta}-\bbeta^0\|_1 \nonumber\\
 &~+4\gamma {\rm TV}(\tilde{g}^{(k)}-\bar{g}^{(k)})\nonumber \\
\leq & ~4\|g_0(\bZ)-\bar{g}(\bZ)\|_n^2+4 \langle \beps, \bX(\tilde{\bbeta}-\bbeta^0)+\tilde{g}(\bZ)-\bar{g}(\bZ)\rangle_n+ 8\lambda\|\bbeta^0_S-\tilde{\bbeta}_S\|_1\nonumber  \\
&~+8\gamma {\rm TV}(\bar{g}^{(k)}) +\delta_0 R^2  \nonumber  \\
\leq & ~5\delta_0 R^2 + 4b_k^2L_g^2\frac{22^2\log^2n}{\ell_z^2n^2}+8a_kL_g\gamma+8\lambda \sqrt{s}\|\tilde{\bbeta}-\bbeta^0\|_2, \label{L2:bound2}
\end{align}
where the third inequality holds due to \eqref{L2:bound}, and in the last inequality we have used result (i) and Cauchy-Schwarz inequality for $\|\bbeta^0_S-\tilde{\bbeta}_S\|_1$. Now Condition \ref{con:four} implies that 
\begin{align*}
&8\lambda \sqrt{s}\|\tilde{\bbeta}-\bbeta^0\|_2\leq 8\lambda \sqrt{s}\Lambda^{-1/2}_{min} \|\tilde{\bx}'(\tilde{\bbeta}-\bbeta^0)\| \\
\leq &~32\lambda^2s \Lambda^{-1}_{min}+\frac{1}{2}\|\tilde{\bx}'(\tilde{\bbeta}-\bbeta^0)\|^2 \leq 32\lambda^2s \Lambda^{-1}_{min}+\frac{1}{2}\|\tilde{f}-\bar{f}\|^2.
\end{align*}
Here, the second inequality follows $ab\leq \frac{1}{2}a^2+\frac{1}{2}b^2$, and the third inequality is due to the orthogonal decomposition $\|\bx'\bbeta+g(z)\|^2=\|\tilde{\bx}'\bbeta\|^2+\|\bh(z)'\bbeta+g(z)\|^2$. This together with \eqref{L2:bound2} yields
\begin{align*}
 & \|\tilde{f}-\bar{f}\|^2 + 8\lambda\|\tilde{\bbeta}-\bbeta^0\|_1+8\gamma {\rm TV}(\tilde{g}^{(k)}-\bar{g}^{(k)}) \\
  \leq &~10\delta_0 R^2 + 8b_k^2L_g^2\frac{22^2\log^2n}{\ell_z^2n^2}+16a_kL_g\gamma+ 64\lambda^2 s \Lambda^{-1}_{min} \leq 16\delta_0R^2,
\end{align*}
where the last inequality holds under the assumed condition on $\delta_0R^2$ (with proper choice of constant $D_1$). The above bound further implies $ \|\tilde{f}-\bar{f}\|\leq 4\sqrt{\delta_0}R\leq \frac{2}{5}R$ since $\delta_0\leq \frac{1}{100}$, and $\frac{\lambda\|\tilde{\bbeta}-\bbeta^0\|_1+\gamma {\rm TV}(\tilde{g}^{(k)}-\bar{g}^{(k)})}{20\delta_0R}\leq \frac{1}{10}R$, leading to the bound $\tau_{\delta_0,R}(\tilde{f}-\bar{f})\leq \frac{R}{10}+\frac{2R}{5}=\frac{R}{2}$.
\end{proof}

\begin{lemma}
\label{lemma:2}
Under the same conditions of Lemma \ref{lemma:1}, then on $\mathcal{T}_1(\delta_0, R) \cap \mathcal{T}_2(\delta_0, R) \cap \mathcal{T}_3$, it holds that 
\[
\|\hat{g}(z)-g_0(z)\|\leq \Big(1+\sqrt{\frac{\Lambda_{max}}{\Lambda_{min}}}\Big)R+22b_kL_g\ell_z^{-1}\frac{\log n}{n}.
\]
Moreover, define a subevent of $\mathcal{T}_1(\delta_0, R)$: 
\[
\mathcal{T}_0(\delta_0, R,\tilde{R})=\Bigg\{\sup_{\tau_{\delta_0,R}(f)\leq \tilde{R}}\Big|\|f\|_n^2-\|f\|^2\Big| \leq \delta_0 R^2\Bigg\}
\]
with $\tilde{R}=(2+\sqrt{\frac{\Lambda_{max}}{\Lambda_{min}}})R+22b_kL_g\ell_z^{-1}\frac{\log n}{n}+\frac{(a_k+1)L_g\gamma}{20\delta_0R}$. Then on $\mathcal{T}_0(\delta_0, R, \tilde{R}) \cap \mathcal{T}_2(\delta_0, R) \cap \mathcal{T}_3$, it holds that 
\[
\|\hat{g}(z)-g_0(z)\|_n \leq \Big(1+\sqrt{\delta_0}+\sqrt{\frac{\Lambda_{max}}{\Lambda_{min}}}\Big)R+22b_kL_g\ell_z^{-1}\frac{\log n}{n}.
\]
\end{lemma}

\begin{proof}
Part (ii) in Lemma \ref{lemma:1} together with Part (i) of Lemma \ref{basics} implies that $\|\hat{g}(z)-\bar{g}(z)\|  \leq (1+\sqrt{\frac{\Lambda_{max}}{\Lambda_{min}}})R$. As a result, combining it with Part (i) of Lemma \ref{lemma:1} gives
\[
\|\hat{g}(z)-g_0(z)\|\leq \|\hat{g}(z)-\bar{g}(z)\|+\|\bar{g}(z)-g_0(z)\|\leq  \Big(1+\sqrt{\frac{\Lambda_{max}}{\Lambda_{min}}}\Big)R+\frac{22b_k L_g \log n}{\ell_z n}. 
\]
To prove the second result, we first use results (i)(ii) of Lemma \ref{lemma:1} to bound
\begin{align*}
\frac{\gamma {\rm TV}(\hat{g}^{(k)}-g_0^{(k)})}{20\delta_0 R} \leq \frac{\gamma {\rm TV}(\hat{g}^{(k)}-\bar{g}^{(k)})}{20\delta_0 R}+\frac{\gamma {\rm TV}(\bar{g}^{(k)}-g_0^{(k)})}{20\delta_0 R} \leq R+\frac{(a_k+1)L_g\gamma}{20\delta_0R}.
\end{align*}
Combining the above with the first result on $\|\hat{g}(z)-g_0(z)\|$, we see that
\[
\tau_{\delta_0,R}(\hat{g}-g_0)=\frac{\gamma {\rm TV}(\hat{g}^{(k)}-g_0^{(k)})}{20\delta_0 R}+ \|\hat{g}(z)-g_0(z)\| \leq \tilde{R}. 
\]
Hence, on $\mathcal{T}_0(\delta_0, R, \tilde{R}) \cap \mathcal{T}_2(\delta_0, R) \cap \mathcal{T}_3$, we can use the bound on $\|\hat{g}(z)-g_0(z)\|$ to obtain the bound on $\|\hat{g}(z)-g_0(z)\|_n$:
\begin{align*}
\|\hat{g}(z)-g_0(z)\|_n\leq \sqrt{\|\hat{g}(z)-g_0(z)\|^2+\delta_0R^2}\leq \|\hat{g}(z)-g_0(z)\|+\sqrt{\delta_0}R.
\end{align*}
\end{proof}

\begin{lemma}
\label{hg:one}
Assume Conditions \ref{con:one}, \ref{con:three}, \ref{con:four} hold. Define 
\begin{align}
\label{kf:def}
K_{\mathcal{F}}=(1\vee C_{z,k})\Big(\frac{20\delta_0 R}{\gamma}+1+\sqrt{\frac{\Lambda_{max}}{\Lambda_{min}}}\Big)R, 
\end{align}
where $C_{z,k}>0$ is the constant in Lemma \ref{basics}. Suppose that
\begin{align*}
&\frac{\log p}{n}\leq 1,~~\delta_0\sqrt{\frac{\log p}{n}}\frac{R^2}{\lambda^2}\leq D_1,~~D_2 \delta_0^{\frac{-4k-4}{2k+3}}K_{\mathcal{F}}^2n^{\frac{-2k-2}{2k+3}} \leq R^2, \\
&D_3(1+K_{\mathcal{F}}^2)R^2\log R^{-1}\leq n\lambda^2, ~~D_4(\lambda^2 nR^{\frac{-2k-1}{k+1}}K_{\mathcal{F}}^{\frac{-1}{k+1}} \wedge \lambda \sqrt{n}R^{-1}) \geq \log n.
\end{align*}
Then, 
\begin{align*}
\mathbb{P}(\mathcal{T}_1(\delta_0, R))\geq 1-2p^{-10}-8pe^{-D_4(\lambda^2 nR^{\frac{-2k-1}{k+1}}K_{\mathcal{F}}^{\frac{-1}{k+1}} \wedge \lambda \sqrt{n}R^{-1})}-8pe^{-\frac{\delta_0^2R^2n}{D_5K_{\mathcal{F}}^2}}. 
\end{align*}
Further assume $\frac{1}{2}\delta_0R^2\geq D_1 \gamma, R\geq \frac{\log n}{n}$ where this constant $D_1$ is the one from Lemma \ref{lemma:1}. It holds that
\begin{align*}
\mathbb{P}(\mathcal{T}_0(\delta_0, R,\tilde{R}))\geq 1-2p^{-10}-8pe^{-D_4(\lambda^2 nR^{\frac{-2k-1}{k+1}}K_{\mathcal{F}}^{\frac{-1}{k+1}} \wedge \lambda \sqrt{n}R^{-1})}-8pe^{-\frac{\delta_0^2R^2n}{D_5K_{\mathcal{F}}^2}}. 
\end{align*}
\end{lemma}

\begin{proof}
For a given function $f(\bx,z)=\bx'\bbeta+g(z)$, it is clear that
\begin{align*}
\|f\|_n^2-\|f\|^2&=\frac{1}{n}\sum_{i=1}^n((\bx_i'\bbeta)^2-\|\bx'\bbeta\|^2)+\frac{1}{n}\sum_{i=1}^n(g^2(z_i)-\|g(z)\|^2)\\
&~~~~+\frac{2}{n}\sum_{i=1}^n(\bx_i'\bbeta g(z_i)-\mathbb{E}\bx'\bbeta g(z)),
\end{align*}
such that 
\begin{align*}
\sup_{\tau_{\delta_0,R}(f)\leq R}\Big|\|f\|_n^2-\|f\|^2\Big| \leq&  \sup_{\tau_{\delta_0,R}(f)\leq R} \Big| \frac{1}{n}\sum_{i=1}^n((\bx_i'\bbeta)^2-\|\bx'\bbeta\|^2)\Big |\\
&+ \sup_{\tau_{\delta_0,R}(f)\leq R} \Big| \frac{1}{n}\sum_{i=1}^n(g^2(z_i)-\|g(z)\|^2) \Big |\\
&+ \sup_{\tau_{\delta_0,R}(f)\leq R} \Big| \frac{2}{n}\sum_{i=1}^n(\bx_i'\bbeta g(z_i)-\mathbb{E}\bx'\bbeta g(z)) \Big |:= I+II+III.
\end{align*}
We now bound the above three terms separately. According to Lemma \ref{basics} Part (i), the $\bbeta$ in any $f(\bx,z)=\bx'\bbeta+g(z)$ with $\tau_{\delta_0,R}(f)\leq R$ satisfies
\begin{align}
\label{beta:def}
\bbeta \in \mathcal{B}:=\Big\{\bbeta: \|\bbeta\|_1\leq  20\delta_0R^2\lambda^{-1}, \|\bx'\bbeta\|\leq (2+\sqrt{\frac{\Lambda_{max}}{\Lambda_{min}}})R \Big\}. 
\end{align}
Therefore,
\begin{align}
\label{term1:key}
I\leq  \sup_{\bbeta\in \mathcal{B}} \Big| \frac{1}{n}\sum_{i=1}^n((\bx_i'\bbeta)^2-\|\bx'\bbeta\|^2)\Big |\leq \max_{1\leq j,k \leq p} \Big|\frac{1}{n}\sum_{i=1}^n(x_{ij}x_{ik}-\mathbb{E}x_jx_k)\Big| \cdot \sup_{\bbeta\in \mathcal{B}}\|\bbeta\|^2_1,
\end{align}
where in the last inequality we have used the fact $|\ba'\bA\ba| \leq \|\ba\|^2_1\cdot \|\bA\|_{\max}, \forall \ba\in \mathbb{R}^p, \bA\in \mathbb{R}^{p\times p}$. Condition \ref{con:one} implies that $\{x_{ij}x_{ik}-\mathbb{E}x_jx_k\}_{i=1}^n$ are independent, zero-mean, sub-exponential random variables with the sub-exponential norm $\|x_{ij}x_{ik}-\mathbb{E}x_jx_k\|_{\psi_1}\leq C_1K_x^2$. We then use Bernstein's inequality in Theorem \ref{hoeffding:quote} together with a simple union bound to obtain that when $\frac{\log p}{n}\leq 1$,
\[
\max_{1\leq j,k \leq p} \Big|\frac{1}{n}\sum_{i=1}^n(x_{ij}x_{ik}-\mathbb{E}x_jx_k)\Big| \leq C_2K_x^2\sqrt{\frac{\log p}{n}}
\]
holds with probability at least $1-2p^{-10}$. Putting this result together with \eqref{term1:key} enables us to conclude $\mathbb{P}(I\leq \delta_0R^2/3)\geq 1-2p^{-10}$, as long as $\delta_0\sqrt{\frac{\log p}{n}}\frac{R^2}{\lambda^2}\leq D_1$.

Now we bound $II$. Using Lemma \ref{basics} we have $II\leq \sup_{g\in \mathcal{F}} \big|\|g\|_n^2-\|g\|^2\big|$, where
\begin{align}
&\mathcal{F}:=\Big\{g: \|g\| \leq R_{\mathcal{F}}, {\rm TV}(g^{(k)})\leq 20\delta_0R^2\gamma^{-1}, \|g\|_{\infty}\leq K_{\mathcal{F}} \Big\}, \label{fset:def}\\
&R_{\mathcal{F}}=\Big(1+\sqrt{\frac{\Lambda_{max}}{\Lambda_{min}}}\Big)R, \quad K_{\mathcal{F}}=(1 \vee C_{z,k})\Big(\frac{20\delta_0 R}{\gamma}+1+\sqrt{\frac{\Lambda_{max}}{\Lambda_{min}}}\Big)R. \nonumber
\end{align}
We aim to apply Theorem \ref{emp:norm} to bound $\sup_{g\in \mathcal{F}} \big|\|g\|_n^2-\|g\|^2\big|$. Adopting the notation in Theorem \ref{emp:norm}, we first calculate $\mathcal{J}_0(t,\mathcal{F})$:
\begin{align}
\label{J0:term}
\mathcal{J}_0(t,\mathcal{F})\leq C_3t\int_0^1 \sqrt{C_4\Big(\frac{ut}{2K_{\mathcal{F}}}\Big)^{-\frac{1}{k+1}}} du=\underbrace{C_3\sqrt{C_4}K_{\mathcal{F}}^{\frac{1}{2k+2}}\frac{2^{\frac{1}{2k+2}}(2k+2)}{2k+1}}_{\bar{C}_3}\cdot t^{\frac{2k+1}{2k+2}},
\end{align}
where the first inequality is due to the entropy bound of Corollary 1 in \cite{sadhanala2019additive}. As a result, 
\begin{align*}
H(u)&=\sup_{t\geq 0} \big(ut-(\mathcal{J}^{-1}_0(t,\mathcal{F}))^2\big)\leq \sup_{t\geq 0} \big(ut-\bar{C}_3^{-\frac{4k+4}{2k+1}}t^{\frac{4k+4}{2k+1}} \big) \\
&=\Big(\frac{2k+1}{4k+4}\Big)^{\frac{4k+4}{2k+3}}\frac{2k+3}{2k+1}\bar{C}_3^{\frac{4k+4}{2k+3}}u^{\frac{4k+4}{2k+3}}=D_2K_{\mathcal{F}}^{\frac{2}{2k+3}}u^{\frac{4k+4}{2k+3}}.
\end{align*}
We are ready to invoke Theorem \ref{emp:norm}. The condition $R_{\mathcal{F}}^2\geq H(4K_{\mathcal{F}}/\sqrt{n})$ is reduced to 
\begin{align}
\label{init:con}
R^2\geq D_3 K_{\mathcal{F}}^2 n^{\frac{-2k-2}{2k+3}},
\end{align}
and under this condition it holds with probability at least $1-e^{-t}$ that
\begin{align*}
\sup_{g\in \mathcal{F}} \Big|\|g\|_n^2-\|g\|^2\Big| \leq D_4\cdot \Big(\frac{ K_{\mathcal{F}}^{\frac{2k+3}{2k+2}}R^{\frac{2k+1}{2k+2}}}{\sqrt{n}}+\frac{ K_{\mathcal{F}} R \sqrt{t}}{\sqrt{n}}+\frac{ K_{\mathcal{F}}^2 t}{n}\Big).
\end{align*}
Given the above, choosing $t=\frac{\delta_0^2R^2n}{D_5K_{\mathcal{F}}^2}$ and assuming a slightly stronger condition $R^2\geq D_6\delta_0^{\frac{-4k-4}{2k+3}} K_{\mathcal{F}}^2 n^{\frac{-2k-2}{2k+3}}$ compared to \eqref{init:con}, it is then straightforward to verify that $\mathbb{P}(II\leq \delta_0R^2/3)\geq 1-e^{-\frac{\delta_0^2R^2n}{D_5K_{\mathcal{F}}^2}}$.

Next we bound $III$. We first use H\"{o}lder's inequality and Lemma \ref{basics} to obtain
\begin{align*}
III & \leq \sup_{\tau_{\delta_0,R}(f)\leq R} 2\|\beta\|_1\cdot \Big\|\frac{1}{n}\sum_{i=1}^n(\bx_ig(z_i)-\mathbb{E}\bx g(z))\Big\|_{\infty} \\
& \leq \frac{40\delta_0R^2}{\lambda} \cdot \sup_{g\in \mathcal{F}, 1\leq j \leq p} \Big|\frac{1}{n}\sum_{i=1}^n(x_{ij}g(z_i)-\mathbb{E}x_jg(z))\Big|.
\end{align*}
We then aim to use Theorem \ref{emp:sc} to bound $\sup_{g\in \mathcal{F}}|\frac{1}{n}\sum_{i=1}^n(x_{ij}g(z_i)-\mathbb{E}x_jg(z))|$ for each $1\leq j \leq p$. The quantities $Q_{ \mathcal{F}}, \hat{Q}_{ \mathcal{F}}$ in Theorem \ref{emp:sc} become $Q_{ \mathcal{F}}=\sup_{g\in  \mathcal{F}}\sqrt{\mathbb{E}(x_j^2g^2(z))}, \hat{Q}_{ \mathcal{F}}=\sup_{g\in  \mathcal{F}}\sqrt{\frac{1}{n}\sum_{i=1}^nx^2_{ij}g^2(z_i)}$. Let $\hat{R}_{ \mathcal{F}}=\sup_{g\in  \mathcal{F}}\|g\|_n, \mathbf{X}_j=(x_{1j},\ldots, x_{nj}), \mathbf{Z}=(z_1,\ldots, z_n)$, and $\epsilon_1,\ldots, \epsilon_n$ be a Rademacher sequence independent of $\mathbf{X}_j,\mathbf{Z}$.  Applying Dudley's entropy integral gives
\begin{align}
&\mathbb{E}\Big(\sup_{g\in \mathcal{F}}|\frac{1}{n}\sum_{i=1}^n\epsilon_ix_{ij}g(z_i)| \Big| \mathbf{X}_j, \mathbf{Z}\Big) \leq C_5 \frac{\hat{Q}_{ \mathcal{F}}}{\sqrt{n}}\int_0^1 \sqrt{\mathcal{H}(u \hat{Q}_{ \mathcal{F}}, \mathcal{F}, \|\cdot\|^w_n)} du, \nonumber \\
\leq&   C_5 \frac{\hat{Q}_{ \mathcal{F}}}{\sqrt{n}}\int_0^1 \sqrt{\mathcal{H}(u \hat{Q}_{ \mathcal{F}}\|\mathbf{X}_j\|_{\infty}^{-1}, \mathcal{F}, \|\cdot\|_n)} du\leq \frac{C_5\|\mathbf{X}_j\|_{\infty}}{2C_3\sqrt{n}}\mathcal{J}_0(2\hat{Q}_{ \mathcal{F}}\|\mathbf{X}_j\|_{\infty}^{-1}, \mathcal{F}), \nonumber \\
\leq & D_7 n^{-1/2} K_{\mathcal{F}}^{\frac{1}{2k+2}}\hat{Q}_{ \mathcal{F}}^{\frac{2k+1}{2k+2}}\|\mathbf{X}_j\|_{\infty}^{\frac{1}{2k+2}}\leq D_7 n^{-1/2} K_{\mathcal{F}}^{\frac{1}{2k+2}}\hat{R}_{ \mathcal{F}}^{\frac{2k+1}{2k+2}}\|\mathbf{X}_j\|_{\infty}, \label{ddy:1}
\end{align}
where $\mathcal{J}_0(\cdot, \mathcal{F})$ is from \eqref{J0:term} and the second to last inequality is due to \eqref{J0:term}; $\|\cdot\|^w_n$ denotes the weighted empirical norm $\|g\|^w_n=\sqrt{\frac{1}{n}\sum_{i=1}^nx^2_{ij}g^2(z_i)}$ and it is clear that $\|g\|^w_n\leq \|\mathbf{X}_j\|_{\infty} \|g\|_n$ and $ \hat{Q}_{ \mathcal{F}}\leq \|\mathbf{X}_j\|_{\infty} \hat{R}_{ \mathcal{F}}$. Now applying the second result in Theorem \ref{emp:sc} yields 
\begin{align}
\label{thm4:1}
\mathbb{P}\Big(\sup_{g\in \mathcal{F}}|\frac{1}{n}\sum_{i=1}^n\epsilon_ix_{ij}g(z_i)|\geq C_6 \|\mathbf{X}_j\|_{\infty} \Big(D_7 n^{-1/2} K_{\mathcal{F}}^{\frac{1}{2k+2}}\hat{R}_{ \mathcal{F}}^{\frac{2k+1}{2k+2}}+\hat{R}_{ \mathcal{F}}\sqrt{t/n}\Big)\Big|  \mathbf{X}_j, \mathbf{Z}\Big) \leq e^{-t}.
\end{align}
Define the events 
\[
\mathcal{A}_1=\Big\{\|\mathbf{X}_j\|_{\infty} \leq C_7K_x(\sqrt{\log n}+\sqrt{t})\Big\},~~ \mathcal{A}_2=\Big\{\sup_{g\in \mathcal{F}} |\|g\|_n^2-\|g\|^2|\leq \delta_0R^2/3\Big\}.
\]
Condition \ref{con:one} together with standard union bound for sub-Gaussian tails gives $\mathbb{P}(\mathcal{A}_1^c)\leq e^{-t}$ (with a proper choice of $C_7$). The previous result on term $II$ shows $\mathbb{P}(\mathcal{A}_2^c)\leq e^{-\frac{\delta_0^2R^2n}{D_5K_{\mathcal{F}}^2}}$, and it holds on $\mathcal{A}_2$ that 
\begin{align}
\hat{R}_{ \mathcal{F}}\leq \sqrt{\sup_{g\in \mathcal{F}} |\|g\|_n^2-\|g\|^2|+\sup_{g\in \mathcal{F}} \|g\|^2}\leq \sqrt{\delta_0/3+(1+\sqrt{\Lambda_{max}/\Lambda_{min}})^2}R. \label{hat:R}
\end{align}
 As a result, we can further bound \eqref{thm4:1} on $\mathcal{A}_1\cap \mathcal{A}_2$ to arrive at
\begin{align*}
\mathbb{P}\Big(\sup_{g\in \mathcal{F}}|\frac{1}{n}\sum_{i=1}^n\epsilon_ix_{ij}g(z_i)|\geq D_8(\sqrt{\log n}+\sqrt{t}) \Big(n^{-1/2} K_{\mathcal{F}}^{\frac{1}{2k+2}}R^{\frac{2k+1}{2k+2}}+R\sqrt{t/n}\Big) \Big|  \mathbf{X}_j, \mathbf{Z} \Big)\cdot \mathbbm{1}_{\mathcal{A}_1\cap \mathcal{A}_2} \leq e^{-t}
\end{align*}
We integrate out the above conditional probability to obtain the unconditional one:
\begin{align*}
&\mathbb{P}\Big(\sup_{g\in \mathcal{F}}|\frac{1}{n}\sum_{i=1}^n\epsilon_ix_{ij}g(z_i)|\geq D_8(\sqrt{\log n}+\sqrt{t}) \big( n^{-1/2} K_{\mathcal{F}}^{\frac{1}{2k+2}}R^{\frac{2k+1}{2k+2}}+R\sqrt{t/n}\big) \Big) \\
\leq & ~e^{-t}+\mathbb{P}(\mathcal{A}_1^c\cup \mathcal{A}_2^c)\leq 2e^{-t}+e^{-\frac{\delta_0^2R^2n}{D_5K_{\mathcal{F}}^2}}.
\end{align*}
It is then direct to verify that choosing $t= D_9(\lambda^2 nR^{\frac{-2k-1}{k+1}}K_{\mathcal{F}}^{\frac{-1}{k+1}} \wedge \lambda \sqrt{n}R^{-1})\geq \log n $ gives
\begin{align}
\label{sym:form}
\mathbb{P}\Big(\sup_{g\in \mathcal{F}}|\frac{1}{n}\sum_{i=1}^n\epsilon_ix_{ij}g(z_i)|\geq \frac{\lambda}{480}\Big)\leq 2e^{-D_9(\lambda^2 nR^{\frac{-2k-1}{k+1}}K_{\mathcal{F}}^{\frac{-1}{k+1}} \wedge \lambda \sqrt{n}R^{-1})}+e^{-\frac{\delta_0^2R^2n}{D_5K_{\mathcal{F}}^2}}.
\end{align}
We will invoke the symmetrization result of Theorem \ref{emp:sc} to transfer the above bound to $\sup_{g\in \mathcal{F}}|\frac{1}{n}\sum_{i=1}^n(x_{ij}g(z_i)-\mathbb{E}x_jg(z))|$. In order to do so, we need to show $8\times 480^2 Q_{\mathcal{F}}^2\leq n \lambda^2$. Note that 
\begin{align*}
Q_{\mathcal{F}}^2&=\sup_{g\in \mathcal{F}}\mathbb{E}x_j^2g^2(z)=\sup_{g\in \mathcal{F}}\big(\mathbb{E}x_j^2g^2(z)\mathbbm{1}(|x_j|\leq s)+\mathbb{E}x_j^2g^2(z)\mathbbm{1}(|x_j|> s)\big) \\
&\leq  s^2 \sup_{g\in \mathcal{F}} \|g\|^2 + \sup_{g\in \mathcal{F}}\|g\|_{\infty}^2 \cdot  \mathbb{E}x_j^2\mathbbm{1}(|x_j|>s) \\
&\leq s^2\Big(1+\sqrt{\frac{\Lambda_{max}}{\Lambda_{miin}}}\Big)^2R^2+K^2_{\mathcal{F}} \cdot \mathbb{E}x_j^2\mathbbm{1}(|x_j|>s).
\end{align*}
Using the sub-Gaussian tail bound for $x_j$ to bound $\mathbb{E}x_j^2\mathbf{1}(|x_j|>s)$, it is not hard to obtain that $Q_{\mathcal{F}}^2\leq D_{10}(1+ K^2_{\mathcal{F}})R^2\log R^{-1}$ with the choice $s=D_{11} \sqrt{\log R^{-1}}$. Therefore, under the condition $D_{12}(1+ K^2_{\mathcal{F}})R^2\log R^{-1}\leq n\lambda^2$, Theorem \ref{emp:sc} together with \eqref{sym:form} leads to 
\[
\mathbb{P}\Big(\sup_{g\in \mathcal{F}}|\frac{1}{n}\sum_{i=1}^n(x_{ij}g(z_i)-\mathbb{E}x_jg(z))| \geq\frac{\lambda}{120}\Big)\leq 8e^{-D_9(\lambda^2 nR^{\frac{-2k-1}{k+1}}K_{\mathcal{F}}^{\frac{-1}{k+1}} \wedge \lambda \sqrt{n}R^{-1})}+4e^{-\frac{\delta_0^2R^2n}{D_5K_{\mathcal{F}}^2}}.
\]
We follow via a union bound,
\begin{align*}
&\mathbb{P}(III \geq \delta_0R^2/3)\leq \mathbb{P}\Big(\sup_{g\in \mathcal{F}, 1\leq j \leq p} \Big|\frac{1}{n}\sum_{i=1}^n(x_{ij}g(z_i)-\mathbb{E}x_jg(z))\Big|\geq \frac{\lambda}{120}\Big) \\
\leq& ~p \cdot \max_{1\leq j \leq p} \mathbb{P}\Big(\sup_{g\in \mathcal{F}}|\frac{1}{n}\sum_{i=1}^n(x_{ij}g(z_i)-\mathbb{E}x_jg(z))| \geq\frac{\lambda}{120}\Big) \\
\leq &~8pe^{-D_9(\lambda^2 nR^{\frac{-2k-1}{k+1}}K_{\mathcal{F}}^{\frac{-1}{k+1}} \wedge \lambda \sqrt{n}R^{-1})}+4pe^{-\frac{\delta_0^2R^2n}{D_5K_{\mathcal{F}}^2}}.
\end{align*}
Finally, we collect the results on $I,II,III$ to bound $\mathbb{P}(\mathcal{T}_1(\delta_0, R))$ via $\mathbb{P}(\mathcal{T}_1(\delta_0, R))\geq 1-\mathbb{P}(I\geq \delta_0R^2/3)-\mathbb{P}(II\geq \delta_0R^2/3)-\mathbb{P}(III\geq \delta_0R^2/3)$.

Regarding the bound for $\mathbb{P}(\mathcal{T}_0(\delta_0, R,\tilde{R}))$, first note that $2R\leq \tilde{R}\leq D_{12}R$ under the assumption $\frac{1}{2}\delta_0R^2\geq D_1 \gamma, R\geq \frac{\log n}{n}$. Moreover, $\{f:\tau_{\delta_0,R}(f)\leq \tilde{R}\}\subset \{f:\tau_{\delta_0,\tilde{R}}(f)\leq \tilde{R}\}$. Hence,
\[
\mathbb{P}(\mathcal{T}_0(\delta_0, R,\tilde{R})) \geq \mathbb{P}\Big(\sup_{\tau_{\delta_0,\tilde{R}}(f)\leq \tilde{R}}\Big|\|f\|_n^2-\|f\|^2\Big| \leq D_{12}^{-2}\delta_0 \tilde{R}^2\Big).
\]
Essentially, all the previous arguments in bounding $\mathbb{P}(\mathcal{T}_1(\delta_0, R))$ carry over into $\mathbb{P}(\mathcal{T}_0(\delta_0, R,\tilde{R}))$, up to constants $D_j$'s. We only need to update those constants in the conditions and results, in a way so that $\mathbb{P}(\mathcal{T}_1(\delta_0, R))$ and $\mathbb{P}(\mathcal{T}_0(\delta_0, R,\tilde{R}))$ share the same constants (e.g., taking a minimum or maximum among constants).
\end{proof}

\begin{lemma}
\label{hp:event2}
Assume Conditions \ref{con:one}-\ref{con:four} hold.
\begin{itemize}
\item[(i)] When $R^2\geq D_1 \delta_0^{\frac{-4k-4}{2k+3}} n^{\frac{-2k-2}{2k+3}}(K_{\mathcal{F}}^2+K_{\mathcal{F}}^{\frac{2}{2k+3}}\sigma^{\frac{4k+4}{2k+3}})$ where $K_{\mathcal{F}}$ is defined in \eqref{kf:def} of Lemma \ref{hg:one}, then 
\begin{align*}
\mathbb{P}\Big( \mathcal{T}_2(\delta_0, R)\Big) \geq 1-2pe^{-C_1n\min\big(\frac{\lambda^2}{(K_x K_{\varepsilon}\sigma)^2}, \frac{\lambda}{K_xK_{\varepsilon}\sigma}\big)}-2e^{-D_2 \sigma^{-2}\delta_0^2nR^2}-e^{-\frac{\delta_0^2R^2n}{D_3 K_{\mathcal{F}}^2}}.
\end{align*}
\item[(ii)] As long as $\frac{11\log n}{\ell_z n}<1$, then $\mathbb{P}\big( \mathcal{T}_3 \big) \geq 1-2\ell_z n^{-10}$.
\end{itemize}
\end{lemma}

\begin{proof}
Recalling the sets $\mathcal{B}$ in \eqref{beta:def} and $\mathcal{F}$ in \eqref{fset:def} from the proof of Lemma \ref{hg:one}, we have
\begin{align*}
\sup_{\tau_{\delta_0,R}(f)\leq R}\Big|\langle \beps, f(\bX, \bZ)\rangle_n\Big| & \leq \sup_{\tau_{\delta_0,R}(f)\leq R} \Big|\frac{1}{n}\sum_{i=1}^n\varepsilon_i\bx_i'\bbeta\Big|+ \sup_{\tau_{\delta_0,R}(f)\leq R} \Big|\frac{1}{n}\sum_{i=1}^n\varepsilon_ig(z_i)\Big| \\
 & \leq  \sup_{\bbeta\in \mathcal{B}}\|\bbeta\|_1 \cdot   \Big\|\frac{1}{n}\sum_{i=1}^n\bx_i\varepsilon_i\Big\|_{\infty}+ \sup_{g\in \mathcal{F}}\Big|\frac{1}{n}\sum_{i=1}^n\varepsilon_ig(z_i)\Big| \\
 & \leq  \frac{20\delta_0R^2}{\lambda}\Big\|\frac{1}{n}\sum_{i=1}^n\bx_i\varepsilon_i\Big\|_{\infty}+ \sup_{g\in \mathcal{F}}\Big|\frac{1}{n}\sum_{i=1}^n\varepsilon_ig(z_i)\Big|:=I+II.
\end{align*}
We first bound term $I$. Using Bernstein's inequality in Theorem \ref{hoeffding:quote} we obtain that $\mathbb{P}(|\frac{1}{n}\sum_{i=1}^n\varepsilon_ix_{ij}|\geq \frac{\lambda}{40})\leq 2e^{-C_1n\min\big(\frac{\lambda^2}{(K_xK_{\varepsilon}\sigma)^2}, \frac{\lambda}{K_xK_{\varepsilon}\sigma}\big)}$. This combined with a union bound gives
\begin{align*}
\mathbb{P}(I\leq \delta_0R^2/2)\geq \mathbb{P}\Big(\Big\|\frac{1}{n}\sum_{i=1}^n\bx_i\varepsilon_i\Big\|_{\infty} \leq \frac{\lambda}{40}\Big)\geq 1-2pe^{-C_1n\min\big(\frac{\lambda^2}{(K_x K_{\varepsilon}\sigma)^2}, \frac{\lambda}{K_xK_{\varepsilon}\sigma}\big)}.
\end{align*}
To bound term $II$, let $\hat{R}_{\mathcal{F}}=\sup_{g\in \mathcal{F}}\|g\|_n$. Conditional on $z_1,\ldots,z_n$ (which are independent from $\varepsilon_1,\ldots, \varepsilon_n$), Theorem \ref{ddy:2} implies that the following holds with probability at least $1-2e^{-t}$:
\begin{align*}
II &\leq C_2K_{\varepsilon}\sigma\cdot \Bigg(\frac{\hat{R}_{\mathcal{F}}}{\sqrt{n}}\int_0^1\sqrt{\mathcal{H}(u\hat{R}_{\mathcal{F}},\mathcal{F},\|\cdot\|_n)}du+\hat{R}_{\mathcal{F}}\sqrt{\frac{t}{n}}\Bigg), \\
&\leq C_2K_{\varepsilon}\sigma\cdot \Big(D_1n^{-1/2}K_{\mathcal{F}}^{\frac{1}{2k+2}}\hat{R}_{\mathcal{F}}^{\frac{2k+1}{2k+2}}+\hat{R}_{\mathcal{F}}\sqrt{t/n}\Big),
\end{align*}
 where the last inequality follows as in \eqref{ddy:1}. Integrating over $z_1,\ldots, z_n$, the above holds marginally with probability at least $1-2e^{-t}$ as well. Moreover, from \eqref{hat:R} we already know that, 
\[
\mathbb{P}(\hat{R}_{\mathcal{F}}\leq D_3R)\geq 1-e^{-\frac{\delta_0^2R^2n}{D_4K_{\mathcal{F}}^2}}, \text{~~~when~} R^2\geq D_2\delta_0^{\frac{-4k-4}{2k+3}} K_{\mathcal{F}}^2 n^{\frac{-2k-2}{2k+3}}.
\]
Combining these two results via a union bound gives us
 \begin{align*}
 \mathbb{P}\Big(II\leq D_5 \sigma \Big(n^{-1/2}K_{\mathcal{F}}^{\frac{1}{2k+2}}R^{\frac{2k+1}{2k+2}}+R\sqrt{t/n}\Big)\Big) \geq 1-2e^{-t}-e^{-\frac{\delta_0^2R^2n}{D_4K_{\mathcal{F}}^2}}.
 \end{align*}
Choosing $t=D_6 \sigma^{-2}\delta_0^2nR^2$ and assuming $R^2\geq D_7 \delta_0^{\frac{-4k-4}{2k+3}} n^{\frac{-2k-2}{2k+3}}(K_{\mathcal{F}}^2+K_{\mathcal{F}}^{\frac{2}{2k+3}}\sigma^{\frac{4k+4}{2k+3}})$, it is direct to confirm
\[
\mathbb{P}(II\leq \delta_0R^2/2)\geq 1- 2e^{-D_6 \sigma^{-2}\delta_0^2nR^2}-e^{-\frac{\delta_0^2R^2n}{D_4K_{\mathcal{F}}^2}}.
\]
Combining the bounds for $I,II$ gives the bound for 
\[
\mathbb{P}( \mathcal{T}_2(\delta_0, R))\geq \mathbb{P}(I\leq \delta_0R^2/2, II\leq \delta_0R^2/2)\geq 1-\mathbb{P}(I > \delta_0R^2/2)-\mathbb{P}(II> \delta_0R^2/2).
\]
Finally, the bound on $\mathbb{P}\big( \mathcal{T}_3 \big)$ is directly taken from Lemma 5 in \cite{wang2014falling}. 
\end{proof}

\begin{lemma}
\label{basics}
Consider any $f(\bx,z)=\bx'\bbeta+g(z)$ that satisfies $\tau_{\delta_0,R}(f)\leq R$.
\begin{itemize}
\item[(i)] Under Condition \ref{con:four} it holds that
\begin{align*}
&\|\bbeta\|_2\leq \Lambda_{min}^{-1/2}R,~~ \|\bbeta\|_1\leq \frac{20\delta_0R^2}{\lambda},~~ \|\bx'\bbeta\|\leq \Big(2+\sqrt{\frac{\Lambda_{max}}{\Lambda_{min}}}\Big)R,\\
&\|g\|\leq \Big(1+\sqrt{\frac{\Lambda_{max}}{\Lambda_{min}}}\Big)R,~~ {\rm TV}(g^{(k)})\leq \frac{20\delta_0 R^2}{\gamma}.
\end{align*}
\item[(ii)] Under additional Condition \ref{con:three}, it holds that 
\[
\|g\|_{\infty}\leq C_{z,k}\Big(\frac{20\delta_0 R}{\gamma}+1+\sqrt{\frac{\Lambda_{max}}{\Lambda_{min}}}\Big)R, 
\]
where $C_{z,k}>0$ is a constant only dependent on $k$ and $\ell_z$.
\end{itemize}
\end{lemma}

\begin{proof}
The bound on $\|\bbeta\|_1$ and ${\rm TV}(g^{(k)})$ is clear from the definition of $\tau_{\delta_0,R}(f)$ in \eqref{funl:def}. Using the orthogonal decomposition: 
\begin{align*}
\|\bx'\bbeta+g(z)\|^2=\|\tilde{\bx}'\bbeta\|^2+\|\bh(z)'\bbeta+g(z)\|^2\leq R^2,
\end{align*}
we have $\|\bbeta\|_2\leq \Lambda_{min}^{-1/2}\|\tilde{\bx}'\bbeta\| \leq \Lambda_{min}^{-1/2}R$, and 
\[
\|g(z)\|\leq \|\bh(z)'\beta+g(z)\|+\|\bh(z)'\bbeta\|\leq R+\Lambda^{1/2}_{max}\|\bbeta\|_2\leq \Big(1+\sqrt{\frac{\Lambda_{max}}{\Lambda_{min}}}\Big)R. 
\]
Having the bound on $\|g\|$, we can further obtain $\|\bx'\bbeta\|\leq \|\bx'\bbeta+g(z)\|+\|g(z)\|\leq  (2+\sqrt{\frac{\Lambda_{max}}{\Lambda_{min}}})R$. 

It remains to bound $\|g\|_{\infty}$. Define $g_*(z)=\frac{g(z)}{(\frac{20\delta_0 R}{\gamma}+1+\sqrt{\frac{\Lambda_{max}}{\Lambda_{min}}})R}$, then ${\rm TV}(g_*^{(k)})\leq 1, \|g_*\|\leq 1$ due to the derived bounds on $\|g\|$ and ${\rm TV}(g^{(k)})$. Hence it is sufficient to show that 
\[
\sup_{g_*: {\rm TV}(g_*^{(k)})\leq 1, \|g_*\|\leq 1} \|g_*\|_{\infty}\leq C_{z, k}.
\]
To prove the above, we first decompose $g_*=p_*+q_*$, where $p_*$ is a polynomial of degree $k$ and $q_*$ is orthogonal to all polynomials of degree $k$ (with respect to the $L_2$ inner product $\int_0^1p(z)q(z)dz$). Note that ${\rm TV}(q_*^{(k)})={\rm TV}(g_*^{(k)})\leq 1$. Then Lemma 5 in \cite{sadhanala2019additive} implies that $\|q_*\|_{\infty}\leq c_k$ for a constant $c_k>0$ only depending on $k$. Now write $p_*(z)=\sum_{j=0}^ka_jz^j$. We have 
\[
\sqrt{\ba'\bV_0 \ba}=\|p_*\|\leq \|g_*\|+\|q_*\|\leq 1+ c_k,
\]
where $\ba=(a_0,\ldots, a_k)$ and $\bV_0=\mathbb{E}(\bv(z)\bv(z)')$ with $\bv(z)=(1,z,z^2,\ldots, z^k)$. Thus, we obtain 
\[
\|p_*\|_{\infty}\leq \|\ba\|_1\leq \sqrt{k+1}\|\ba\|_2\leq \sqrt{k+1}(1+c_k)\lambda_{min}^{-1/2}(\bV_0),
\] 
which implies $\|g_*\|_{\infty}\leq \|p_*\|_{\infty}+ \|q_*\|_{\infty}\leq \sqrt{k+1}(1+c_k)\lambda_{min}^{-1/2}(\bV_0)+c_k$. Finally, we need to show $\lambda_{min}(\bV_0)$ is bounded below by a constant only depending on $k$ and $\ell_z$. Under Condition \ref{con:three} we have
\begin{align*}
\lambda_{min}(\bV_0)&=\min_{\|\bb\|_2=1}\mathbb{E}(\bb'\bv(z))^2\geq \ell_z \cdot \min_{\|\bb\|_2=1}\mathbb{E}(\bb'\bv(\tilde{z}))^2, \quad \tilde{z}\sim {\rm Unif}(0,1) \\
&=\ell_z \cdot \mathbb{E}(\bb_*'\bv(\tilde{z}))^2,
\end{align*}
where $\bb_*$ is the minimizer; $\mathbb{E}(\bb_*'\bv(\tilde{z}))^2$ is a constant only depending on $k$, and it is positive because the roots of any polynomial (not identically zero) form a set of Lebesgue measure zero.

\end{proof}

\subsection{Completion of the proof}
\label{prove:both}
We are in the position to complete the proof of Theorem \ref{thm:g}. In fact, we shall prove Theorems \ref{thm:g} and \ref{thm:beta} altogether. We will choose $R^2,\lambda, \gamma, \delta_0$ such that the conditions in Lemmas \ref{lemma:1}-\ref{hp:event2} (for Theorem \ref{thm:g}) and Lemmas \ref{fast:beta}-\ref{fast:hpe} (for Theorem \ref{thm:beta}) are all satisfied. As a result, we will then conclude that the bounds on $\hat{g}(z)$ in Lemma \ref{lemma:2} and bounds on $\hat{\beta}$ in Lemma \ref{fast:beta} hold with probability at least $\mathbb{P}(\mathcal{T}_0(\delta_0, R, \tilde{R}) \cap \mathcal{T}_2(\delta_0, R) \cap \mathcal{T}_3\cap \mathcal{S}_1\cap \mathcal{S}_2\cap \mathcal{S}_3)$. Towards this end, we first list the conditions to meet\footnote{With a bit abuse of notation, the subscripts for constants $D_j$'s we use here may be different from the ones in the lemmas.}:
\begin{itemize}
\item[(i)] Conditions from Lemmas \ref{lemma:1}-\ref{lemma:2}: $\delta_0\leq \frac{1}{100}, ~\delta_0R^2\geq D_1(\frac{\log^2 n}{n^2}+\gamma+ \lambda^2 s)$.
\item[(ii)] Conditions from Lemma \ref{hg:one}: $\frac{\log p}{n}\leq 1,~\delta_0\sqrt{\frac{\log p}{n}}\frac{R^2}{\lambda^2}\leq D_2,~D_3 \delta_0^{\frac{-4k-4}{2k+3}}K_{\mathcal{F}}^2n^{\frac{-2k-2}{2k+3}} \leq R^2$, ~$D_4(1+K_{\mathcal{F}}^2)R^2\log R^{-1}\leq n\lambda^2, ~D_5(\lambda^2 nR^{\frac{-2k-1}{k+1}}K_{\mathcal{F}}^{\frac{-1}{k+1}} \wedge \lambda \sqrt{n}R^{-1}) \geq \log n$, ~$\frac{1}{2}\delta_0R^2\geq D_1 \gamma, ~R\geq \frac{\log n}{n}$. Recall that $K_{\mathcal{F}}$ is defined in \eqref{kf:def}. 
\item[(iii)] Conditions from Lemma \ref{hp:event2}: $R^2\geq D_6 \delta_0^{\frac{-4k-4}{2k+3}} n^{\frac{-2k-2}{2k+3}}(K_{\mathcal{F}}^2+K_{\mathcal{F}}^{\frac{2}{2k+3}}\sigma^{\frac{4k+4}{2k+3}})$, ~$\frac{11\log n}{\ell_z n}<1$.
\item[(iv)] Conditions from Lemma \ref{fast:beta}: $\delta_0\leq \frac{1}{100}, ~\delta_0R^2\geq D_1(\frac{\log^2 n}{n^2}+\gamma+ \lambda^2 s), ~\lambda \geq D_7(R^2\lambda^{-1}\frac{\log^2 n}{n^2}+\gamma)$.
\item[(v)] Conditions from Lemma \ref{fast:hpe}: $\frac{\log p}{n}\leq 1,~\delta_0\sqrt{\frac{\log p}{n}}\frac{R^2}{\lambda^2}\leq D_8,~D_9 \delta_0^{\frac{-4k-4}{2k+3}}K_{\mathcal{F}}^2n^{\frac{-2k-2}{2k+3}} \leq R^2$, ~$D_{10}(1+K_{\mathcal{F}}^2)R^2\log R^{-1}\leq n\lambda^2, ~D_{11}(\lambda^2 nR^{\frac{-2k-1}{k+1}}K_{\mathcal{F}}^{\frac{-1}{k+1}} \wedge \lambda \sqrt{n}R^{-1}) \geq \log n$, ~$\frac{\log n}{n}(1+K_{\mathcal{F}}+\sigma+R^2\lambda^{-1})+\sigma^{-1}R^2\leq D_{12}\lambda$.
\end{itemize} 
We choose $\delta_0=\frac{1}{100}, \lambda=D\sigma\sqrt{\frac{\log p}{n}}, \gamma=\frac{\delta_0R^2}{2D_1}$, where $D>0$ is any given constant that may only depend on $k, L_g, L_h, K_x, K_{\epsilon}, \ell_z, \Lambda_{min}, \Lambda_{max}$. It remains to specify $R^2$. We will consider $R^2\leq 1$, which together with $ \gamma=\frac{\delta_0R^2}{2D_1}$ implies 
\[
K_{\mathcal{F}} =(1\vee C_{z,k})\Big(\frac{20\delta_0 R}{\gamma}+1+\sqrt{\frac{\Lambda_{max}}{\Lambda_{min}}}\Big)R\leq (1\vee C_{z,k})\Big(40D_1+1+\sqrt{\frac{\Lambda_{max}}{\Lambda_{min}}}\Big):=D_{13}.
\]
Define
\begin{align*}
R_0^2=\frac{2D_1}{\delta_0}\Big(\frac{\log^2 n}{n^2}+\lambda^2s\Big)+\delta_0^{\frac{-4k-4}{2k+3}}n^{\frac{-2k-2}{2k+3}}\Big((D_3+D_9)D^2_{13}+D_6(D_{13}^2+D_{13}^{\frac{2}{2k+3}}\sigma^{\frac{4k+4}{2k+3}})\Big).
\end{align*}
For any $R^2\in [R_0^2, 1]$, it is straightforward to verify that Conditions from Lemmas \ref{lemma:1}-\ref{lemma:2} and all the conditions involving $n^{\frac{-2k-2}{2k+3}}$ are satisfied. To meet other conditions, we will pick $R^2=cR_0^2$ for any constant $c\geq 1$. We verify the remaining conditions by showing that under the scaling $\frac{(s+\log p)s\log p}{n}=o(1), p\rightarrow \infty$ as $n\rightarrow \infty$ \big(so that $\lambda \asymp \sqrt{\frac{\log p}{n}}, R^2\asymp \gamma \asymp \frac{s\log p}{n}+n^{\frac{-2k-2}{2k+3}}, R^2=o(\lambda)$\big):
\begin{itemize}
\item[(ii)] Conditions from Lemma \ref{hg:one}: $\frac{\log p}{n}=o(1), ~\sqrt{\frac{\log p}{n}}\frac{R^2}{\lambda^2}=O(R^2/\lambda)=o(1), ~R^2\log R^{-1}=o(1) \ll \log p \asymp n\lambda^2, ~\lambda\sqrt{n}R^{-1}\gg \sqrt{n}R \gtrsim n^{\frac{1}{4k+6}}\gg \log n, ~R\gtrsim n^{\frac{-k-1}{2k+3}}\gg\frac{\log n}{n}$. Moreover, $\lambda^2 nR^{\frac{-2k-1}{k+1}}\gtrsim n^{\frac{2k+1}{2k+3}}\cdot \log p  \gg \log n$ when $\frac{s\log p}{n}\leq n^{\frac{-2k-2}{2k+3}}$; when $\frac{s\log p}{n}> n^{\frac{-2k-2}{2k+3}}$, it still holds that 
\[
\lambda^2 nR^{\frac{-2k-1}{k+1}}\gtrsim \Big(\frac{n}{s\log p}\Big)^{\frac{2k+1}{2k+2}}\cdot \log p \gg s^{\frac{2k+1}{2k+2}} \cdot \log p \gtrsim n^{\frac{2k+1}{(2k+2)(2k+3)}}\cdot (\log p)^{\frac{1}{2k+2}} \gg \log n.
\]
\item[(iii)] Conditions from Lemma \ref{hp:event2}: $\frac{n}{\log n}=o(1)$.
\item[(iv)] Conditions from Lemma \ref{fast:beta}: $R^2\lambda^{-1}\frac{\log^2 n}{n^2}+\gamma \asymp R^2 \ll \lambda$.
\item[(v)] Conditions from Lemma \ref{fast:hpe}: Most conditions have been verified for conditions in Lemma \ref{hg:one}. It remains to see that $\frac{\log n}{n} (1+R^2\lambda^{-1})+R^2\lesssim \frac{\log n}{n}+R^2 \ll \lambda$.
\end{itemize}
Finally, given what we have obtained, it is straightforward to further evaluate the high probability (by choosing the constant $D$ in $\lambda$ and $c$ in $R^2$ large enough):
\[
\mathbb{P}(\mathcal{T}_0(\delta_0, R, \tilde{R}) \cap \mathcal{T}_2(\delta_0, R) \cap \mathcal{T}_3\cap \mathcal{S}_1\cap \mathcal{S}_2\cap \mathcal{S}_3)\geq 1-p^{-D_{14}}-n^{-D_{15}},
\]
when $n$ is large enough under the scaling $\frac{s\log^2 p}{n}=o(1), p\rightarrow \infty$.

\section{Proof of Theorem \ref{thm:beta}}\label{proof:thm2}

\subsection{Roadmap of the proof}

Recall the sets $\mathcal{B}$ in \eqref{beta:def} and $\mathcal{F}$ in \eqref{fset:def} from the proof of Lemma \ref{hg:one}. Define the following set and events:
\begin{align*}
\tilde{\mathcal{F}}&=\Big\{g: \|g\|_{\infty}\leq 22b_k (L_g+L_h)\ell_z^{-1}n^{-1}\log n\Big\}, \\
\mathcal{S}_1&=\Big\{\Big\|\frac{1}{n}(\bDelta(\bZ)+\tilde{\bX})'(g(\bZ)-\bh(\bZ)\bbeta+\beps)\Big\|_{\infty} \leq \frac{1}{8}\lambda, ~\forall \bbeta\in \mathcal{B}, g\in \mathcal{F} \cup \tilde{\mathcal{F}}, \Delta_j\in \tilde{\mathcal{F}}, 1\leq j\leq p\Big\}, \\
\mathcal{S}_2&=\Big\{\Big|\|\tilde{\bX}\bbeta\|_n^2-\|\tilde{\bx}'\bbeta\|^2\Big|\leq \lambda \|\bbeta\|_1, \forall \bbeta\in \mathcal{B}\Big\}, ~~\mathcal{S}_3=\Big\{\Big|\|\bX\bbeta\|_n^2-\|\bx'\bbeta\|^2\Big|\leq \lambda \|\bbeta\|_1, \forall \bbeta\in \mathcal{B}\Big\}.
\end{align*}
The proof of Theorem \ref{thm:g} already yields an error rate for $\hat{\bbeta}$ (see Lemmas \ref{lemma:1} and \ref{basics}), however, the rate is not optimal yet. We will base on this ``slow rate" result to perform a finer analysis of $\hat{\bbeta}$ to obtain the ``fast rate". This is again inspired by  \cite{muller2015partial}. Specifically, we prove in Lemma \ref{fast:beta} that the ``fast rate" results on $\hat{\bbeta}$ in Theorem \ref{thm:beta} hold by intersecting with another event $\mathcal{S}_1\cap \mathcal{S}_2 \cap \mathcal{S}_3$. We then show in Lemma \ref{fast:hpe} that this additional event $\mathcal{S}_1\cap \mathcal{S}_2 \cap \mathcal{S}_3$ also has large probability.

\subsection{Important lemmas}

\begin{lemma}
\label{fast:beta}
Assume Conditions \ref{con:three}-\ref{con:five} and 
\[
\delta_0\leq \frac{1}{100},~~\delta_0R^2\geq D_1(n^{-2}\log^2 n+\gamma+ \lambda^2 s), ~~\lambda \geq D_2(R^2\lambda^{-1}n^{-2}\log^2n+\gamma),
\]
where the constant $D_1$ is the same as in Lemma \ref{lemma:1}. Then on $\mathcal{T}_1(\delta_0, R) \cap \mathcal{T}_2(\delta_0, R) \cap \mathcal{T}_3 \cap \mathcal{S}_1 \cap \mathcal{S}_2$,
\begin{align*}
\|\hat{\bbeta}-\bbeta^0\|_2^2\leq 16\Lambda^{-2}_{min}s\lambda^2, \quad  \|\bx'(\hat{\bbeta}-\bbeta^0)\|^2\leq 32(\Lambda_{min}+\Lambda_{max})\Lambda^{-2}_{min}s\lambda^2.
\end{align*}
Further intersecting with the event $\mathcal{S}_3$, we have $\|\bX\hat{\bbeta}-\bX\bbeta^0\|_n^2\leq 16(3\Lambda_{min}+2\Lambda_{max})\Lambda^{-2}_{min}s\lambda^2$.
\end{lemma}

\begin{proof}
Recall that under Condition \ref{con:five}, $\max_{1\leq j\leq p}{\rm TV}(h_j^{(k)})\leq L_h$ where $\bh(z)=(h_1(z),\ldots, h_p(z))=\mathbb{E}(\bx|z)$. According to Lemma 16 in \cite{tibshirani2022divided}, there exist $\bar{h}_j\in \mathcal{H}_n$ for all $1\leq j \leq p$ such that
\begin{align*}
{\rm TV}(\bar{h}_j^{(k)})\leq a_k L_h, ~~\|\bar{h}_j-h_j\|_{\infty}\leq b_k L_h \cdot \max_{1\leq i \leq n-1} (z_{(i+1)}-z_{(i)}),
\end{align*}
where $a_k,b_k>0$ are constants only depending on $k$, and $\mathcal{H}_n$ is the span of the $k$th degree falling factorial basis functions with knots $z_{(k+1)}\leq z_{(k+2)}\ldots\leq z_{(n-1)}$. Define 
\[
\bar{\bh}(z)=(\bar{h}_1(z),\ldots, \bar{h}_p(z)), \quad \check{g}(z)=\hat{g}(z)+\bar{\bh}'(z)(\hat{\bbeta}-\bbeta^0). 
\]
Note that $\check{g}\in  \mathcal{H}_n$ since $\hat{g}, \bar{h}_j\in \mathcal{H}_n$. Therefore, we can write down a basic inequality,
\begin{align*}
\frac{1}{2}\|\by-\bX\hat{\bbeta}-\hat{g}(\bZ)\|_n^2+\lambda\|\hat{\bbeta}\|_1+\gamma {\rm TV}(\hat{g}^{(k)}) \leq \frac{1}{2}\|\by-\bX\bbeta^0-\check{g}(\bZ)\|_n^2+\lambda\|\bbeta^0\|_1+\gamma {\rm TV}(\check{g}^{(k)}),
\end{align*} 
which can be reformulated as
\begin{align*}
& \frac{1}{2}\|\tilde{\bX}(\hat{\bbeta}-\bbeta^0)\|_n^2+\lambda\|\hat{\bbeta}\|_1 \leq \frac{1}{2}\|(\bh(\bZ)-\bar{\bh}(\bZ))(\hat{\bbeta}-\bbeta^0)\|_n^2 + \gamma ({\rm TV}(\check{g}^{(k)})-{\rm TV}(\hat{g}^{(k)}))+\\
& \quad \quad \big\langle g_0(\bZ)-\hat{g}(\bZ)-\bh(\bZ)(\hat{\bbeta}-\bbeta_0)+\beps, (\bh(\bZ)-\bar{\bh}(\bZ)+\tilde{\bX})(\hat{\bbeta}-\bbeta^0) \big\rangle_n+\lambda\|\bbeta^0\|_1.
\end{align*}
Using triangle inequality for ${\rm TV}(\cdot)$ and H\"{o}lder's inequality, we can proceed to obtain
\begin{align}
&\frac{1}{2}\|\tilde{\bX}(\hat{\bbeta}-\bbeta^0)\|_n^2+\lambda\|\hat{\bbeta}\|_1 \nonumber\\ &~~~~\leq  \frac{1}{2n}\|(\bh(\bZ)-\bar{\bh}(\bZ))'(\bh(\bZ)-\bar{\bh}(\bZ))\|_{\max} \cdot \|\hat{\bbeta}-\bbeta^0\|_1^2 \nonumber\\
&~~~~~~~~+\gamma a_kL_h \|\hat{\bbeta}-\bbeta^0\|_1 \nonumber \\
 &~~~~~~~~+ \big\|\frac{1}{n}(\bh(\bZ)-\bar{\bh}(\bZ)+\tilde{\bX})'(g_0(\bZ)-\hat{g}(\bZ)-\bh(\bZ)(\hat{\bbeta}-\bbeta^0)+\beps)\big\|_{\infty}\cdot \|\hat{\bbeta}-\bbeta^0\|_1+\lambda\|\bbeta^0\|_1.\label{bq:2} 
\end{align}
To further simplify the above inequality, we observe the following:
\begin{itemize}
\item[(i)] $\frac{1}{n}\|(\bh(\bZ)-\bar{\bh}(\bZ))'(\bh(\bZ)-\bar{\bh}(\bZ))\|_{\max} \leq \max_{j} \|h_j-\bar{h}_j\|^2_{\infty} \leq b_k^2L_h^2 22^2\ell_z^{-2}n^{-2}\log^{2}n$ on $\mathcal{T}_3$
\item[(ii)] Conditions in Lemma \ref{fast:beta} imply the conditions in Lemma \ref{lemma:1} and Lemma \ref{basics}. Hence, combining Lemmas \ref{lemma:1} and \ref{basics} shows $\hat{\bbeta}-\bbeta^0 \in \mathcal{B}, \bar{g}-\hat{g}\in \mathcal{F}, g_0-\bar{g} \in \tilde{\mathcal{F}}$.
\end{itemize}
Based on these two results and assuming $4840b_k^2L_h^2\delta_0\ell_z^{-2}R^2\lambda^{-1}n^{-2}\log^2n+\gamma a_kL_h \leq \lambda/4$, we continue from \eqref{bq:2} to achieve that on $\mathcal{T}_1(\delta_0, R) \cap \mathcal{T}_2(\delta_0, R) \cap \mathcal{T}_3 \cap \mathcal{S}_1$,
\begin{align*}
 \frac{1}{2}\|\tilde{\bX}(\hat{\bbeta}-\bbeta^0)\|_n^2+\lambda\|\hat{\bbeta}\|_1 & \leq (4840b_k^2L_h^2\delta_0\ell_z^{-2}R^2\lambda^{-1}n^{-2}\log^2n+\gamma a_kL_h+\lambda/4)\cdot \|\hat{\bbeta}-\bbeta^0\|_1\\
 &~~~~+\lambda \|\bbeta^0\|_1 \\
 & \leq \frac{\lambda}{2}\|\hat{\bbeta}-\bbeta^0\|_1+\lambda \|\bbeta^0\|_1.
\end{align*}
A standard argument in the literature of LASSO (e.g., Lemma 6.3 in \cite{buhlmann2011statistics}) simplifies the above to
\begin{align*}
 &\|\tilde{\bX}(\hat{\bbeta}-\bbeta^0)\|_n^2+\lambda\|\hat{\bbeta}-\bbeta^0\|_1\leq 4\lambda \|\hat{\bbeta}_S-\bbeta^0_S\|_1 \leq 4\sqrt{s}\lambda \|\hat{\bbeta}-\bbeta^0\|_2 \\
 \leq & 4\sqrt{s}\lambda \Lambda_{min}^{-1/2} \|\tilde{\bx}'(\hat{\bbeta}-\bbeta^0)\| \leq 8s\lambda^2 \Lambda^{-1}_{min} + \frac{1}{2}\|\tilde{\bx}'(\hat{\bbeta}-\bbeta^0)\|^2 \\
  \leq & 8s\lambda^2 \Lambda^{-1}_{min} + \frac{1}{2}\|\tilde{\bX}(\hat{\bbeta}-\bbeta^0)\|_n^2+ \frac{\lambda}{2}\|\hat{\bbeta}-\bbeta^0\|_1,
\end{align*} 
where the second to last equality is due to $ab\leq a^2/2+b^2/2$, and the last inequality holds on $\mathcal{S}_2$. Rearranging the terms leads to 
\begin{align}
\label{lasso:b}
\|\tilde{\bX}(\hat{\bbeta}-\bbeta^0)\|_n^2+\lambda\|\hat{\bbeta}-\bbeta^0\|_1 \leq 16\Lambda^{-1}_{min}s\lambda^2.
\end{align}
Now \eqref{lasso:b} combined with the condition from $\mathcal{S}_2$ implies
\begin{align*}
\|\hat{\bbeta}-\bbeta^0\|^2_2 &\leq \Lambda_{min}^{-1} \cdot \|\tilde{\bx}'(\hat{\bbeta}-\bbeta^0)\|^2 \leq \Lambda_{min}^{-1} \cdot (\|\tilde{\bX}(\hat{\bbeta}-\bbeta^0)\|_n^2+\lambda\|\hat{\bbeta}-\bbeta^0\|_1) \\
&\leq 16\Lambda^{-2}_{min}s\lambda^2,  \\
\|\bx'(\hat{\bbeta}-\bbeta^0)\|^2 &= (\|\tilde{\bx}'(\hat{\bbeta}-\bbeta^0)\|+\|\bh(z)'(\hat{\bbeta}-\bbeta^0)\|)^2 \\
&\leq 2\|\tilde{\bx}'(\hat{\bbeta}-\bbeta^0)\|^2+2\|\bh(z)'(\hat{\bbeta}-\bbeta^0)\|^2 \\
&\leq 32\Lambda^{-1}_{min}s\lambda^2+2\Lambda_{max}\|\hat{\bbeta}-\bbeta^0\|^2_2 \leq  32(\Lambda_{min}+\Lambda_{max})\Lambda^{-2}_{min}s\lambda^2.
\end{align*}
Further intersecting with $\mathcal{S}_3$, we obtain
\[
\|\bX(\hat{\bbeta}-\bbeta^0)\|_n^2\leq \|\bx'(\hat{\bbeta}-\bbeta^0)\|^2+ \lambda  \|\hat{\bbeta}-\bbeta^0\|_1\leq 32(\Lambda_{min}+\Lambda_{max})\Lambda^{-2}_{min}s\lambda^2+ 16\Lambda^{-1}_{min}s\lambda^2.
\]
\end{proof}

\begin{lemma}
\label{fast:hpe}
Assume Conditions \ref{con:one}-\ref{con:five} hold.
\begin{itemize}
\item[(i)] Suppose
\begin{align}
&D_1 \delta_0^{\frac{-4k-4}{2k+3}}K_{\mathcal{F}}^2n^{\frac{-2k-2}{2k+3}} \leq R^2, \quad D_2(1+K_{\mathcal{F}}^2)R^2\log R^{-1}\leq n\lambda^2, \label{relb:1}\\
&D_3(\lambda^2 nR^{\frac{-2k-1}{k+1}}K_{\mathcal{F}}^{\frac{-1}{k+1}} \wedge \lambda \sqrt{n}R^{-1}) \geq \log n, \label{relb:2}\\
&n^{-1}\log n (1+K_{\mathcal{F}}+\sigma+R^2\lambda^{-1})+\sigma^{-1}R^2\leq D_4\lambda, \label{relb:3}
\end{align}
where $K_{\mathcal{F}}$ is defined in \eqref{kf:def} of Lemma \ref{hg:one}. Then,
\begin{align*}
\mathbb{P}(\mathcal{S}_1) \geq & 1-6pe^{-n}-2(p+p^2)e^{-D_5n\min(\sigma^{-2}\lambda^2, \sigma^{-1}\lambda)} \\
&-8pe^{-D_6(\lambda^2 nR^{\frac{-2k-1}{k+1}}K_{\mathcal{F}}^{\frac{-1}{k+1}} \wedge \lambda \sqrt{n}R^{-1})}-4pe^{-\frac{\delta_0^2R^2n}{D_7K_{\mathcal{F}}^2}}.
\end{align*}
\item[(ii)] Suppose $\frac{\log p}{n}\leq 1, \delta_0\sqrt{\frac{\log p}{n}}\frac{R^2}{\lambda^2}\leq D_8$. Then,
\begin{align*}
\mathbb{P}(\mathcal{S}_2)\geq 1-2p^{-10}, \quad \mathbb{P}(\mathcal{S}_3)\geq 1-2p^{-10}.
\end{align*}
\end{itemize}
\end{lemma}

\begin{proof}
For $\bDelta(z)=(\Delta_1(z),\ldots, \Delta_p(z))$ with $\Delta_j\in \tilde{\mathcal{F}}, 1\leq j \leq p$,
\begin{align*}
&\Big\|\frac{1}{n}\bDelta(\bZ)'(g(\bZ)-\bh(\bZ)\bbeta+\beps)\Big\|_{\infty} \\
\leq&~ \|\bDelta(\bZ)\|_{\max}\cdot \Big(\frac{1}{n}\|g(\bZ)\|_1+\frac{1}{n}\|\bh(\bZ)\bbeta\|_1+\frac{1}{n}\|\beps\|_1\Big) \\
\leq &~ 22b_k(L_g+L_h)\ell_z^{-1}n^{-1}\log n \cdot \Big(\|g\|_{\infty}+\|\bbeta\|_1\cdot \max_j \|h_j(\bZ)\|_1/n+ \|\beps\|_1/n\Big).
\end{align*}
Using Hoeffding's inequality in Theorem \ref{hoeffding:quote}, we obtain
\begin{align*}
\mathbb{P}( \|\beps\|_1/n\leq C_3\sigma K_{\varepsilon}) \geq 1- 2e^{-n}, \quad \mathbb{P}(\max_j \|h_j(\bZ)\|_1/n \leq C_4 K_x)\geq 1-2pe^{-n}. 
\end{align*}
Given that $\|x_j\|_{\psi_2}\leq K_x$, it holds that $\|h_j(z)\|_{\psi_2}\leq C_1K_x, \|\tilde{x}_j\|_{\psi_2}\leq C_2K_x$. Thus with probability at least $1-2(p+1)e^{-n}$, 
\begin{align}
\label{proof:1}
&\sup_{\bbeta\in\mathcal{B} , g\in \mathcal{F} \cup \tilde{\mathcal{F}},\Delta_j\in \tilde{\mathcal{F}}}\Big\|\frac{1}{n}\bDelta(\bZ)'(g(\bZ)-\bh(\bZ)\bbeta+\beps)\Big\|_{\infty} \nonumber \\
\leq &  D_1n^{-1}\log n \cdot \big(K_{\mathcal{F}}+ n^{-1}\log n +R^2\lambda^{-1}+\sigma \big).
\end{align}
Next we have
\begin{align*}
&\Big\|\frac{1}{n}\tilde{\bX}'(g(\bZ)-\bh(\bZ)\bbeta+\beps)\Big\|_{\infty} \leq \Big\|\frac{1}{n}\tilde{\bX}'g(\bZ)\Big\|_{\infty}+ \Big\|\frac{1}{n}\tilde{\bX}'\bh(\bZ)\bbeta \Big\|_{\infty}+ \Big\|\frac{1}{n}\tilde{\bX}'\beps\Big\|_{\infty} \\
\leq & \max_{1\leq j\leq p} \big|\frac{1}{n}\sum_{i=1}^n\tilde{x}_{ij}g(z_i)\big|+\|\beta\|_1\cdot \max_{1\leq j,k\leq p} \big|\frac{1}{n}\sum_{i=1}^n\tilde{x}_{ij}h_k(z_i)\big|+\big\|\frac{1}{n}\sum_{i=1}^n \tilde{\bx}_i\varepsilon_i\big\|_{\infty}.
\end{align*}
Note that we have used Bernstein's inequality to bound $\big\|\frac{1}{n}\sum_{i=1}^n \bx_i\varepsilon_i\big\|_{\infty}$ in the proof of Lemma \ref{hp:event2}. The same result holds here,
\[
\mathbb{P}\Big(\Big\|\frac{1}{n}\sum_{i=1}^n\tilde{\bx}_i\varepsilon_i\Big\|_{\infty} \leq \frac{\lambda}{40}\Big)\geq 1-2pe^{-C_5n\min\big(\frac{\lambda^2}{(K_xK_{\varepsilon}\sigma)^2}, \frac{\lambda}{K_xK_{\varepsilon}\sigma}\big)}.
\]
Since $\{\tilde{x}_{ij}h_k(z_i)\}_{i=1}^n$ are independent, zero-mean, sub-exponential random variables with $\|\tilde{x}_{ij}h_k(z_i)\|_{\psi_1}\leq \|\tilde{x}_{ij}\|_{\psi_2}\cdot \|h_k(z_i)\|_{\psi_2}\leq C_6K_x^2$, we use again Bernstein's inequality to obtain
\[
\mathbb{P}\Big(\max_{1\leq j,k\leq p} \big|\frac{1}{n}\sum_{i=1}^n\tilde{x}_{ij}h_k(z_i)\big| \leq \lambda/(20\sigma)\Big)\geq 1-2p^2e^{-C_7n\min\big(\frac{\lambda^2}{K^4_x\sigma^2}, \frac{\lambda}{K^2_x\sigma}\big)}.
\]
Hence with probability at least $1-2(p+p^2)e^{-C_8n\min\big(\frac{\lambda^2}{K_x^2\sigma^2(K^2_{\varepsilon}+K_x^2)}, \frac{\lambda}{K_x\sigma(K_{\varepsilon}+K_x)}\big)}$,
\begin{align}
\label{proof:2}
\sup_{\bbeta\in\mathcal{B} , g\in \mathcal{F} \cup \tilde{\mathcal{F}},\Delta_j\in \tilde{\mathcal{F}}} \Big\|\frac{1}{n}\tilde{\bX}'(g(\bZ)-\bh(\bZ)\bbeta+\beps)\Big\|_{\infty} &\leq \sup_{g\in g\in \mathcal{F} \cup \tilde{\mathcal{F}}, 1\leq j\leq p}  \big|\frac{1}{n}\sum_{i=1}^n\tilde{x}_{ij}g(z_i)\big|\nonumber \\
&~~+ \sigma^{-1}\delta_0R^2+\frac{\lambda}{40}.
\end{align}
We continue in the following way,
\begin{align}
\label{proof:3}
\sup_{g\in \mathcal{F} \cup \tilde{\mathcal{F}}, 1\leq j\leq p}  \big|\frac{1}{n}\sum_{i=1}^n\tilde{x}_{ij}g(z_i)\big| &\leq \sup_{ g\in \mathcal{F}, 1\leq j\leq p}  \big|\frac{1}{n}\sum_{i=1}^n\tilde{x}_{ij}g(z_i)\big|+ \sup_{g\in \tilde{\mathcal{F}}, 1\leq j\leq p}  \big|\frac{1}{n}\sum_{i=1}^n\tilde{x}_{ij}g(z_i)\big| \nonumber \\
&\hspace{-0.3cm} \leq \sup_{g\in \mathcal{F}, 1\leq j\leq p}  \big|\frac{1}{n}\sum_{i=1}^n\tilde{x}_{ij}g(z_i)\big|\nonumber \\
&~~+ \frac{22b_k(L_g+L_h)\log n}{\ell_z n}\cdot \max_{1\leq j\leq p} \frac{1}{n}\sum_{i=1}^n|\tilde{x}_{ij}|
\end{align}
As in bounding $\max_j \|h_j(\bZ)\|_1/n$ earlier, we bound 
\[
\mathbb{P}\Big(\max_{1\leq j\leq p} \frac{1}{n}\sum_{i=1}^n|\tilde{x}_{ij}| \leq C_9 K_x\Big)\geq 1-2pe^{-n}. 
\]
Moreover, applying the same arguments for bounding $\sup_{g\in \mathcal{F}}|\frac{1}{n}\sum_{i=1}^n(x_{ij}g(z_i)-\mathbb{E}x_jg(z))|$ in the proof of Lemma \ref{hg:one}, we have that under the conditions \eqref{relb:1}-\eqref{relb:2},
\begin{align}
\label{proof:4}
  \mathbb{P}\Big(\sup_{g\in \mathcal{F}, 1\leq j\leq p}  \big|\frac{1}{n}\sum_{i=1}^n\tilde{x}_{ij}g(z_i)\big| \geq\frac{\lambda}{120}\Big)\leq 8pe^{-D_2(\lambda^2 nR^{\frac{-2k-1}{k+1}}K_{\mathcal{F}}^{\frac{-1}{k+1}} \wedge \lambda \sqrt{n}R^{-1})}+4pe^{-\frac{\delta_0^2R^2n}{D_3K_{\mathcal{F}}^2}}.
\end{align}
Combining \eqref{relb:3}-\eqref{proof:4} completes the bound for $\mathbb{P}(\mathcal{S}_1)$.

Regarding the bound for $\mathbb{P}(\mathcal{S}_2)$ and $\mathbb{P}(\mathcal{S}_3)$, the proof follows the same arguments used for bounding term $I$ in the proof of Lemma \ref{hg:one}.
\end{proof}

\subsection{Completion of the proof}

See Section \ref{prove:both}. 

\section{Reference materials for proofs}

\begin{theorem} 
\label{hoeffding:quote}
(General Hoeffding's inequality). Let $x_1, \ldots, x_n \in \RR$ be independent, zero-mean, sub-gaussian random variables. Then for every $t \geq 0$, we have
\begin{align*}
\mathbb{P}\bigg(\Big|\sum_{i=1}^nx_i \Big|\geq t \bigg) \leq 2\exp \Big(-\frac{ct^2}{\sum_{i=1}^n\|x_i\|^2_{\psi_2}}\Big),
\end{align*}
where $c>0$ is an absolute constant, and $\|\cdot \|_{\psi_2}$ is the sub-gaussian norm defined as $\|x \|_{\psi_2}=\inf\{t>0: \mathbb{E}e^{x^2/t^2}\leq 2\}$. 

\vspace{0.1cm}

\noindent (Bernstein's inequality). Let $x_1,\ldots, x_n \in \RR $ be independent, zero-mean, sub-exponential random variables. Then for every $t\geq 0$, we have
\begin{align*}
\mathbb{P}\bigg(\Big|\sum_{i=1}^nx_i \Big|\geq t \bigg) \leq 2\exp\Bigg[-c \min \bigg(\frac{t^2}{\sum_{i=1}^n\|x_i\|^2_{\psi_1}}, \frac{t}{\max_i \|x_i\|_{\psi_1}} \bigg) \Bigg],
\end{align*}
where $c>0$ is an absolute constant, and $\|\cdot\|_{\psi_1}$ is the sub-exponential norm defined as $\|x\|_{\psi_1}=\inf\{t>0: \mathbb{E}e^{|x|/t}\leq 2\}$.
\end{theorem}

The above two results are Theorem 2.6.2 and Theorem 2.8.1, respectively in \cite{vershynin2018high}.

\begin{theorem}
\label{maximal:noteq}
Let $x_1,\ldots, x_p \in \mathbb{R}$ be sub-gaussian random variables, which are not necessarily independent. Then there exists an absolute constant $c>0$ such that for all $p>1$,
\begin{eqnarray*}
\mathbb{E}\max_{1\leq i\leq p} |x_i| \leq c \sqrt{\log p} \max_{1\leq i\leq p} \|x_i\|_{\psi_2}.
\end{eqnarray*}
\end{theorem}

The above result can be found in Lemma 2.4 of \cite{blm13}.

\begin{theorem}
\label{emp:norm}
(Uniform convergence of empirical norms). Denote $\|g\|_n^2=\frac{1}{n}\sum_{i=1}^ng(z_i)^2$ for i.i.d. samples $z_i\in \mathcal{Z}$. For a class $\mathcal{F}$ of functions on $\mathcal{Z}$, let $\mathcal{J}_0(t, \mathcal{F})=C_0 t\int_{0}^1\sup_{\{z_1,\ldots, z_n\}\subseteq \mathcal{Z}}\sqrt{\mathcal{H}(ut/2,\mathcal{F}, \|\cdot\|_n)}du, t>0$, where $C_0>0$ is some universal constant and $\mathcal{H}$ is the metric entropy. Let $R_{ \mathcal{F}}:=\sup_{g\in \mathcal{F}}\|g\|, K_{ \mathcal{F}}:=\sup_{g\in \mathcal{F}}\|g\|_{\infty}$, and $H(t)$ be the convex conjugate of $G(t):=(\mathcal{J}^{-1}_0(t,\mathcal{F}))^2$. Then, for all $R^2_{ \mathcal{F}}\geq H(4K_{ \mathcal{F}}/\sqrt{n})$ and all $t>0$,
\begin{align*}
\mathbb{P}\Bigg( \sup_{g\in \mathcal{F}}\Big|\|g\|_n^2-\|g\|^2\Big|\geq C_1\bigg(\frac{2K_{ \mathcal{F}}\mathcal{J}_0(2R_{ \mathcal{F}},\mathcal{F})+K_{ \mathcal{F}}R_{ \mathcal{F}}\sqrt{t}}{\sqrt{n}}+\frac{K^2_{ \mathcal{F}}t}{n}\bigg)\Bigg)\leq e^{-t},
\end{align*}
where $C_1>0$ is a universal constant. 
\end{theorem}
The above is Theorem 2.2 from \cite{van2014uniform}.

\begin{theorem}
\label{emp:sc}
(Symmetrization and concentration). Let $\mathbf{Y}=(y_1,\ldots, y_n)$ be i.i.d. samples in some sample space $\mathcal{Y}$ and let $\mathcal{F}$ be a class of real-valued functions on $\mathcal{Y}$. Define $Q_{\mathcal{F}}=\sup_{f\in \mathcal{F}}\|f\|, \hat{Q}_{\mathcal{F}}=\sup_{f\in \mathcal{F}}\|f\|_n$. Then,
\begin{align*}
 &\mathbb{P}\Big(\sup_{f\in \mathcal{F}}\Big|\frac{1}{n}\sum_{i=1}^n(f(y_i)-\mathbb{E}f(y_i))\Big|\geq 4Q_{\mathcal{F}}\sqrt{\frac{2t}{n}}\Big)\leq 4\mathbb{P}\Big(\sup_{f\in \mathcal{F}}\Big|\frac{1}{n}\sum_{i=1}^n\epsilon_if(y_i)\Big|\geq Q_{\mathcal{F}}\sqrt{\frac{2t}{n}}\Big), ~~\forall t\geq 4, \\
 & \mathbb{P}\Big(\sup_{f\in \mathcal{F}}\Big|\frac{1}{n}\sum_{i=1}^n\epsilon_if(y_i)\Big| \geq C_2\cdot \Big(\mathbb{E}\Big(\sup_{f\in \mathcal{F}}\Big|\frac{1}{n}\sum_{i=1}^n\epsilon_if(y_i)\Big| \Big| \mathbf{Y}\Big)+\hat{Q}_{\mathcal{F}}\sqrt{t/n}\Big) \Big| \mathbf{Y}\Big) \leq e^{-t}, ~~\forall t >0,
\end{align*}
where $\epsilon_1,\ldots, \epsilon_n$ is a Rademacher sequence independent of $\mathbf{Y}$ and $C_2>0$ is a universal constant.
\end{theorem}
The first result is Lemma 16.1 in \cite{van2016estimation} and the second result is implied by Theorem 16.4 in \cite{van2016estimation}. 

\begin{theorem}
\label{ddy:2}
(Dudley's integral tail bound). Let $(X_t)_{t\in T}$ be a separable random process on a metric space $(T,d)$ with sub-Gaussian increments: $\|X_t-X_s\|_{\psi_2}\leq Kd(t,s), \forall t, s\in T$. Then, for every $u>0$, the event
\[
\sup_{t,s\in T}|X_t-X_s|\leq CK\Big[\int_0^{{\rm diam(T)}}\sqrt{\mathcal{H}(\varepsilon, T,d)}d\varepsilon+u\cdot {\rm diam(T)}\Big]
\]
holds with probability at least $1-2e^{-u^2}$. Here, $C>0$ is a universal constant, $\mathcal{H}$ is the metric entropy, and ${\rm diam(T)}=\sup_{s,t\in T}d(s,t)$.
\end{theorem}

The result above is Theorem 8.1.6 in \cite{vershynin2018high}. 

%Sample of cross-reference to the formula (\ref{path}) in Appendix \ref{appB}.
\section{Details regarding IDATA study and the data collection}\label{app:IDATA}

\subsection{Additional rationale for our study}
Identifying biomarkers of Ultra-processed food (UPF) intake has the potential to address limitations of traditional dietary assessment methods, like food frequency questionnaires (FFQ), that are prone to measurement error \citep{subar2003using, subar2001comparative, park2018comparison}. FFQ tend to underestimate nutrient intake, and do not systematically capture data on food source or method preparation, which inform accurate NOVA classification \citep{steele2023identifying}. In addition, biomarkers of UPF may provide novel insight into biological mechanisms underlying potential UPF-disease associations. Domiciled feeding trials offer a reliable way to control UPF consumption, study its metabolic effects, and identify candidate biomarkers of UPF intake \citep{o2023metabolomic}. A recent metabolomics investigation within a randomized, crossover, controlled-feeding trial found that consuming a diet with 80\% energy from UPF for two weeks, compared to one with 0\% energy from UPF, resulted in measurable changes to more than 250 circulating metabolites and more than 600 urine metabolites, within individuals. Building on these findings, we apply our method to investigate associations between UPF intake and a range of metabolites and to uncover nuanced connections that may improve understanding of UPF-chronic disease associations and ultimately inform public health guidance on UPF.

\subsection{Interactive Diet and Activity Tracking in AARP (IDATA) Study}
The Interactive Diet and Activity Tracking in AARP (IDATA) Study was specifically designed to evaluate the efficacy of web-based dietary assessment tools, including the Automated Self-Administered 24-hour Dietary Assessment Tool (ASA-24), four-day food records (4DFRs), and the Dietary History Questionnaire (DHQ) II, in comparison to reference biomarkers \citep{subar2020performance}. Briefly, participants were recruited from a cohort of AARP members aged 50-74 years, residing in or near Pittsburgh, Pennsylvania, who met the following inclusion criteria: English-speaking, internet access, not currently engaged in a weight-loss diet, a body mass index (BMI) below 40 kg/m$^2$, and no significant medical conditions or mobility limitations. Between 2012 and 2013, 1,082 participants were enrolled in the study, all of whom provided consent for biospecimen collection \citep{park2018comparison}. The study protocol was approved by the NCI Special Studies Institutional Review Board and registered on ClinicalTrials.gov (Identifier: NCT03268577). All participants provided written informed consent \citep{subar2020performance}.
Additional details on the design and methodology of the IDATA study have been thoroughly outlined on \href{https://cdas.cancer.gov/idata/}{https://cdas.cancer.gov/idata/}. 

\subsection{Dietary data and biological sample collections}

%\subsubsection{}

%Participants were divided into four groups to reduce the influence of seasonal variation in diet as well as for practical study-center-related reasons. Over 12 months, participants completed up to six web-based ASA-24s on a randomly assigned day approximately every other month \citep{park2018comparison}. Each food and beverage item reported was assigned to a unique 8-digit food code based on foods and beverages reported in What We Eat in America (WWEIA)\citep{martin2014usda}, NHANES \citep{steele2023identifying}. This code was linked to the Food and Nutrient Database for Dietary Studies (FNDDS)\citep{martin2014usda}, an application to convert food and beverage portions into gram amounts and to estimate nutrient values, including energy, using standard reference codes (SR codes) from the USDA National Nutrient Database for Standard Reference. 
Participants were divided into four groups to minimize the impact of seasonal variations in diet and for practical reasons related to study centers. Over a 12-month period, participants completed up to six web-based ASA-24 assessments, administered on randomly assigned days approximately every two months \citep{park2018comparison}. Each reported food and beverage item was assigned a unique 8-digit food code based on the "What We Eat in America" (WWEIA) database \citep{martin2014usda}, as part of NHANES \citep{steele2023identifying}. These codes were linked to the Food and Nutrient Database for Dietary Studies (FNDDS) \citep{martin2014usda}, which is used to convert food and beverage portions into gram amounts and estimate nutrient values, including energy, by employing standard reference codes (SR codes) from the USDA National Nutrient Database for Standard Reference.

The intake of ultra-processed foods (UPF) was estimated using the NOVA system, which categorizes foods and beverages into four groups according to the degree and purpose of industrial processing. Group 1 includes unprocessed or minimally processed foods, such as fresh, dried, or frozen fruits and vegetables, grains, legumes, meat, fish, and milk. In contrast, Group 4 comprises UPFs, such as ready-to-eat products like commercially prepared breads and baked goods, which contain ingredients not commonly used in traditional culinary practices \citep{monteiro2018decade}. We disaggregated each FNDDS food code into its corresponding SR codes through a series of merges. Each IDATA food item (identified by food code and SR code) was then assigned to one of the four Nova groups and one of 37 mutually exclusive food subgroups, as specified by the WWEIA, NHANES database developed by \citet{steele2019dietary, steele2023identifying}. For FNDDS food codes that had not yet been classified into a Nova group or subgroup, a manual review was conducted, and the classification was determined using the "reference approach" outlined by \citep{steele2023identifying}.

For each ASA-24, food and beverage intake, in gram weight, was estimated and classified according to Nova. For each participant, we averaged the absolute gram weight derived from UPF across multiple recall days. Further details regarding the calculation of UPF intake using the Nova system can be found in \citet{steele2019dietary, steele2023identifying}.

Serum metabolomics data were generated using blood samples collected during two study center visits—either at months 1 and 6 or at months 6 and 12. Urine metabolomics data were generated using home collections, which occurred approximately 7 to 10 days after a study center visit. Participants were asked to collect 100 mL of their first-morning void (FMV) prior to beginning a 24-hour urine collection.

Serum and urine metabolomic analyses were conducted by Metabolon Inc. using ultra-high-performance liquid chromatography (UPLC) coupled with tandem mass spectrometry (MS/MS) to measure a wide range of metabolites. These included endogenously derived amino acids, carbohydrates, lipids, cofactors and vitamins, energy metabolism intermediates, as well as xenobiotics originating from external sources like food or drugs. Serum and urine were analyzed separately, and all samples were prepared using the automated MicroLab STAR system (Hamilton Company). Metabolite values below the detection limit were assigned the minimum observed value for each metabolite and were standardized. For analysis, metabolite levels were log-transformed.

\section{Supplementary tables}\label{supp:tables}
\setcounter{table}{0}
\renewcommand{\thetable}{S\arabic{table}}

We present the full list of the top 10 selected metabolites for both serum and urine datasets, stratified by gender, in Tables \ref{supp:tab:serum_selected:f}-\ref{supp:tab:urine_selected:m}.

\begin{table}[!htbp] 
\caption{Top 10 selected variables results for serum dataset, stratified by female and ordered based on the magnitude of coefficients.} 
\small
\centering 
\begin{tabular}{@{\extracolsep{0pt}} cccc} 
\\[-1.8ex]\hline 
\hline 
 & & & Selected Metabolites \\ 
\hline 
 & \multirow{20}{*}{Female} & \multirow{5}{*}{Age} & Diglycerol, X-21807, \\ 
 &&& PC(18:0/22:4(7Z,10Z,13Z,16Z)), 2-Aminobenzoic acid, \\ 
 &&& 4-allylphenol sulfate, X-25523, \\ 
 &&& N-Acetylglucosamine/N-Acetylgalactosamine, Eicosapentaenoylcholine, \\ 
 &&& FIBRINOPEPTIDEA(3-15), Branched chain 14:0 dicarboxylic acid \\ \cline{3-4}
&&\multirow{5}{*}{BMI} & X-21807, Quinic acid, \\ 
 &&& Eicosapentaenoylcholine, Diglycerol, \\ 
 &&& X-25523, 4-Deoxythreonic acid, \\ 
 &&& (R)-2,3-Dihydroxy-isovalerate, FIBRINOPEPTIDEA(3-15), \\ 
 &&& 4-allylphenol sulfate, Deoxyuridine \\  \cline{3-4} 
&&\multirow{5}{*}{Hip} & Diglycerol, X-21807, \\ 
 &&& Quinic acid, Eicosapentaenoylcholine, \\ 
 &&& X-25523, PC(18:0/22:4(7Z,10Z,13Z,16Z)), \\ 
 &&& 4-allylphenol sulfate, Branched chain 14:0 dicarboxylic acid, \\ 
 &&& 2-Aminobenzoic acid, 1-Methyluric acid \\  \cline{3-4} 
&&\multirow{5}{*}{Waist} & Diglycerol, PC(18:0/22:4(7Z,10Z,13Z,16Z)), \\ 
 &&& X-21807, X-25523, \\ 
 &&& 2-Aminobenzoic acid, 4-allylphenol sulfate, \\ 
 &&& Eicosapentaenoylcholine, N2-gamma-Glutamylglutamine, \\ 
 &&& Branched chain 14:0 dicarboxylic acid, Pentoic acid \\
\hline \hline \\ \\[-1.8ex] 
\end{tabular} 
\label{supp:tab:serum_selected:f} 
\end{table}

\begin{table}[!htbp] 
\caption{Top 10 selected variables results for serum dataset, stratified by male and ordered based on the magnitude of coefficients.} 
\label{supp:tab:serum_selected:m} 
\small
\centering 
\begin{tabular}{@{\extracolsep{0pt}} cccc} 
\\[-1.8ex]\hline 
\hline 
 & & & Selected Metabolites \\ 
\hline 
& \multirow{20}{*}{Male} & \multirow{5}{*}{Age} & X-21442, X-23655, \\ 
 &&& 1-Methylhistidine, Quinic acid, \\ 
 &&& LysoPC(24:0), X-19183, \\ 
 &&& PE(P-16:0/18:1(9Z)), Saccharin, \\ 
 &&& N,N-DIMETHYLALANINE, Orotidine \\  \cline{3-4} 
&&\multirow{5}{*}{BMI} & X-21442, X-23655, \\ 
 &&& 1-Methylhistidine, Quinic acid, \\ 
 &&& decadienedioic acid (C10:2-DC), X-19183, \\ 
 &&& LysoPC(24:0), 23585, \\ 
 &&& Saccharin, PE(P-16:0/18:1(9Z)) \\  \cline{3-4} 
&&\multirow{5}{*}{Hip} & X-21442, X-23655, \\ 
 &&& Quinic acid, 1-Methylhistidine, \\ 
 &&& N-delta-acetylornithine, Branched chain 14:0 dicarboxylic acid, \\ 
 &&& LysoPC(24:0), Xanthurenic acid, \\ 
 &&& X-19183, PE(P-16:0/18:1(9Z)) \\  \cline{3-4} 
&&\multirow{5}{*}{Waist} & X-21442, X-23655, \\ 
 &&& Quinic acid, 1-Methylhistidine, \\ 
 &&& LysoPC(24:0), X-19183, \\ 
 &&& N-delta-acetylornithine, Saccharin, \\ 
 &&& Branched chain 14:0 dicarboxylic acid, Xanthurenic acid \\ 
\hline \hline \\ \\[-1.8ex] 
\end{tabular} 
\end{table}

\begin{table}[!htbp] 
\caption{Top 10 selected variables results for urine dataset, stratified by female and ordered based on the magnitude of coefficients.} 
\label{supp:tab:urine_selected:f} 
\small
\centering 
\begin{tabular}{@{\extracolsep{5pt}} cccc} 
\\[-1.8ex]\hline 
\hline 
 & & & Selected Metabolites \\ 
\hline 
 & \multirow{20}{*}{Female} & \multirow{5}{*}{Age} & X-25952, X-12096, \\ 
 &&& Saccharin, 2,3-dihydroxypyridine, \\ 
 &&& X-23680, Galactonic acid, \\ 
 &&& Glutamine conjugate of C8H12O2 (3), X-12753, \\ 
 &&& Sucrose, X-25936 \\  \cline{3-4} 
&& \multirow{5}{*}{BMI} & X-25952, X-12096, \\ 
 &&& Saccharin, 2,3-dihydroxypyridine, \\ 
 &&& Galactonic acid, Glutamine conjugate of C8H12O2 (3), \\ 
 &&& X-23680, Riboflavin, \\ 
 &&& X-12753, X-24352 \\  \cline{3-4} 
&& \multirow{5}{*}{Hip} & X-25952, X-12096, \\ 
 &&& 2,3-dihydroxypyridine, Saccharin, \\ 
 &&& Galactonic acid, Glutamine conjugate of C8H12O2 (3), \\ 
 &&& X-23680, X-24352, \\ 
 &&& Riboflavin, X-25936 \\  \cline{3-4} 
&& \multirow{5}{*}{Waist} & X-25952, X-12096, \\ 
 &&& Saccharin, 2,3-dihydroxypyridine, \\ 
 &&& Galactonic acid, Glutamine conjugate of C8H12O2 (3), \\ 
 &&& X-23680, Riboflavin, \\ 
 &&& X-25936, X-12753 \\ 
\hline \hline \\ \\[-1.8ex] 
\end{tabular} 
\end{table} 

\begin{table}[!htbp] 
\caption{Top 10 selected variables results for urine dataset, stratified by male and ordered based on the magnitude of coefficients.} 
\label{supp:tab:urine_selected:m} 
\small
\centering 
\begin{tabular}{@{\extracolsep{5pt}} cccc} 
\\[-1.8ex]\hline 
\hline 
 & & & Selected Metabolites \\ 
\hline 
 & \multirow{20}{*}{Male} & \multirow{5}{*}{Age} & X-25952, X-23161, \\ 
 &&& Quinic acid, Levoglucosan, \\ 
 &&& 3-Methoxytyramine, Cortisone, \\ 
 &&& 5-Aminolevulinic acid, X-13847, \\ 
 &&& N-a-Acetylcitrulline, X-21807 \\  \cline{3-4} 
&& \multirow{5}{*}{BMI} & X-25952, X-23161, \\ 
 &&& Quinic acid, X-13847, \\ 
 &&& 3-Methoxytyramine, Levoglucosan, \\ 
 &&& Cortisone, N-a-Acetylcitrulline, \\ 
 &&& Succinyladenosine, Ursocholic acid \\  \cline{3-4} 
&& \multirow{5}{*}{Hip} & X-23161, X-25952, \\ 
 &&& Quinic acid, X-13847, \\ 
 &&& 3-Methoxytyramine, N-a-Acetylcitrulline, \\ 
 &&& Levoglucosan, Cortisone, \\ 
 &&& Ursocholic acid, Thioproline \\  \cline{3-4} 
&& \multirow{5}{*}{Waist} & X-23161, X-25952, \\ 
 &&& Quinic acid, X-13847, \\ 
 &&& 3-Methoxytyramine, Levoglucosan, \\ 
 &&& Thioproline, Ursocholic acid, \\ 
 &&& N-a-Acetylcitrulline, Cortisone \\ 
\hline \hline \\ \\[-1.8ex] 
\end{tabular} 
\end{table}

\section{Additional simulation results for $n=100$ cases}\label{supp:n_100}
\setcounter{figure}{0}
\renewcommand{\thefigure}{S\arabic{figure}}
{\color{black} We further consider simulations for $n=100$, to evaluate the performance of PLTF and PLSS under smaller sample size. Since the number of observations is significantly reduced, we slightly modify the simulation setup in Section \ref{sec:simul_res} to obtain more stable results: (1) The values for $\beta_j, j=1,2,3,4$ are $(1.5,3,3,4.5)$; (2) The nonlinear function in Model 3 is ${\rm sin}(1/z_i)$; (3) The tSNR varies from 14 to 40. The results are presented in Figures \ref{supp:fig:model1_supp}-\ref{supp:fig:model3_supp}. We observe similar patterns when comparing PLTF and PLSS as those shown in Figures \ref{fig:model1}-\ref{fig:model3}.}

\begin{figure}[!t]
\centering
%\textbf{\small Conditional longitudinal disparity decomposition with modifier (cmLDD)}
\includegraphics[width=0.78\textwidth, height=0.7\textheight]{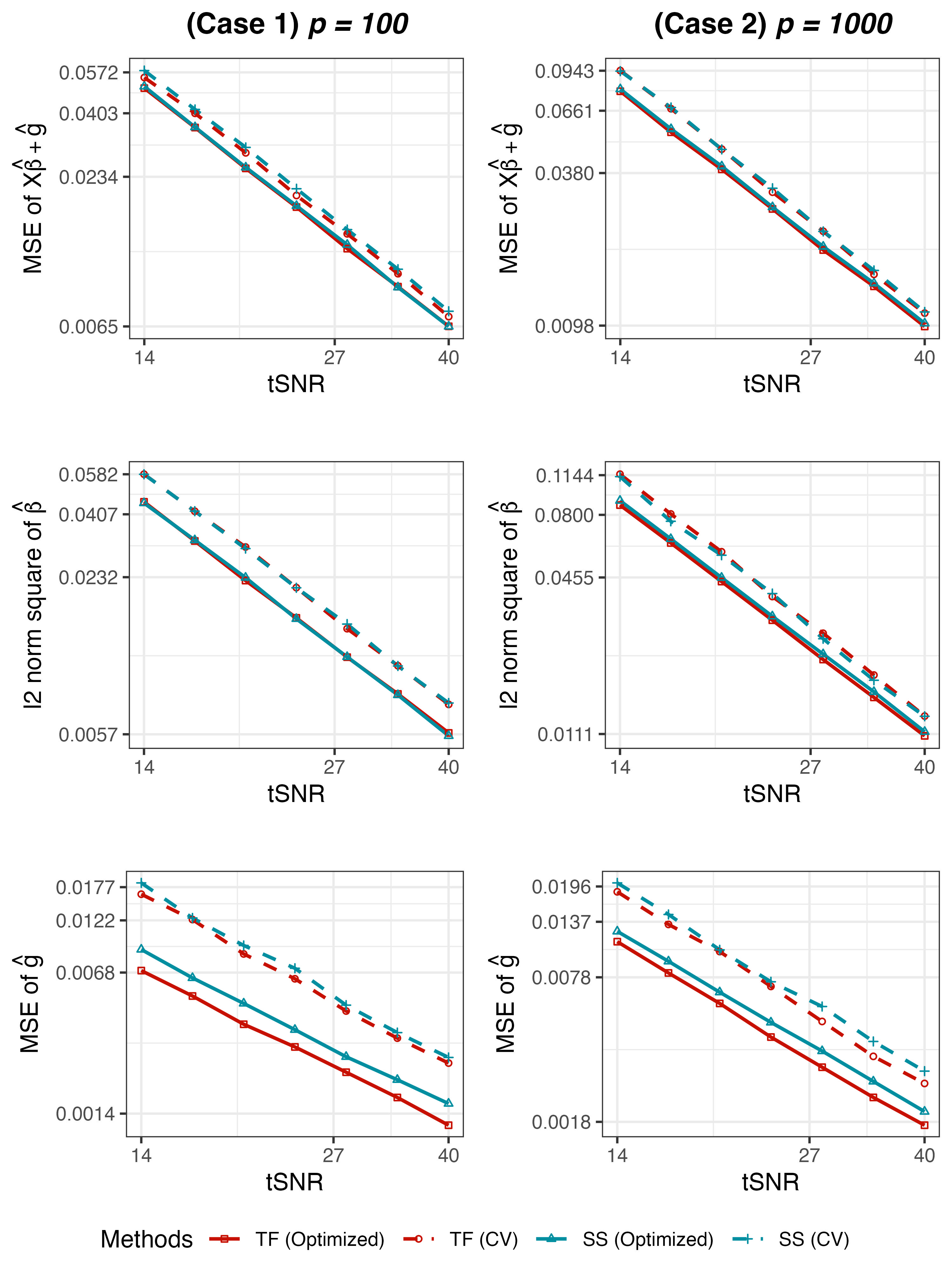}
\caption{PLTF v.s. PLSS under Model 1, with tSNR ranging from 14 to 40, and $n$ is set to 100. TF (Optimized or CV) denotes PLTF with (optimally or CV) tuned parameters. SS (Optimized or CV) denotes PLSS with (optimally or CV) tuned parameters.
}
\label{supp:fig:model1_supp}
\end{figure}

\begin{figure}[!t]
\centering
%\textbf{\small Conditional longitudinal disparity decomposition with modifier (cmLDD)}
\includegraphics[width=0.78\textwidth, height=0.7\textheight]{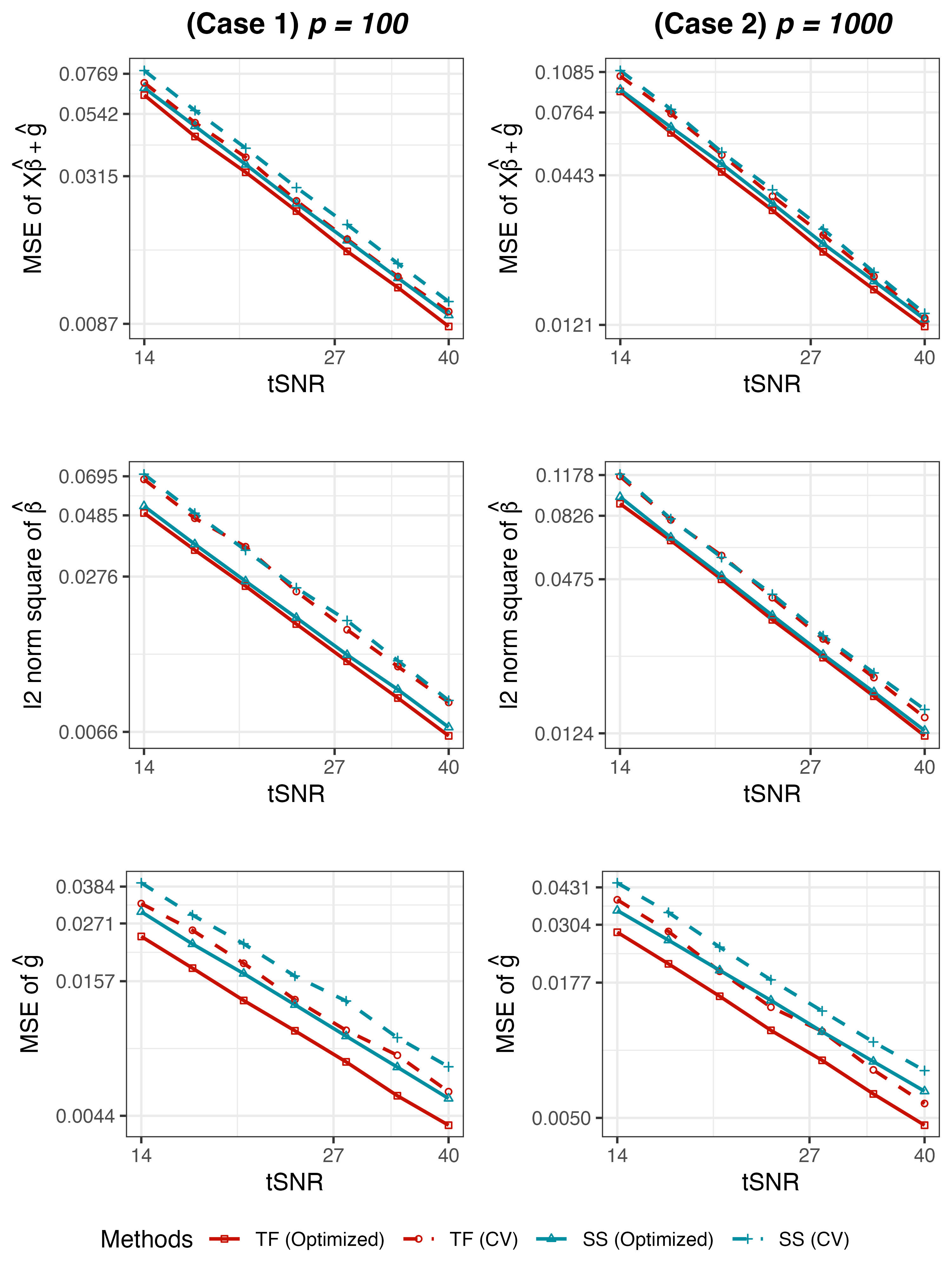}
\caption{PLTF v.s. PLSS under Model 2, with tSNR ranging from 14 to 40, and $n$ is set to 100. TF (Optimized or CV) denotes PLTF with (optimally or CV) tuned parameters. SS (Optimized or CV) denotes PLSS with (optimally or CV) tuned parameters.
}
\label{supp:fig:model2_supp}
\end{figure}

\begin{figure}[!t]
\centering
%\textbf{\small Conditional longitudinal disparity decomposition with modifier (cmLDD)}
\includegraphics[width=0.78\textwidth, height=0.7\textheight]{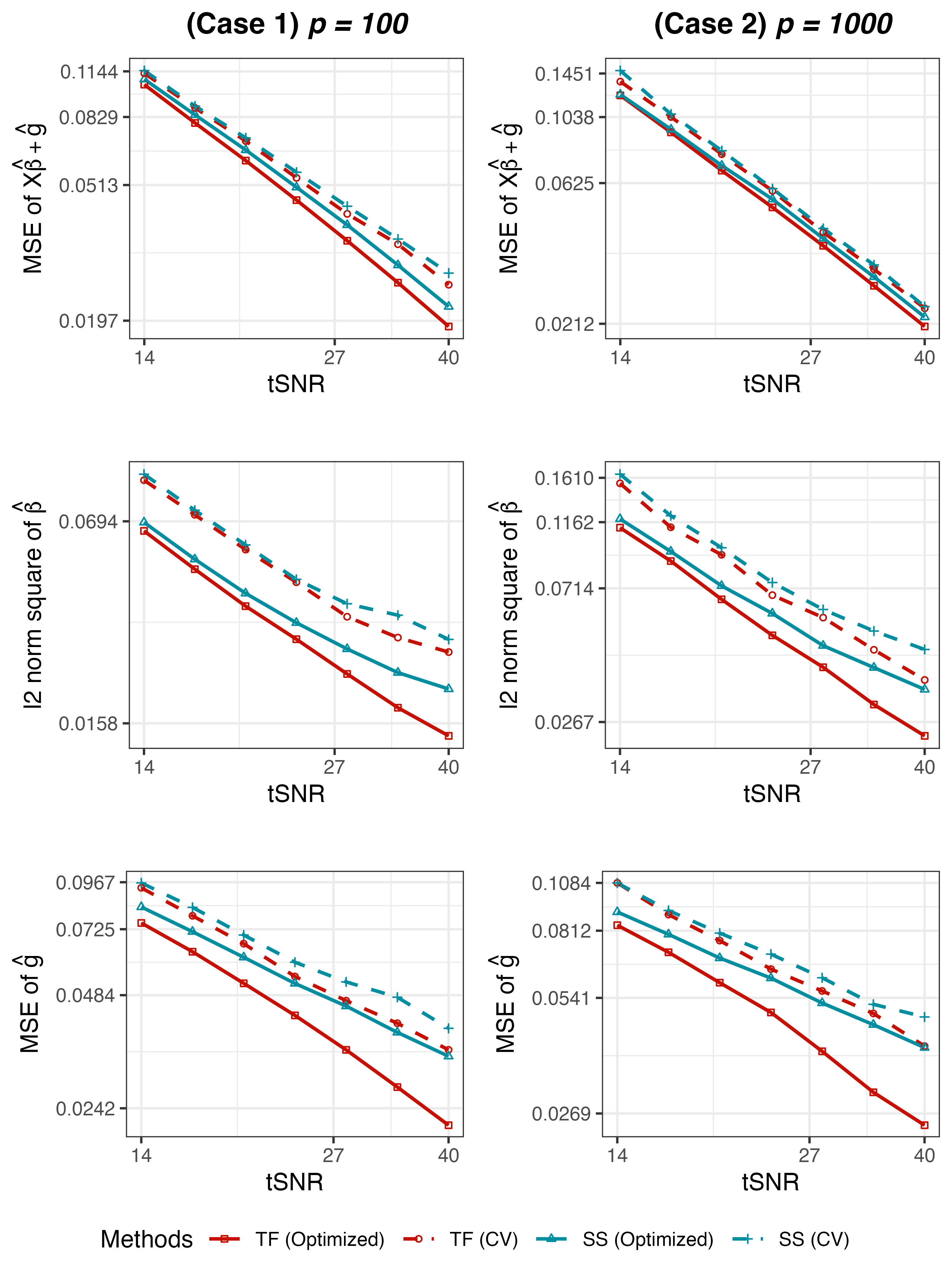}
\caption{PLTF v.s. PLSS under Model 3, with tSNR ranging from 14 to 40, and $n$ is set to 100. TF (Optimized or CV) denotes PLTF with (optimally or CV) tuned parameters. SS (Optimized or CV) denotes PLSS with (optimally or CV) tuned parameters.
}
\label{supp:fig:model3_supp}
\end{figure}

\end{appendix}

%%%%%%%%%%%%%%%%%%%%%%%%%%%%%%%%%%%%%%%%%%%%%%
%% Acknowledgements                         %%
%% should be provided in the                %%
%% Acknowledgements section.                %%
%%%%%%%%%%%%%%%%%%%%%%%%%%%%%%%%%%%%%%%%%%%%%%
\begin{acks}[Acknowledgments]
The authors would like to thank the anonymous referee, an Associate Editor and the Editor for their constructive comments that improved the quality of this paper.
\end{acks}

%%%%%%%%%%%%%%%%%%%%%%%%%%%%%%%%%%%%%%%%%%%%%%
%% Funding information, if any,             %%
%% should be provided in the                %%
%% funding section.                         %%
%%%%%%%%%%%%%%%%%%%%%%%%%%%%%%%%%%%%%%%%%%%%%%
\begin{funding}
Sang Kyu Lee was supported by the National Research Foundation of Korea (NRF) grant funded by the Korea government (RS-2026-25494847). Haolei Weng was supported by NSF Grant DMS-1915099.
\end{funding}

\bibliographystyle{imsart-nameyear} % Style BST file (imsart-number.bst or imsart-nameyear.bst)
\bibliography{mc_pltf.bib}       % Bibliography file (usually '*.bib')

@article{drikvandi2025high,
  title={High dimensional regression with many nuisance parameters: Both cases of specified and unspecified parameters of interest},
  author={Drikvandi, Reza},
  journal={Electronic Journal of Statistics},
  volume={19},
  number={1},
  pages={2923--2957},
  year={2025},
  publisher={The Institute of Mathematical Statistics and the Bernoulli Society}
}

@article{abar2025identification,
  title={Identification of Poly-Metabolite Scores for Diets High in Ultra-Processed Food in an Observational Study with Validation in a Randomized Controlled Crossover-Feeding Trial},
  author={Abar, Leila and Steele, Eur{\'\i}dice Mart{\'\i}nez and Lee, Sang Kyu and Kahle, Lisa and Moore, Steven C and Watts, Eleanor and Matthews, Charles E and Herrick, Kirsten A and Hall, Kevin D and O’Connor, Lauren E and others},
  journal={medRxiv},
  year={2025}
}

@article{monteiro2018decade,
  title={The UN Decade of Nutrition, the NOVA food classification and the trouble with ultra-processing},
  author={Monteiro, Carlos Augusto and Cannon, Geoffrey and Moubarac, Jean-Claude and Levy, Renata Bertazzi and Louzada, Maria Laura C and Jaime, Patr{\'\i}cia Constante},
  journal={Public health nutrition},
  volume={21},
  number={1},
  pages={5--17},
  year={2018},
  publisher={Cambridge University Press}
}

@article{park2018comparison,
  title={Comparison of self-reported dietary intakes from the Automated Self-Administered 24-h recall, 4-d food records, and food-frequency questionnaires against recovery biomarkers},
  author={Park, Yikyung and Dodd, Kevin W and Kipnis, Victor and Thompson, Frances E and Potischman, Nancy and Schoeller, Dale A and Baer, David J and Midthune, Douglas and Troiano, Richard P and Bowles, Heather and others},
  journal={The American Journal of Clinical Nutrition},
  volume={107},
  number={1},
  pages={80--93},
  year={2018},
  publisher={Elsevier}
}

@article{subar2003using,
  title={Using intake biomarkers to evaluate the extent of dietary misreporting in a large sample of adults: the OPEN study},
  author={Subar, Amy F and Kipnis, Victor and Troiano, Richard P and Midthune, Douglas and Schoeller, Dale A and Bingham, Sheila and Sharbaugh, Carolyn O and Trabulsi, Jillian and Runswick, Shirley and Ballard-Barbash, Rachel and others},
  journal={American Journal of Epidemiology},
  volume={158},
  number={1},
  pages={1--13},
  year={2003},
  publisher={Oxford University Press}
}

@article{subar2001comparative,
  title={Comparative validation of the Block, Willett, and National Cancer Institute food frequency questionnaires: the Eating at America's Table Study},
  author={Subar, Amy F and Thompson, Frances E and Kipnis, Victor and Midthune, Douglas and Hurwitz, Paul and McNutt, Suzanne and McIntosh, Anna and Rosenfeld, Simon},
  journal={American Journal of Epidemiology},
  volume={154},
  number={12},
  pages={1089--1099},
  year={2001},
  publisher={Oxford University Press}
}

@article{zhang2024apobec,
  title={APOBEC shapes tumor evolution and age at onset of lung cancer in smokers},
  author={Zhang, Tongwu and Sang, Jian and Hoang, Phuc H and Zhao, Wei and Rosenbaum, Jennifer and Johnson, Kofi Ennu and Klimczak, Leszek J and McElderry, John and Klein, Alyssa and Wirth, Christopher and others},
  journal={bioRxiv},
  year={2024},
  publisher={Cold Spring Harbor Laboratory Preprints}
}

@article{steele2023identifying,
  title={Identifying and estimating ultraprocessed food intake in the US NHANES according to the Nova classification system of food processing},
  author={Steele, Eur{\'\i}dice Mart{\'\i}nez and O’Connor, Lauren E and Juul, Filippa and Khandpur, Neha and Baraldi, Larissa Galastri and Monteiro, Carlos A and Parekh, Niyati and Herrick, Kirsten A},
  journal={The Journal of Nutrition},
  volume={153},
  number={1},
  pages={225--241},
  year={2023},
  publisher={Elsevier}
}

@article{playdon2024measuring,
  title={Measuring diet by metabolomics: a 14-d controlled feeding study of weighed food intake},
  author={Playdon, Mary C and Tinker, Lesley F and Prentice, Ross L and Loftfield, Erikka and Hayden, Kathleen M and Van Horn, Linda and Sampson, Joshua N and Stolzenberg-Solomon, Rachael and Lampe, Johanna W and Neuhouser, Marian L and others},
  journal={The American Journal of Clinical Nutrition},
  volume={119},
  number={2},
  pages={511--526},
  year={2024},
  publisher={Elsevier}
}

@article{muli2024association,
  title={Association of ultra-processed foods intake with untargeted metabolomics profiles in adolescents and young adults in the DONALD cohort study},
  author={Muli, Samuel and Blumenthal, Annika and Conzen, Christina-Alexandra and Benz, Maike Elena and Alexy, Ute and Schmid, Matthias and Keski-Rahkonen, Pekka and Floegel, Anna and N{\"o}thlings, Ute},
  journal={The Journal of Nutrition, In press},
  year={2024},
  publisher={Elsevier}
}

@article{schenkelaars2024intake,
  title={The intake of ultra-processed foods and homocysteine levels in women with (out) overweight and obesity: The Rotterdam Periconceptional Cohort},
  author={Schenkelaars, Nicole and van Rossem, Lenie and Willemsen, Sten P and Faas, Marijke M and Schoenmakers, Sam and Steegers-Theunissen, R{\'e}gine PM},
  journal={European Journal of Nutrition},
volume={63},
  pages={1257-1269},
  year={2024},
  publisher={Springer}
}

@article{falcao2019processed,
  title={Processed and ultra-processed foods are associated with high prevalence of inadequate selenium intake and low prevalence of vitamin B1 and zinc inadequacy in adolescents from public schools in an urban area of northeastern Brazil},
  author={Falc{\~a}o, Raphaela Cec{\'\i}lia Th{\'e} Maia de Arruda and Lyra, Cl{\'e}lia de Oliveira and Morais, C{\'e}lia M{\'a}rcia Medeiros de and Pinheiro, Liana Galv{\~a}o Bacurau and Pedrosa, Lucia F{\'a}tima Campos and Lima, Severina Carla Vieira Cunha and Sena-Evangelista, Karine Cavalcanti Maur{\'\i}cio},
  journal={PLoS One},
  volume={14},
  number={12},
  pages={e0224984},
  year={2019},
  publisher={Public Library of Science San Francisco, CA USA}
}

@article{leung2020serum,
  title={Serum metabolomic profiling and its association with 25-hydroxyvitamin D},
  author={Leung, Raymond YH and Li, Gloria HY and Cheung, Bernard MY and Tan, Kathryn CB and Kung, Annie WC and Cheung, Ching-Lung},
  journal={Clinical Nutrition},
  volume={39},
  number={4},
  pages={1179--1187},
  year={2020},
  publisher={Elsevier}
}

@article{sellem2024food,
  title={Food additive emulsifiers and cancer risk: Results from the French prospective NutriNet-Sant{\'e} cohort},
  author={Sellem, Laury and Srour, Bernard and Javaux, Guillaume and Chazelas, Eloi and Chassaing, Benoit and Viennois, Emilie and Debras, Charlotte and Druesne-Pecollo, Nathalie and Esseddik, Younes and de Edelenyi, Fabien Szabo and others},
  journal={Plos Medicine},
  volume={21},
  number={2},
  pages={e1004338},
  year={2024},
  publisher={Public Library of Science}
}

@article{o2023metabolomic,
  title={Metabolomic profiling of an ultraprocessed dietary pattern in a domiciled randomized controlled crossover feeding trial},
  author={O’Connor, Lauren E and Hall, Kevin D and Herrick, Kirsten A and Reedy, Jill and Chung, Stephanie T and Stagliano, Michael and Courville, Amber B and Sinha, Rashmi and Freedman, Neal D and Hong, Hyokyoung G and others},
  journal={The Journal of Nutrition},
  volume={153},
  number={8},
  pages={2181--2192},
  year={2023},
  publisher={Elsevier}
}

@article{fu2024association,
  title={The association between ultra-processed food intake and age-related hearing loss: a cross-sectional study},
  author={Fu, Yanpeng and Chen, Wenyu and Liu, Yuehui},
  journal={BMC Geriatrics},
  volume={24},
  number={1},
  pages={450},
  year={2024},
  publisher={Springer}
}

@article{canhada2020ultra,
  title={Ultra-processed foods, incident overweight and obesity, and longitudinal changes in weight and waist circumference: the Brazilian Longitudinal Study of Adult Health (ELSA-Brasil)},
  author={Canhada, Scheine Leite and Luft, Vivian Cristine and Giatti, Luana and Duncan, Bruce Bartholow and Chor, Dora and Maria de Jesus, M and Matos, Sheila Maria Alvim and Molina, Maria del Carmen Bisi and Barreto, Sandhi Maria and Levy, Renata Bertazzi and others},
  journal={Public Health Nutrition},
  volume={23},
  number={6},
  pages={1076--1086},
  year={2020},
  publisher={Cambridge University Press}
}

@article{martin2014usda,
  title={USDA Food and Nutrient Database for Dietary Studies 2011-2012},
  author={Martin, CL and Montville, JB and Steinfeldt, LC and Omolewa-Tomobi, Grace and Heendeniya, KY and Adler, ME and Moshfegh, AJ},
  journal={US Department of Agriculture, Agricultural Research Service, Food Surveys Research Group},
  year={2014}
}

@article{steele2019dietary,
  title={Dietary share of ultra-processed foods and metabolic syndrome in the US adult population},
  author={Steele, Euridice Martinez and Juul, Filippa and Neri, Daniela and Rauber, Fernanda and Monteiro, Carlos A},
  journal={Preventive medicine},
  volume={125},
  pages={40--48},
  year={2019},
  publisher={Elsevier}
}

@article{sung2021consumption,
  title={Consumption of ultra-processed foods increases the likelihood of having obesity in Korean women},
  author={Sung, Hyuni and Park, Ji Min and Oh, Se Uk and Ha, Kyungho and Joung, Hyojee},
  journal={Nutrients},
  volume={13},
  number={2},
  pages={698},
  year={2021},
  publisher={MDPI}
}

@article{juul2018ultra,
  title={Ultra-processed food consumption and excess weight among US adults},
  author={Juul, Filippa and Martinez-Steele, Euridice and Parekh, Niyati and Monteiro, Carlos A and Chang, Virginia W},
  journal={British Journal of Nutrition},
  volume={120},
  number={1},
  pages={90--100},
  year={2018},
  publisher={Cambridge University Press}
}

@article{subar2020performance,
  title={Performance and feasibility of recalls completed using the automated self-administered 24-hour dietary assessment tool in relation to other self-report tools and biomarkers in the interactive diet and activity tracking in AARP (IDATA) study},
  author={Subar, Amy F and Potischman, Nancy and Dodd, Kevin W and Thompson, Frances E and Baer, David J and Schoeller, Dale A and Midthune, Douglas and Kipnis, Victor and Kirkpatrick, Sharon I and Mittl, Beth and others},
  journal={Journal of the Academy of Nutrition and Dietetics},
  volume={120},
  number={11},
  pages={1805--1820},
  year={2020},
  publisher={Elsevier}
}

@article{tibshirani2012degrees,
  title={Degrees of freedom in lasso problems},
  author={Tibshirani, Ryan J and Taylor, Jonathan},
  journal={The Annals of Statistics},
  volume={40},
  number={2},
  pages={1198--1232},
  year={2012},
  publisher={Institute of Mathematical Statistics}
}

@book{hastie1990generalized,
  title={Generalized Additive Models},
  author={Hastie, TJ and Tibshirani, RJ},
  volume={43},
  year={1990},
  publisher={CRC Press}
}

@article{padilla2023temporal,
  title={Temporal-spatial model via Trend Filtering},
  author={Madrid Padilla, Carlos Misael and Madrid Padilla, Oscar Hernan and Wang, Daren},
  journal={arXiv:2308.16172},
  year={2023}
}

@article{madrid2022risk,
  title={Risk bounds for quantile trend filtering},
  author={Madrid Padilla, Oscar Hernan and Chatterjee, Sabyasachi},
  journal={Biometrika},
  volume={109},
  number={3},
  pages={751--768},
  year={2022},
  publisher={Oxford University Press}
}

@article{wang2016trend,
  author  = {Yu-Xiang Wang and James Sharpnack and Alexander J. Smola and Ryan J. Tibshirani},
  title   = {Trend Filtering on Graphs},
  journal = {The Journal of Machine Learning Research},
  year    = {2016},
  volume  = {17},
  number  = {105},
  pages   = {1--41}
}

@article{mammen1997locally,
  title={Locally adaptive regression splines},
  author={Mammen, Enno and Van De Geer, Sara},
  journal={The Annals of Statistics},
  volume={25},
  number={1},
  pages={387--413},
  year={1997},
  publisher={Institute of Mathematical Statistics}
}

@article{steidl2006splines,
  title={Splines in higher order TV regularization},
  author={Steidl, Gabriele and Didas, Stephan and Neumann, Julia},
  journal={International journal of computer vision},
  volume={70},
  pages={241--255},
  year={2006},
  publisher={Springer}
}

@article{zhu2017nonasymptotic,
  title={Nonasymptotic Analysis of Semiparametric Regression Models with High-Dimensional Parametric Coefficients},
  author={Zhu, Ying},
  journal={The Annals of Statistics},
  volume={45},
  number={5},
  pages={2274--2298},
  year={2017}
}

@article{ma2016asymptotic,
  title={Asymptotic properties of lasso in high-dimensional partially linear models},
  author={Ma, Chi and Huang, Jian},
  journal={Science China Mathematics},
  volume={59},
  pages={769--788},
  year={2016},
  publisher={Springer}
}

@article{xie2009scad,
  title={SCAD-Penalized Regression in High-Dimensional Partially Linear Models},
  author={Xie, Huiliang and Huang, Jian},
  journal={The Annals of Statistics},
  volume={37},
  number={2},
  pages={673--696},
  year={2009},
  publisher={JSTOR}
}

@book{hardle2000partially,
  title={Partially Linear Models},
  author={H{\"a}rdle, Wolfgang and Liang, Hua and Gao, Jiti},
  year={2000},
  publisher={Springer Science \& Business Media}
}

@book{wahba1990spline,
  title={Spline models for observational data},
  author={Wahba, Grace},
  year={1990},
  publisher={Society for Industrial and Applied Mathematics}
}

@article{fan2001variable,
  title={Variable selection via nonconcave penalized likelihood and its oracle properties},
  author={Fan, Jianqing and Li, Runze},
  journal={Journal of the American Statistical Association},
  volume={96},
  number={456},
  pages={1348--1360},
  year={2001},
  publisher={Taylor \& Francis}
}

@article{tibshirani1996regression,
  title={Regression shrinkage and selection via the lasso},
  author={Tibshirani, Robert},
  journal={Journal of the Royal Statistical Society Series B: Statistical Methodology},
  volume={58},
  number={1},
  pages={267--288},
  year={1996},
  publisher={Oxford University Press}
}

@article{ye2010rate,
  title={Rate minimaxity of the Lasso and Dantzig selector for the $\ell_q$ loss in $\ell_r$ balls},
  author={Ye, Fei and Zhang, Cun-Hui},
  journal={The Journal of Machine Learning Research},
  volume={11},
  pages={3519--3540},
  year={2010},
  publisher={JMLR. org}
}

@article{verzelen2012minimax,
  title={Minimax risks for sparse regressions: Ultra-high dimensional phenomenons},
  author={Verzelen, Nicolas},
  journal={Electronic Journal of Statistics},
  volume={6},
  pages={38--90},
  year={2012}
}

@article{raskutti2011minimax,
  title={Minimax rates of estimation for high-dimensional linear regression over $\ell_q$-balls},
  author={Raskutti, Garvesh and Wainwright, Martin J and Yu, Bin},
  journal={IEEE transactions on information theory},
  volume={57},
  number={10},
  pages={6976--6994},
  year={2011},
  publisher={IEEE}
}

@article{yu2011semi,
  title={Semi-parametric regression: Efficiency gains from modeling the nonparametric part},
  author={Yu, Kyusang and Mammen, Enno and Park, Byeong Uk},
  journal={Bernoulli},
  volume={17},
  number={2},
  pages={736--748},
  year={2011}
}

@article{kim2009ell_1,
  title={$\ell_1$ trend filtering},
  author={Kim, Seung-Jean and Koh, Kwangmoo and Boyd, Stephen and Gorinevsky, Dimitry},
  journal={SIAM Review},
  volume={51},
  number={2},
  pages={339--360},
  year={2009},
  publisher={SIAM}
}

@article{tseng2001convergence,
  title={Convergence of a block coordinate descent method for nondifferentiable minimization},
  author={Tseng, Paul},
  journal={Journal of optimization theory and applications},
  volume={109},
  pages={475--494},
  year={2001},
  publisher={Springer}
}

@article{ramdas2016fast,
  title={Fast and flexible ADMM algorithms for trend filtering},
  author={Ramdas, Aaditya and Tibshirani, Ryan J},
  journal={Journal of Computational and Graphical Statistics},
  volume={25},
  number={3},
  pages={839--858},
  year={2016},
  publisher={Taylor \& Francis}
}

@Article{glmnet2010,
    title = {Regularization Paths for Generalized Linear Models via Coordinate Descent},
    author = {Jerome Friedman and Trevor Hastie and Robert Tibshirani},
    journal = {Journal of Statistical Software},
    year = {2010},
    volume = {33},
    number = {1},
    pages = {1--22},
  }

@article{tibshirani2014adaptive,
  title={Adaptive piecewise polynomial estimation via trend filtering},
  author={Tibshirani, Ryan J},
  journal={The Annals of Statistics},
  volume={42},
  number={1},
  pages={285--323},
  year={2014},
  publisher={Citeseer}
}

@article{sadhanala2019additive,
  title={Additive models with trend filtering},
  author={Sadhanala, Veeranjaneyulu and Tibshirani, Ryan J},
  journal={The Annals of Statistics},
  volume={47},
  number={6},
  pages={3032--3068},
  year={2019},
  publisher={JSTOR}
}

@article{tibshirani2022divided,
  title={Divided differences, falling factorials, and discrete splines: Another look at trend filtering and related problems},
  author={Tibshirani, Ryan J},
  journal={Foundations and Trends{\textregistered} in Machine Learning},
  volume={15},
  number={6},
  pages={694--846},
  year={2022},
  publisher={Now Publishers, Inc.}
}

@book{vershynin2018high,
  title={High-dimensional probability: An introduction with applications in data science},
  author={Vershynin, Roman},
  volume={47},
  year={2018},
  publisher={Cambridge university press}
}

@article{van2014uniform,
  title={On the uniform convergence of empirical norms and inner products, with application to causal inference},
  author={van de Geer, Sara},
  journal={Electronic Journal of Statistics},
  volume={8},
  pages={543--574},
  year={2014}
}

@book{van2016estimation,
  title={Estimation and testing under sparsity},
  author={Van de Geer, Sara},
  year={2016},
  publisher={Springer}
}

@inproceedings{wang2014falling,
  title={The falling factorial basis and its statistical applications},
  author={Wang, Yu-Xiang and Smola, Alex and Tibshirani, Ryan},
  booktitle={International Conference on Machine Learning},
  pages={730--738},
  year={2014},
  organization={PMLR}
}

@book{buhlmann2011statistics,
  title={Statistics for high-dimensional data: methods, theory and applications},
  author={B{\"u}hlmann, Peter and Van De Geer, Sara},
  year={2011},
  publisher={Springer Science \& Business Media}
}

@book{blm13,
    title={Concentration Inequalities: A Nonasymptotic Theory of Independence},
    author={Boucheron, St\'{e}phane and Lugosi, G\'{a}bor and Massart, Pascal},
    year = {2013},
    month = {02},
    publisher = {Oxford University Press}   
}

@article{muller2015partial,
  title={The partial linear model in high dimensions},
  author={M{\"u}ller, Patric and Van de Geer, Sara},
  journal={Scandinavian Journal of Statistics},
  volume={42},
  number={2},
  pages={580--608},
  year={2015},
  publisher={Wiley Online Library}
}

@article{yu2019minimax,
  title={Minimax optimal estimation in partially linear additive models under high dimension},
  author={Yu, Zhuqing and Levine, Michael and Cheng, Guang},
  journal={Bernoulli},
  volume={25},
  number={2},
  pages={1289--1325},
  year={2019}
}

@article{chen1988convergence,
  title={Convergence rates for parametric components in a partly linear model},
  author={Chen, Hung},
  journal={The Annals of Statistics},
volume={16},
  pages={136--146},
  year={1988},
  publisher={JSTOR}
}

@article{engle1986semiparametric,
  title={Semiparametric estimates of the relation between weather and electricity sales},
  author={Engle, Robert F and Granger, Clive WJ and Rice, John and Weiss, Andrew},
  journal={Journal of the American statistical Association},
  volume={81},
  number={394},
  pages={310--320},
  year={1986},
  publisher={Taylor \& Francis}
}

@article{bunea2004consistent,
  title={Consistent covariate selection and post model selection inference in semiparametric regression},
  author={Bunea, Florentina},
  journal={The Annals of Statistics},
  volume={32},
  number={1},
  pages={898--927},
  year={2004}
}

@article{mammen1997penalized,
  title={Penalized quasi-likelihood estimation in partial linear models},
  author={Mammen, Enno and van de Geer, Sara},
  journal={The Annals of Statistics},
  volume={25},
  number={3},
  pages={1014--1035},
  year={1997},
  publisher={Institute of Mathematical Statistics}
}

@article{lv2022debiased,
  title={Debiased distributed learning for sparse partial linear models in high dimensions},
  author={Lv, Shaogao and Lian, Heng},
  journal={The Journal of Machine Learning Research},
  volume={23},
  number={2},
  pages={1--32},
  year={2022}
}

@inproceedings{zhu2019high,
  title={High dimensional inference in partially linear models},
  author={Zhu, Ying and Yu, Zhuqing and Cheng, Guang},
  booktitle={The 22nd International Conference on Artificial Intelligence and Statistics},
  pages={2760--2769},
  year={2019},
  organization={PMLR}
}

@article{fu2024semiparametric,
  title={Semiparametric efficient estimation in high-dimensional partial linear regression models},
  author={Fu, Xinyu and Huang, Mian and Yao, Weixin},
volume={51},
pages={1259--1287},
  journal={Scandinavian Journal of Statistics},
  year={2024},
  publisher={Wiley Online Library}
}

@article{donoho1998minimax,
  title={Minimax estimation via wavelet shrinkage},
  author={Donoho, David L and Johnstone, Iain M},
  journal={The Annals of Statistics},
  volume={26},
  number={3},
  pages={879--921},
  year={1998},
  publisher={Institute of Mathematical Statistics}
}

@article{guntuboyina2020adaptive,
  title={Adaptive risk bounds in univariate total variation denoising and trend filtering},
  author={Guntuboyina, Adityanand and Lieu, Donovan and Chatterjee, Sabyasachi and Sen, Bodhisattva},
  journal={The Annals of Statistics},
  volume={48},
  number={1},
  pages={205--229},
  year={2020},
  publisher={JSTOR}
}

@article{ortelli2021prediction,
  title={Prediction bounds for higher order total variation regularized least squares},
  author={Ortelli, Francesco and van de Geer, Sara},
  journal={The Annals of Statistics},
  volume={49},
  number={5},
  pages={2755--2773},
  year={2021},
  publisher={Institute of Mathematical Statistics}
}

@article{wang2017generalized,
  title={Generalized scalar-on-image regression models via total variation},
  author={Wang, Xiao and Zhu, Hongtu and Alzheimer’s Disease Neuroimaging Initiative},
  journal={Journal of the American Statistical Association},
  volume={112},
  number={519},
  pages={1156--1168},
  year={2017},
  publisher={Taylor \& Francis}
}

@article{madrid2020adaptive,
  title={Adaptive nonparametric regression with the k-nearest neighbour fused lasso},
  author={Madrid Padilla, Oscar Hernan and Sharpnack, James and Chen, Yanzhen and Witten, Daniela M},
  journal={Biometrika},
  volume={107},
  number={2},
  pages={293--310},
  year={2020},
  publisher={Oxford University Press}
}

@article{petersen2019data,
  title={Data-adaptive additive modeling},
  author={Petersen, Ashley and Witten, Daniela},
  journal={Statistics in Medicine},
  volume={38},
  number={4},
  pages={583--600},
  year={2019},
  publisher={Wiley Online Library}
}

@article{wakayama2023trend,
  title={Trend filtering for functional data},
  author={Wakayama, Tomoya and Sugasawa, Shonosuke},
  journal={Stat},
  volume={12},
  number={1},
  pages={e590},
  year={2023},
  publisher={Wiley Online Library}
}

@article{rahardiantoro2024spatio,
  title={Spatio-temporal clustering analysis using generalized lasso with an application to reveal the spread of Covid-19 cases in Japan},
  author={Rahardiantoro, Septian and Sakamoto, Wataru},
  journal={Computational Statistics},
  volume={39},
  number={3},
  pages={1513--1537},
  year={2024},
  publisher={Springer}
}

@article{tan2019doubly,
  title={Doubly penalized estimation in additive regression with high-dimensional data},
  author={Tan, Zhiqiang and Zhang, Cun-Hui},
  journal={The Annals of Statistics},
  volume={47},
  number={5},
  pages={2567--2600},
  year={2019},
  publisher={JSTOR}
}

@article{stein1981estimation,
  title={Estimation of the mean of a multivariate normal distribution},
  author={Stein, Charles M},
  journal={The annals of Statistics},
volume={9},
  pages={1135--1151},
  year={1981},
  publisher={JSTOR}
}

@article{tibshirani2011solution,
  title={The solution path of the generalized lasso},
  author={Tibshirani, Ryan J and Taylor, Jonathan},
  journal={The Annals of Statistics},
volume={39},
  pages={1335--1371},
  year={2011},
  publisher={JSTOR}
}

@article{tibshirani2013lasso,
  title={The lasso problem and uniqueness},
  author={Tibshirani, Ryan J},
  journal={Electronic Journal of Statistics},
  volume={7},
  pages={1456--1490},
  year={2013}
}

@article{tsybakov2009simultaneous,
  title={Simultaneous analysis of Lasso and Dantzig selector},
  author={Tsybakov, AB and Bickel, PJ and Ritov, Y},
  journal={The Annals of Statistics},
  volume={37},
  number={4},
  pages={1705--1732},
  year={2009}
}

@article{zhang2011linear,
  title={Linear or nonlinear? Automatic structure discovery for partially linear models},
  author={Zhang, Hao Helen and Cheng, Guang and Liu, Yufeng},
  journal={Journal of the American Statistical Association},
  volume={106},
  number={495},
  pages={1099--1112},
  year={2011},
  publisher={Taylor \& Francis}
}

@article{huang2012semiparametric,
  title={Semiparametric regression pursuit},
  author={Huang, Jian and Wei, Fengrong and Ma, Shuangge},
  journal={Statistica Sinica},
  volume={22},
  number={4},
  pages={1403},
  year={2012},
  publisher={NIH Public Access}
}

@article{lian2015separation,
  title={Separation of covariates into nonparametric and parametric parts in high-dimensional partially linear additive models},
  author={Lian, Heng and Liang, Hua and Ruppert, David},
  journal={Statistica Sinica},
volume={25},
  pages={591--607},
  year={2015},
  publisher={JSTOR}
}

@article{dai2023false,
  title={False discovery rate control via data splitting},
  author={Dai, Chenguang and Lin, Buyu and Xing, Xin and Liu, Jun S},
  journal={Journal of the American Statistical Association},
  volume={118},
  number={544},
  pages={2503--2520},
  year={2023},
  publisher={Taylor \& Francis}
}

@article{candes2018panning,
  title={Panning for gold:‘model-X’knockoffs for high dimensional controlled variable selection},
  author={Candes, Emmanuel and Fan, Yingying and Janson, Lucas and Lv, Jinchi},
  journal={Journal of the Royal Statistical Society Series B: Statistical Methodology},
  volume={80},
  number={3},
  pages={551--577},
  year={2018},
  publisher={Oxford University Press}
}

@article{barber2015controlling,
  title={Controlling the false discovery rate via knockoffs},
  author={Barber, Rina Foygel and Cand{\`e}s, Emmanuel J},
  journal={The Annals of statistics},
volume={43},
  pages={2055--2085},
  year={2015},
  publisher={JSTOR}
}

@article{zhang2014confidence,
  title={Confidence intervals for low dimensional parameters in high dimensional linear models},
  author={Zhang, Cun-Hui and Zhang, Stephanie S},
  journal={Journal of the Royal Statistical Society: Series B: Statistical Methodology},
  volume={76},
  pages={217--242},
  year={2014},
  publisher={JSTOR}
}

@article{van2014asymptotically,
  title={On Asymptotically Optimal Confidence Regions And Tests For High-Dimensional Models},
  author={van de Geer, Sara and B{\"u}hlmann, Peter and Ritov, Ya'acov and Dezeure, Ruben},
  journal={The Annals of Statistics},
  volume={42},
  number={3},
  pages={1166--1202},
  year={2014},
  publisher={Institute of Mathematical Statistics}
}

@article{javanmard2014confidence,
  title={Confidence intervals and hypothesis testing for high-dimensional regression},
  author={Javanmard, Adel and Montanari, Andrea},
  journal={The Journal of Machine Learning Research},
  volume={15},
  number={1},
  pages={2869--2909},
  year={2014},
  publisher={JMLR. org}
}

@article{dezeure2015high,
  title={High-dimensional inference: confidence intervals, p-values and R-software hdi},
  author={Dezeure, Ruben and B{\"u}hlmann, Peter and Meier, Lukas and Meinshausen, Nicolai},
  journal={Statistical science},
volume={30},
  pages={533--558},
  year={2015},
  publisher={JSTOR}
}

%% or include bibliography directly:
%\begin{thebibliography}{9}

%\bibitem{r1}
%\textsc{Billingsley, P.} (1999). \textit{Convergence of
%Probability Measures}, 2nd ed.
%Wiley, New York.
%\MR{1700749}

%\bibitem{r2}
%\textsc{Bourbaki, N.}  (1966). \textit{General Topology}  \textbf{1}.
%Addison--Wesley, Reading, MA.

%\bibitem{r3}
%\textsc{Ethier, S. N.} and \textsc{Kurtz, T. G.} (1985).
%\textit{Markov Processes: Characterization and Convergence}.
%Wiley, New York.
%\MR{838085}

%\bibitem{r4}
%\textsc{Prokhorov, Yu.} (1956).
%Convergence of random processes and limit theorems in probability
%theory. \textit{Theory  Probab.  Appl.}
%\textbf{1} 157--214.
%\MR{84896}
%\end{thebibliography}

\end{document}